\RequirePackage{fix-cm}
\documentclass[twocolumn]{svjour3}          
\smartqed  
\usepackage{graphicx}
\usepackage{macros}
\usepackage{amsmath}
\usepackage{multirow}
\usepackage{paralist}
\usepackage{pifont}

\begin{document}

\title{Learning Control Lyapunov Functions from Counterexamples and Demonstrations
}


\author{Hadi Ravanbakhsh         \and
        Sriram Sankaranarayanan 
}


\institute{H. Ravanbakhsh \at
              University of Colorado, Boulder \\
              \email{hadi.ravanbakhsh@colorado.edu}           
           \and
           S. Sankaranarayanan \at
              University of Colorado, Boulder \\
              \email{sriram.sankaranarayanan@colorado.edu}           
}

\date{Received: date / Accepted: date}

\maketitle

\begin{abstract}
  We present a technique for learning control Lyapunov-like functions,
  which are used in turn to synthesize controllers for nonlinear
  dynamical systems that can stabilize the system, or satisfy
  specifications such as remaining inside a safe set, or eventually
  reaching a target set while remaining inside a safe set.  The
  learning framework uses a \emph{demonstrator} that implements a
  black-box, untrusted strategy presumed to solve the problem of
  interest, a \emph{learner} that poses finitely many queries to the
  demonstrator to infer a candidate function, and a \emph{verifier}
  that checks whether the current candidate is a valid control
  Lyapunov function. The overall learning framework is iterative,
  eliminating a set of candidates on each iteration using the
  counterexamples discovered by the verifier and the demonstrations
  over these counterexamples. We prove its convergence using
  ellipsoidal approximation techniques from convex optimization. We
  also implement this scheme using nonlinear MPC controllers to serve
  as demonstrators for a set of state and trajectory stabilization
  problems for nonlinear dynamical systems. We show how the
    verifier can be constructed efficiently using convex relaxations
    of the verification problem for polynomial systems to
    semi-definite programming (SDP) problem instances. Our approach
  is able to synthesize relatively simple polynomial control Lyapunov
  functions, and in that process replace the MPC using a guaranteed
  and computationally less expensive controller.
\end{abstract}

\section{Introduction} \label{sec:intro}
We propose a novel \emph{learning from demonstration} scheme for
inferring control Lyapunov functions (potential functions) for
stabilizing nonlinear dynamical systems to reference states/trajectories, 
and implementing control laws for specifications
  such as maintaining a system inside a set of safe states, reaching a
  target set while remaining inside a safe set and tracking a given
  trajectory while not deviating too far away.  Control Lyapunov
functions (CLFs) have wide applications to autonomous systems~\cite{jadbabaie2002control,galloway2015torque,ames2013towards,nguyen2015optimal,KHANSARIZADEH2014}. They extend the classic notion of Lyapunov functions to
systems involving control
inputs~\cite{Sontag/1982/Characterization,sontag1983lyapunov,artstein1983stabilization}. Finding a CLF also leads us
to an associated feedback control law that can be used to solve the
stabilization problem. Additionally, they can be extended for feedback
motion planning using extensions to time-varying or sequential
CLFs~\cite{burridge1999sequential,tedrake2010lqr}. Likewise, they have
been investigated in the robotics community in many forms including
\emph{artificial potential functions} to solve path planning problems
involving obstacles~\cite{lopez1995autonomous}.

However, synthesizing CLFs for nonlinear systems remains a
challenge~\cite{primbs1999nonlinear}. Standard
approaches to finding CLFs include the use of dynamic programming,
wherein the value function satisfies the conditions of a 
CLF~\cite{bertsekas1995dynamic}, or using non-convex bilinear matrix
inequalities (BMI)~\cite{henrion2005solving}. 

In this article, we investigate the problem of learning a CLF using a
black-box \textsc{Demonstrator} that implements an unknown state
feedback law to stabilize the system to a given equilibrium.  This
\textsc{Demonstrator} can be queried at a given system state, and returns a \emph{demonstration} in the form of a  control input generated
at that state by its feedback law. Such a \textsc{Demonstrator} 
can be realized
using an expensive nonlinear model predictive controller (MPC) that
uses a local optimization scheme, or even a human operator under
certain assumptions~\footnote{However, we do not handle noisy or
erroneous demonstrators in this paper.}.  Additionally, the framework
has a \textsc{Learner} which selects a candidate CLF and
a \textsc{Verifier} that tests whether this CLF is valid. If the CLF
is invalid, the \textsc{Verifier} returns a state at which the current
candidate fails. The \textsc{Learner} queries the \textsc{Demonstrator} to
obtain a control input corresponding to this state.  It subsequently
eliminates the current candidate along with a set of related functions
from further consideration. The framework continues to exhaust the
space of candidate CLFs until no CLFs remain or a valid CLF is found
in this process. We prove the process can converge in finitely many
steps provided the
\textsc{Learner} chooses the candidate function appropriately at each
step. We also provide efficient SDP-based approximations to the
verification problem that can be used to drive the framework. Finally,
we test this approach on a variety of examples, by solving
stabilization problems for nonlinear dynamical systems. We show that
our approach can successfully find CLFs using finite horizon
nonlinear MPC schemes with appropriately chosen cost functions to
serve as demonstrators. In these instances, the CLFs yield control
laws that are computationally inexpensive, and  guaranteed against
the original dynamical model.

This paper is an extended version of our earlier 
work~\cite{Ravanbakhsh-RSS-17}. When compared to the
  earlier work, we have thoroughly expanded the technical sections to
  provide detailed proofs of the various results and a detailed
  exposition of each component of our learning
  framework. Additionally, we have included a new section that
  discusses specifications other than stability properties. We have
  also extended our experimental results and compare different options
  for implementing the overall learning loop as well as comparisons with
  other methods. We also provide a
  detailed discussion of various extensions to the approach presented
  in this paper.

\subsection{Illustrative Example: TORA System}
\begin{figure*}[t]
\begin{center}
\includegraphics[width=0.95\textwidth]{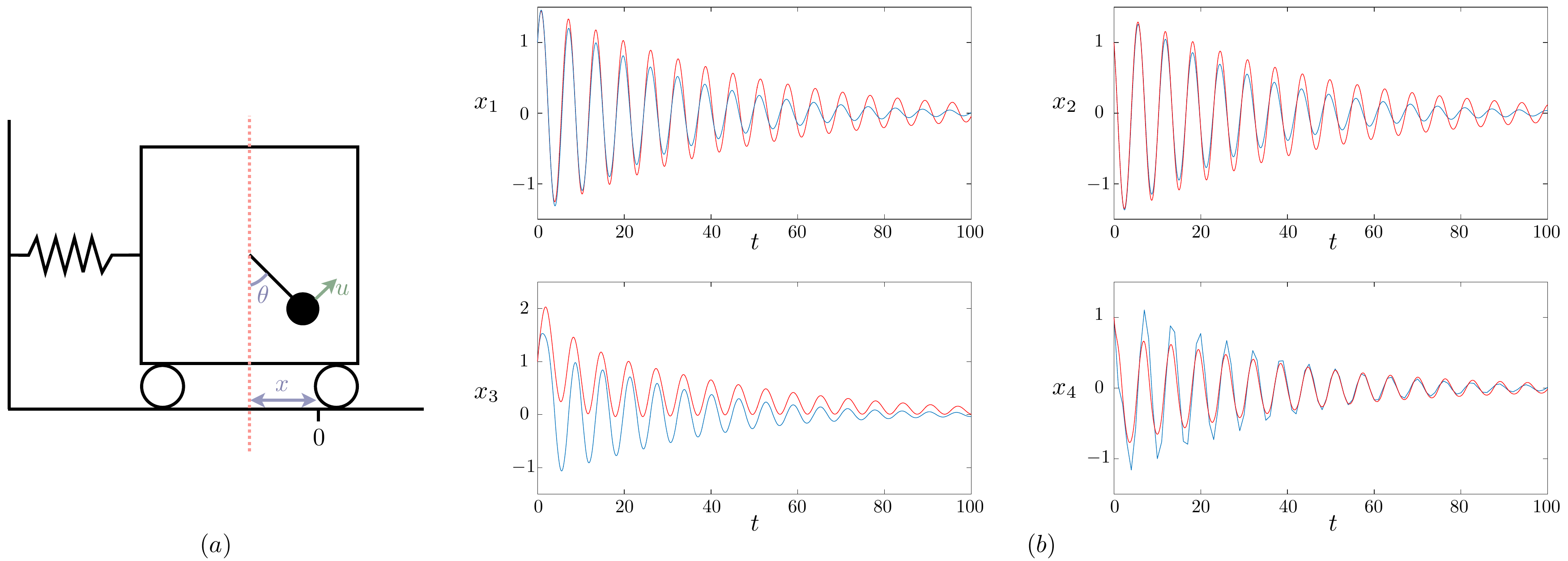}
\end{center}
\caption{TORA System. (a) A schematic diagram of the TORA 
system. (b) Execution traces of the system using MPC control 
(blue traces) and Lyapunov based control (red traces) starting 
from same initial point.}
\label{Fig:tora-example}
\end{figure*}

Figure~\ref{Fig:tora-example}(a) shows a mechanical system, called
translational oscillations with a rotational actuator (TORA).
The system consists of
a cart attached to a wall using a spring. Inside the cart, there is
an arm with a weight which can rotate. The cart itself can oscillate freely
and there are no friction forces. The system has two degrees of freedom,
including the position of the cart $x$, and the rotational position of 
the arm $\theta$. The controller can rotate the arm through input $u$.
The goal is to stabilize the cart
to $x=0$, with its velocity, angle, and angular velocity
$\dot{x} = \theta=\dot{\theta} = 0$. We refer the reader to Jankovic et
al.~\cite{jankovic1996tora} for a derivation of the dynamics, shown below
in terms of state variables $(x_1,\ldots,x_4)$, collectively written
as a vector $\vx$, and a single control input $(u_1)$, written as a vector $\vu$, after a basis transformation:
\begin{equation}\label{eq:tora-dyn}
		\dot{x_1} = x_2,\, \dot{x_2} = -x_1 + \epsilon \sin(x_3),\, \dot{x_3} = x_4,\, \dot{x_4} = u_1\,.
\end{equation}
$\sin(x_3)$ is approximated using a degree three polynomial
approximation which is quite accurate over the range $x_3 \in [-2,2]$.
The equilibrium $x = \dot{x} = \theta = \dot{\theta} = 0$ now
corresponds to $x_1 = x_2 = x_3 = x_4 = 0$.  The system has a single
control input $u_1$ that is bounded $u_1 \in [-1.5, 1.5]$.  Further,
we define a ``safe set''
$S: [-1,1] \times [-1,1] \times [-2,2] \times [-1,1]$, so that if
$\vx(0) \in S$ then $\vx(t) \in S$ for all time $t \geq 0$.

\paragraph{MPC Scheme:} A first approach to solve the problem  uses 
a nonlinear model-predictive control (MPC) scheme using
a discretization of the system dynamics with time step $\tau = 1$.
The time $t$ belongs to set $\{0, \tau, 2\tau,\ldots,N\tau = \T\}$ and:
\begin{equation}\label{eq:discretized-dynamics}
  \vx(t + \tau) =  \vx(t) + \tau f(\vx(t), \vu(t)) \,,
\end{equation}
with $f(\vx,\vu)$ representing the vector field of the ODE in
~\eqref{eq:tora-dyn}. Fixing the time horizon $\T = 30$, we use a
simple cost function
$J(\vx(0), \vu(0), \vu(\tau), \ldots, \vu(\T-\tau)\})$:
\begin{equation}\label{eq:mpc-formulation}
  \sum_{t \in \{0,\tau,...,\T-\tau\}} \left(||\vx(t)||_2^2 +
  ||\vu(t)||_2^2\right) + N \ ||\vx(\T)||_2^2 \,.
\end{equation}
Here, we
constrain $\vu(t) \in [-1.5, 1.5]$ for all
$t$ and define $\vx(t+\tau)$ in terms of
$\vx(t)$ using the discretization in~\eqref{eq:discretized-dynamics}.
Such a control is implemented using a first/second order numerical
gradient descent method to minimize the cost
function~\cite{Nocedal+Wright/2006/Numerical}. The stabilization of
the system was informally confirmed through hundreds of simulations
from different initial states. However, the MPC scheme is expensive,
requiring repeated solutions to (constrained) nonlinear optimization
problems in real-time. Furthermore, in general, the closed loop lacks
formal guarantees despite the \emph{high confidence} gained from
numerous simulations.

\begin{figure}[t]
\begin{center}
\begin{tikzpicture}
\matrix[every node/.style={rectangle, draw=black, line width=1.5pt}, row sep=20pt, column sep=15pt]{ & \node[fill=blue!20](n0){\begin{tabular}{c} 
\textsc{Learner} \end{tabular}}; & \\
\node[fill=green!20](n1){\begin{tabular}{c}
\textsc{Verifier} \end{tabular}}; & & \node[fill=red!20](n2){\begin{tabular}{c}
\textsc{Demonstrator} \end{tabular} }; \\
};
\path[->, line width=2pt] (n0) edge[bend right] node[left]{$V(\vx)?$} (n1)
(n1) edge[bend right] node[below]{\;\;\;\begin{tabular}{c}
Yes or \\
No($\vx_{j+1}$)\end{tabular}} (n0)
(n0) edge [bend left] node[above]{$\vx_j$} (n2)
(n2) edge [bend left] node[above]{$\vu_j$} (n0);
\draw (n0.north)+(0,0.3cm) node {$(\vx_1, \vu_1), \ldots, (\vx_j, \vu_j)$};
\end{tikzpicture}
\end{center}
\caption{Overview of the learning framework for learning a control Lyapunov function.}\label{fig:learning-framework}
\end{figure}
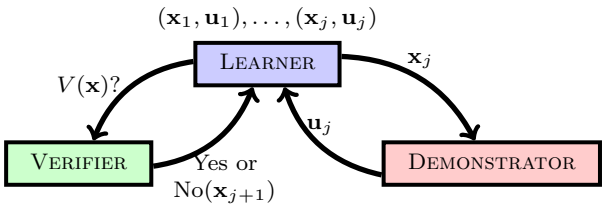

\paragraph{Learning a Control Lyapunov Function:} In this article, we introduce
an approach which uses the MPC scheme as a \textsc{demonstrator},
and attempts to learn a control Lyapunov function. Then, a 
control law (in a closed form) is obtained from the CLF. The overall
idea, depicted in Fig.~\ref{fig:learning-framework}, is to pose
queries to the \emph{offline} MPC at finitely many \emph{witness}
states $\{ \vx^{(1)}, \ldots, \vx^{(j)} \}$. Then, for each witness state
$\vx^{(i)}$, the MPC is applied to generate a sequence of control
inputs $\vu^{(i)}(0), \vu^{(i)}(\tau), \cdots, $ $\vu^{(i)}(\T-\tau)$ with $\vx^{(i)}$ as the initial
state, in order to drive the system into the equilibrium starting from
$\vx^{(i)}$. The MPC then retains the first control input
$\vu^{(i)}:\ \vu^{(i)}(0)$, and discards the remaining (as is standard in
MPC). This yields the so called observation pairs $(\vx^{(i)}, \vu^{(i)})$
that are used by the \textsc{learner}.

The
\textsc{learner} attempts to find a candidate function
$V(\vx)$ that is positive definite and which decreases at each witness
state $\vx^{(i)}$ through the control input $\vu^{(i)}$. This function $V$
is potentially a CLF function for the system.
This function is fed to the \textsc{verifier}, which checks whether
$V(\vx)$ is indeed a CLF, or discovers a state $\vx^{(j+1)}$ which
refutes $V$. This new state is added to the witness set and the
process is iterated. The procedure described in this paper synthesizes
the control Lyapunov function $V(\vx)$ below:
 \begin{align*}
 V =& 1.22 x_2^2 + 0.31 x_2x_3 + 0.44 x_3^2 - 0.28 x_4x_2\\
    & + 0.80 x_4x_3 + 1.69 x_4^2 + 0.07 x_1x_2 - 0.66 x_1x_3\\
    & - 1.85 x_4x_1 + 1.6 x_1^2\,.
 \end{align*}

Next, this function is used to design a associated control law that
guarantees the stabilization of the model described in Eq.~\eqref{eq:tora-dyn}.
Figure~\ref{Fig:tora-example}(b)
shows a closed loop trajectory for this control law vs control law
extracted by the MPC.  At each step, given a
current state $\vx$, we compute an input $\vu \in [-1.5, 1.5]$ such
that:
\begin{equation}\label{eq:clf-decrease}
  (\nabla V) \cdot f(\vx, \vu) < 0 \,.
\end{equation}
 First, the definition of a CLF guarantees that any 
 state $\vx \in S$, a control input $\vu \in [-1.5,1.5]$ that
satisfies Eq.~\eqref{eq:clf-decrease} exists.
Such a $\vu$ may be chosen directly by means of a formula involving
$\vx$~\cite{LIN1991UNIVERSAL,suarez2001global} unlike the MPC which
solves a nonlinear problem in Eq.~\eqref{eq:mpc-formulation}. Furthermore, the
resulting  control law guarantees  the stability of the resulting closed loop.

\section{Background}\label{sec:background}
We recall preliminary notions, including the stabilization problem for
nonlinear dynamical systems.

\subsection{Problem Statement}
We will first define the system model studied throughout this paper.
\begin{definition}[Control System]
$\ $ A state feedback control system $\Psi(X, U, f, \K)$ consists
of a plant, a controller over 
$X \subseteq \reals^n$ and $U \subseteq \reals^m$.
\begin{compactenum}
\item $X \subseteq \reals^n$ is the \emph{state space} of the system.
   The control inputs belong to a set
$U$ defined as a polyhedron:
\begin{equation}\label{eq:input-sat}
	U = \{\vu \ | \ A \vu \geq \vb \} \,.
\end{equation}
\item The plant consists of a vector field defined by a
  continuous and differentiable function $f: X \times U \mapsto \reals^n$.
\item The controller measures the state of the plant $\vx \in X$ and
provides feedback $\vu \in U$. The controller is defined by a 
feedback function $\K:X \mapsto U$
(Fig.~\ref{fig:system}).
\end{compactenum}
\end{definition}

For now, we assume $\K$ is a smooth (continuous and differentiable) function.
For a given feedback law $\K$, an
execution trace of the system, starting from an initial
state $\vx_0$ is  a function: 
$\vx: [0,T(\vx_0)) \mapsto X $, which maps time $t \in [0,T(\vx_0))$ to a state $\vx(t)$,
    such that
\[ \dot{\vx}(t) = f(\vx(t), \K(\vx(t))) \,, \] where $\dot{\vx}(\cdot)$ is
the right derivative of $\vx(\cdot)$ w.r.t. time over $[0, T(\vx_0))$.
Since $f$ and $\K$ are assumed to be smooth,  there exists a unique trajectory
for any $\vx_0$,  defined over some time interval $ [0, T(\vx_0))$.
Here $T(\vx_0)$ is $\infty$ if
trajectory starting from $\vx_0$ exists for all time. Otherwise, $T(\vx_0)$ is finite if the
trajectory ``escapes'' in finite time. For most of the systems we study, the closed loop dynamics
are such that a compact set $S$ will be positive invariant. In fact, this set will be a 
sublevel set of a Lyapunov function for the closed loop dynamics. This fact along
with the smoothness of $f, \K$ suffices to establish that $T(\vx_0) = \infty$ for all $\vx_0 \in S$. Unless otherwise noted,
we will consider control laws $\K$ that will guarantee existence of trajectories for all time.

\begin{figure}[!t]
\begin{center}
	\includegraphics[width=0.25\textwidth]{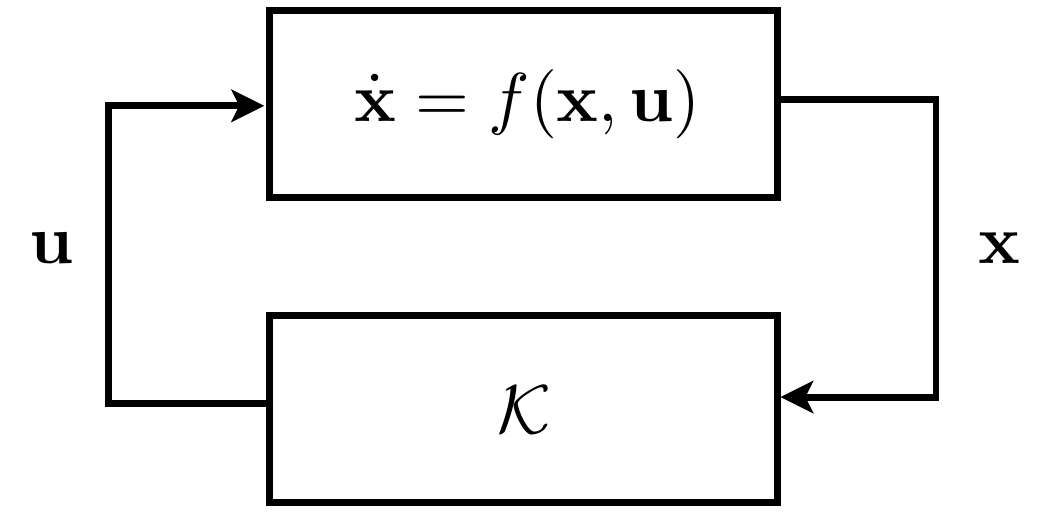}
\end{center}
\caption{Closed-loop state feedback system.}\label{fig:system} 
\end{figure}

A specification describes the desired behavior of all possible
execution traces $\vx(\cdot)$.
In this article, we study a variety of
specifications, including stability, trajectory tracking, and
safety. For simplicity, we first focus on stability. Extensions to
other specifications are presented in Section~\ref{sec:spec}. Also,
without loss of generality, we assume $\vx = \vzero$ is the desired
equilibrium. Moreover, $f(\vzero, \vzero) = \vzero$.

\begin{problem}[Synthesis for Asymptotic
  Stability] \label{prob:stability} Given a plant, the
  control synthesis problem is to design a controller (a feedback law $\K$) s.t.
  all traces $\vx(\cdot)$ of the closed loop system
  $\Psi(X, U, f, \K)$ are asymptotically stable. We require two properties for asymptotic stability.
First, the system is \emph{Lyapunov stable}:
	\[ \begin{array}{ll}
            (\forall \epsilon > 0) \\
          \ \ \; (\exists \delta > 0) \\
          \ \ \ \ \ \ \; \left(\begin{array}{c}\forall \vx(\cdot) \\ \vx(0) \in B_{\delta}(\vzero)\end{array}\right)\ (\forall t \geq 0)  \; \vx(t) \in \B_\epsilon(\vzero) \,,\\
           \end{array}\]
         wherein $\B_\delta(\vx) \subseteq \reals^n$  is the ball of radius 
         $\delta$ centered at $\vx$. In other words, for any
         chosen $\epsilon > 0$,  we may ensure that
         the trajectories will stay inside a ball of
         $\epsilon$ radius by choosing the initial
         conditions to lie inside a ball of $\delta$ radius.

    Furthermore, all the trajectories converge asymptotically towards the origin:
    \[\begin{array}{ll}
      (\forall \epsilon > 0)\ \left(\forall \vx(\cdot)\right)\ (\exists T > 0)\ (\forall t \geq T) \ \vx(t) \in \B_{\epsilon}(\vzero) \,.
    \end{array}\]
    I.e., For any chosen $\epsilon > 0$, all trajectories will eventually 
    reach a ball of radius $\epsilon$ around the origin and stay inside forever.
\end{problem}

Stability in our method is addressed through Lyapunov analysis. More
specifically, our solution is based on control Lyapunov functions
(CLF). CLFs were first introduced by Sontag~\cite{Sontag/1982/Characterization,sontag1983lyapunov}, and
studied at the same time by Artstein~\cite{artstein1983stabilization}. Sontag's work shows that if a system is asymptotically stablizable, then there exists a CLF even if the dynamics are not smooth~\cite{sontag1983lyapunov}. Now, let us recall the definition of a positive and
negative definite functions.
\begin{definition}[Positive Definite]
  A function $V: \reals^n$ $\mapsto \reals$ is \emph{positive
    definite} over a set $X$ containing $\vzero$, iff $V(\vzero) = 0$
  and $V(\vx) > 0$ for all $\vx \in X \setminus \{ \vzero\}$.

  Likewise, $V$ is \emph{negative definite} iff $-V$ is positive definite.
  
\end{definition}

\begin{definition}[Control Lyapunov Function(CLF)]
	A smooth, radially unbounded function $V$ is a control Lyapunov
function (CLF) over $X$, if the following 
conditions hold~\cite{artstein1983stabilization}:
\begin{equation}\label{eq:clf-def} 
\begin{array}{l}
	V\ \mbox{is positive definite over}\ X \\
	\min_{\vu \in U} (\nabla V) \cdot f(\vx, \vu)\ \mbox{is negative definite over} X\,,
\end{array}
\end{equation}
where $\nabla V$ is the gradient of $V$. Note that $(\nabla V) \cdot f$ is
the Lie derivative of $V$ according to the vector field $f$.
\end{definition}
 Another
way of interpreting the second condition is that for each $\vx \in X$,
a control $\vu \in U$ can be chosen to ensure an \emph{instantaneous
  decrease} in the value of $V$, as illustrated in Fig.~\ref{fig:clf}.
  
\paragraph{Solving Stabilization using CLFs:}
Finding a CLF $V$ guarantees the existence of a feedback law
that can stabilize all trajectories to the equilibrium~\cite{artstein1983stabilization}. However, constructing such a feedback law is not trivial and potentially expensive.
Further results can be obtained by restricting the vector field
$f$ to be control affine:
\begin{equation} \label{eq:control-affine}
	f(\vx, \vu):\ f_0(\vx) + \sum_{i=1}^m f_i(\vx)u_i \,,
\end{equation}
wherein $f_i : X \mapsto \reals[X]^n$.
Assuming $U: \reals^m$, Sontag provides a method for extracting a feedback law $\K$, for control affine systems from a control Lyapunov
function~\cite{sontag1989universal}. More specifically, if a CLF $V$ is available, the following feedback law stabilizes the system:
\begin{equation}\label{eq:sontag}
	\K_i(\vx) = \begin{cases} 0 & \beta(\vx) = 0 \\
-b_i(\vx) \frac{a(\vx) + \sqrt{a(\vx)^2 + \beta(\vx)^2}}{\beta(\vx)} & \beta(\vx) \neq 0\,,
	\end{cases}
\end{equation}
where $a(\vx) = \nabla V.f_0(\vx)$, $b_i(\vx) = \nabla V.f_i(\vx)$, and $\beta(\vx) = \sum_{i=1}^m b_i^2(\vx)$.
\begin{remark}
	Feedback law $\K$ provided by the Sontag formula is not necessarily continuous at the origin. Nevertheless, such a feedback law still guarantees stabilization. See~\cite{sontag1989universal} for more details.
\end{remark}
Sontag formula can be extended to systems
with saturated inputs where $U$ is an n-ball~\cite{LIN1991UNIVERSAL} or a polytope~\cite{suarez2001global}. Also switching-based feedback is possible, under some mild assumptions (to avoid Zeno behavior)~\cite{curtis2003clf,Ravanbakhsh-Others/2015/Counter-LMI}. 
We assume dynamics are affine in control and use these results which reduce Problem~\ref{prob:stability} to that of finding a control Lyapunov function $V$.

\begin{figure}[!t]
\begin{center}
	\includegraphics[width=0.3\textwidth]{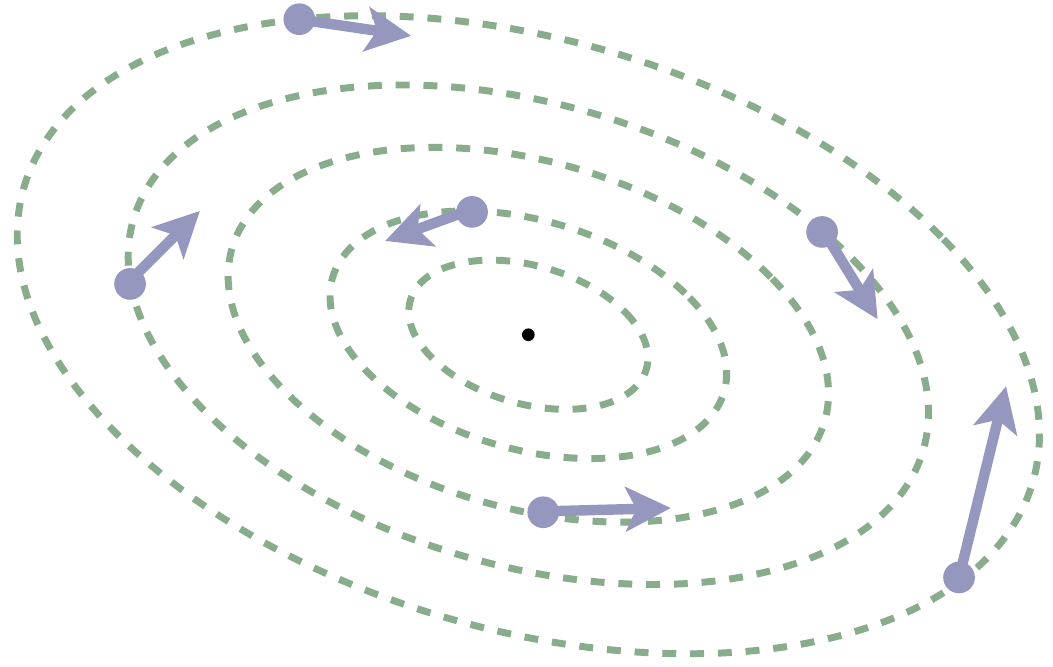}
\end{center}
\caption{Control Lyapunov Function (CLF): Level-sets of a CLF
  $V$ are shown using the green lines. For each state (blue dot), the
  vector field $f(\vx, \vu)$ for $\vu = \K(\vx)$ is the blue arrow,
  and it points to a direction which decreases $V$.}\label{fig:clf}
\end{figure}

\subsection{Discovering CLFs}
We briefly summarize approaches for discovering CLFs for a given plant
model in order to stabilize it to a given equilibrium state.
Efficient methods for discovering CLFs are available only for specific
classes of systems such as feedback linearizable systems, or for
so-called strict feedback systems, wherein a procedure called
\emph{backstepping} can be used~\cite{freeman2008robust}. However, finding CLFs for general nonlinear
systems is challenging~\cite{primbs1999nonlinear}.

One class of solutions uses optimal control theory by setting up the
problem of stabilization as one of minimizing a cost function over the
trajectories of the system.  If the cost function is set up
appropriately, then the value function for the resulting dynamic
programming problem is a a CLF~\cite{primbs1999nonlinear,bertsekas1995dynamic}. To do so,
however, one needs to solve a Hamilton-Jacobi-Bellman (HJB) partial
differential equation to discover the value function, which can be
quite hard in practice\cite{bryson1975applied}. In fact, rather than
solve HJB equations to obtain CLFs, it is more common to derive a CLF
using a procedure such as backstepping and apply inverse optimality
results to derive cost functions~\cite{freeman2008robust}.

A second class of solution is based on parameterization. More
specifically, a class of function $V_{\vc}(\vx)$ is parameterized by a
set of unknown parameters $\vc$. This parameterization is commonly
specified as a linear combination of basis functions of the
form $V_{\vc}(\vx) : \sum c_i g_i(\vx)$. Furthermore, the functions
$g_i$ commonly range over all possible monomials up to some prespecified
degree limit $D$. Next, an instantiation of
the parameters $\vc$ is discovered so that the resulting function $V$
is a CLF.  Unfortunately, discovering such parameters requires the
solution to a quantifier elimination problem, in general. This is quite
computationally expensive for nonlinear systems. Previously, authors
proposed a framework which uses sampling to avoiding expensive
quantifier eliminations~\cite{ravanbakhsh2015counterexample}. Despite
the use of sampling, scalability remains an issue. Another solution is
based on sum-of-squares
relaxations~\cite{Shor/1987/Class,lasserre2001global,Parillo/2003/Semidefinite},
along the lines of approaches used to discover Lyapunov
functions~\cite{Papachristodoulou+Prajna/2002/Construction}. However,
discovering CLFs using this approach entails solving a system of
bilinear matrix
inequalities~\cite{tan2004searching,henrion2005solving}.  In contrast
to LMIs, the set of solutions to a BMIs form a nonconvex set, and
solving BMIs is well-known to be computationally expensive, in
practice.  Rather than solving a BMI to find a CLF, and then
extracting the feedback law from the CLF, an alternative approach is
to simultaneously search for a Lyapunov function $V$ and an unknown
feedback law at the same
time~\cite{el1994synthesis,tan2004searching,majumdar2013control}. The latter approach also yields bilinear matrix inequalities of
comparable sizes.  Rather than seek algorithms that are
guaranteed to solve BMIs, a simpler approach is to attempt to solve the
BMIs using \emph{alternating minimization}: a form of coordinate
descent that fixes one set of variables in BMI, obtaining an LMI over
the remaining variables.  However, these approaches usually stuck in a
local ``saddle point'', and fail as a
result~\cite{Helton+Merino/1997/Coordinate}.

Approaches that parameterize a family of functions $V_{\vc}(\vx)$ face
the issue of choosing a family such that a CLF belonging to that family
is known to exist whenever the system is asymptotically stabilizable in the
first place. There is a rich literature on the existence of CLFs for a given
class of plant models. As mentioned earlier, if a system is
\emph{asymptotically stablizable}, then there exists a CLF even if the
dynamics are not smooth~\cite{sontag1983lyapunov}. However, the CLF does not
have to be smooth. Recent results, have shown some
light on the existence of polynomial Lyapunov functions for certain
classes of systems. Peet showed that an exponentially stable system
has a polynomial local Lyapunov function over a bounded region~\cite{Peet/2009/Exponentially}. Thus,
if there exists some feedback law that exponentially stabilizes a
given plant, we may conclude the existence of a polynomial
CLF for that system. This was recently extended to rationally stable systems
i.e., the distance to equilibrium decays as $o(t^{-k})$ for trajectories
starting from some set $\Omega$, by Leth et al.~\cite{Leth+Others/2017/Existence}. These results
do not guarantee that a search for a polynomial CLF will be successful
due to the lack of a bound on the degree $D$. This can be addressed by increasing
the degree of the monomials until a CLF is found, but the process can be prohibitively
expensive. Another drawback is that most approaches use SOS relaxations over polynomial systems to check the CLF conditions, although there is no guarantee as yet that polynomial CLFs that are also verifiable through SOS relaxations exist.

Another class of solutions involves approximate dynamic programming to
find approximations to value
functions~\cite{bertsekas2008approximate}. The solutions obtained
through these approaches are not guaranteed to be CLFs and thus may
need to be discarded, if the final result does not satisfy the
conditions for a CLF.  Approximate solutions are also investigated
through learning from demonstrations \cite{zhong2013value}.
Khansari-Zadeh et al. learn a CLF from
demonstrations through a combination of sampling states and corresponding
feedback provided by the demonstrator. A likely CLF is learned
through parameterizing a class of functions $V_{\vc}(\vx)$, and finding
conditions on $\vc$ by enforcing the conditions for the
CLFs at the sampled states~\cite{KHANSARIZADEH2014}. The conditions
for being a CLF should be checked on the solution obtained by solving these
constraints.

Compared to the techniques described above, the approach presented in
this paper is based on parameterization by choosing a class of
functions $V_{\vc}(\vx)$ and attempting to find a suitable $\vc \in C$
so that the result is a CLF.  Our approach avoids having to solve BMIs
by instead choosing finitely many sample states, and using
demonstrator's feedback to provide corresponding sample controls for
the state samples. However, instead of choosing these samples at
random, we use a verifier to select samples. Furthermore, our approach
can also systematically explore the space of possible parameters $C$
in a manner that guarantees termination in number of iterations
polynomial in the dimensionality of $C$ and $\vx$.  The result upon
termination can be a guaranteed CLF $V$ or failure to find a CLF among
the class of functions provided.

\section{Formal Learning Framework}\label{sec:framework}
As mentioned earlier, finding a control Lyapunov function is 
computationally expensive, requiring the solution
to BMIs~\cite{tan2004searching} or hard non-linear 
constraints~\cite{Ravanbakhsh-Others/2015/Counter-LMI}. 
The goal is to search for a solution (CLF) over a 
hypothesis space. More specifically, a CLF
is parameterized by a set of unknown parameters 
$\vc \in C$ ($C \subseteq \reals^r$). The parameterized CLF
is shown by $V_\vc$. And the goal is to
find $\vc \in C$ s.t. 
\begin{equation}\label{eq:clf-param-def} 
\begin{array}{l}
      V_\vc\ \mbox{is positive definite}\\
       \min_{\vu \in U} \nabla V_\vc.f(\vx, \vu) \ \mbox{is negative definite} \,.\\
\end{array}
\end{equation}

A standard approach is to choose a 
set of basis functions $g_1, \ldots, g_r$ ($g_i: X \mapsto \reals$)
 and search for a function of the form
\begin{equation} \label{eq:clf-template}
	V_\vc(\vx) = \sum_{j=1}^r c_j g_j(\vx) \,.
\end{equation}
\begin{remark}
The basis functions are chosen s.t. $V_\vc$ is radially 
unbounded and smooth, independent of the coefficients.
\end{remark}

As mentioned earlier, the learning framework has three components: a
demonstrator, a learner, and a verifier (see
Fig.~\ref{fig:learning-framework}).
The demonstrator inputs a state $\vx$ and returns a control input
$\vu \in U$, that is an appropriate ``instantaneous'' feedback for
$\vx$.  Formally, demonstrator is a function $\D : X \mapsto U$.

\begin{remark}[Demonstrator]
  The demonstrator is treated as a black box. This allows to use a
  variety of approaches ranging from trajectory
  optimization~\cite{zhang2016learning}, human
  expert demonstrations~\cite{KHANSARIZADEH2014}, and sample-based
  methods~\cite{lavalle2000rapidly,kocsis2006bandit}, which can be probabilistically complete.  While the
  demonstrator is \emph{presumed} to stabilize the system, our method
  can work even if the demonstrator is faulty. Specifically, a faulty
  demonstrator in worst case, may cause our method to
    terminate without having found a CLF. However,
    if a CLF is found by our approach, it is guaranteed to be correct.
\end{remark}

The formal learning procedure receives inputs:
\begin{enumerate}
	\item A plant described by $f$
	\item A ``black-box" demonstrator
function $\D: X \mapsto U$
	\item A set of basis functions $g_1,\ldots,g_r$ to form the 
hypothesis space $V_\vc(\vx) : \sum_{j=1}^r c_j g_j(\vx)$,
\end{enumerate} 
and either (a)  outputs a $\vc \in C$ s.t. 
$V_\vc(\vx): \vc^t \cdot \vg(\vx)$ is a CLF (Eq.~\eqref{eq:clf-param-def});
or (b)  declares \textsc{Failure:} no CLF could be discovered.

The goal of this framework is to find a CLF from a finite set of
queries to a demonstrator.

\begin{definition}[Observations]
We define a set of observations $O$ as
\[
O : \{(\vx_1, \vu_1),\ldots,(\vx_j, \vu_j)\} \subset X \times U\,,
\]
where $\vu_i$ is the demonstrated feedback for state $\vx_i$, i.e., $\vu_i: \D(\vx_i)$.
Further, we will assume that $\vx_i \not= \vzero$.
\end{definition}

\begin{definition}[Observation Compatibility] \label{def:compatible-data}
	A function $V$ is said to be compatible with a set of 
	observations $O$ iff $V$ respects the CLF conditions 
	(Eq.~\eqref{eq:clf-def}) for every observation in $O$:
	\[
	V(\vzero) = 0 \ \wedge \ \bigwedge\limits_{(\vx_i, \vu_i) \in O_j}
\left(\begin{array}{c} V(\vx_i) > 0\ \land\ \\ \nabla V \cdot f(\vx_i, \vu_i) < 0 \end{array}\right)\,.
	\]
\end{definition}

We note that observation compatible functions need not
  necessarily be a CLF, since they may violate the CLF condition for
  some state $\vx$ that is not part of an observation in $O$.  On the
  flip side, not every CLF (satisfying the conditions in Eq.~\eqref{eq:clf-def}) will necessarily be compatible with a given
  observation set $O$.

\begin{definition}[ Demonstrator Compatibility ] \label{def:compatible-dem}
	A function $V$ is said to be compatible with a demonstrator
	$\D$ iff $V$ respects the CLF conditions 
	(Eq.~\eqref{eq:clf-def}) for every observation that can be generated by 
	the demonstrator:
	\[
	V(\vzero) = 0 \ \wedge \ \forall{\vx \neq \vzero}
\left(\begin{array}{c} V(\vx) > 0\ \land\ \\ \nabla V \cdot f(\vx, \D(\vx)) < 0 \end{array}\right)\,.
	\]
	In other words, $V$ is a Lyapunov function for the closed loop system
	$\Psi(X, U, f, \D)$.
\end{definition}

Now, we describe the learning framework. The framework consists of a learner and a verifier. The learner interacts
with the verifier and the demonstrator. The framework works iteratively and
at each iteration $j$ the learner maintains a set of observations
\[
O_j : \{(\vx_1, \vu_1),\ldots,(\vx_{j}, \vu_{j})\} \subset X \times U\,.
\]
Corresponding to $O_j$, $C_j \subseteq C$ is defined as a set of candidate
unknowns for function $V_\vc(\vx)$. Formally, $C_j$ is a set of all $\vc$
s.t. $V_\vc$ is compatible with $O_j$:
\begin{equation} \label{eq:C_j_0}
C_j : \left\{ \vc \in C\left|
\begin{array}{c} V_\vc(0) = 0 \ \land \\
\bigwedge\limits_{(\vx_i, \vu_i) \in O_j}
\left(\begin{array}{c} V_\vc(\vx_i) > 0\ \land\ \\ \nabla V_\vc \cdot f(\vx_i, \vu_i) < 0 \end{array}\right)
\end{array}
\ \right.\right\}.
\end{equation}

\begin{figure}[!t]
\begin{center}
	\includegraphics[width=0.4\textwidth]{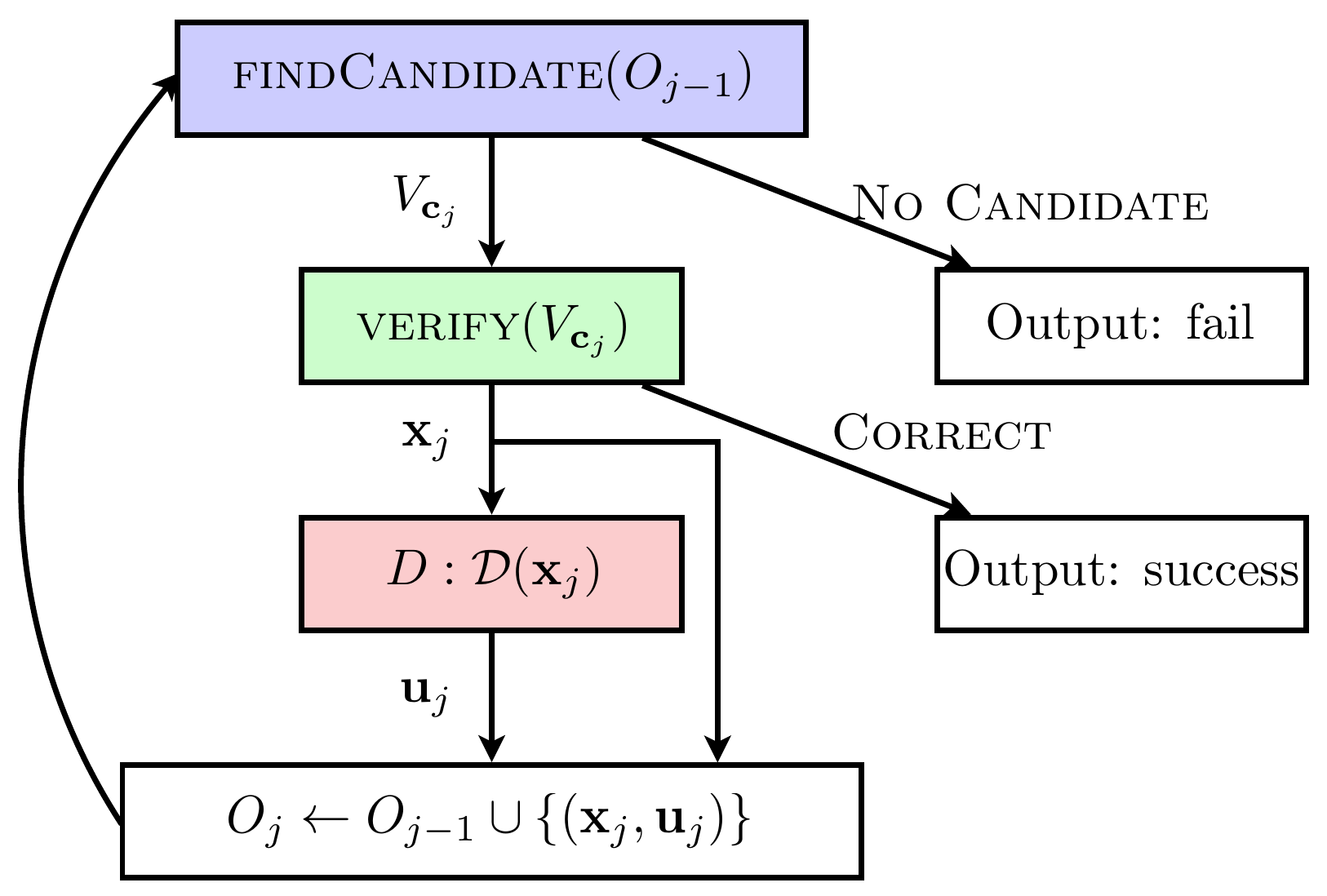}
\end{center}
\caption{Visualization of the learning framework}\label{fig:framework} 
\end{figure}

The overall procedure is shown in Fig.~\ref{fig:framework}.
The procedure starts with an empty set $O_0=\emptyset$ and the corresponding
set of compatible function parameters $C_0 : \{\vc \in C \ | \ V_\vc(\vzero) = 0 \}$.
Each iteration $j$ (starting from $j=1$) involves the following steps:
\begin{enumerate}
	\item \textsc{findCandidate}: The learner checks if there exists a $V_\vc$ compatible with $O_{j-1}$.
	\begin{enumerate}[(a)]
	\item If no such $\vc$ exists, the learner declares failure ($C_{j-1} = \emptyset$).
	\item Otherwise, a candidate $\vc_j \in C_{j-1}$ is chosen and the corresponding function $V_{\vc_j}(\vx):\vc_j.\vg(\vx)$ is considered for verification.
	\end{enumerate}
	\item \textsc{verify}: The verifier oracle tests whether $V_{\vc_j}$ is a CLF (Eq.~\eqref{eq:clf-param-def})
	\begin{enumerate}[(a)]
	\item If yes, the process terminates successfully ($V_{\vc_j}$ is a CLF)
	\item Otherwise, the oracle provides a witness $\vx_j \neq \vzero$ for the negation
	of Eq.~\eqref{eq:clf-param-def}.
	\end{enumerate}
	\item \textsc{update}: Using the demonstrator $\vu_j : \D(\vx_j)$, a new observation $(\vx_j, \vu_j)$ is added to the training set:
	\begin{equation}
	O_{j} : O_{j-1} \cup \{(\vx_j, \vu_j)\}		
	\end{equation}
	\begin{equation}\label{eq:C_j+1}
		C_{j}:\ C_{j-1} \cap \left\{ \vc \ |\ \begin{array}{c} V_{\vc}(\vx_{j}) > 0\ \land\ \\\nabla V_{\vc}\cdot f(\vx_{j}, \vu_{j}) < 0 \end{array} \right\} \,.
	\end{equation}

\end{enumerate}

\begin{theorem}\label{thm:formal-learning-thm}
  The learning framework as described above has the
  following property:
\begin{enumerate}
\item $\vc_j \not\in C_{j}$. I.e., the candidate found at the $j^{th}$ step is eliminated from further consideration.
\item If the algorithm succeeds at iteration $j$, then the output function $V_{\vc_j}$ is a valid CLF for stabilization.
\item The algorithm declares failure at iteration $j$ if and only if no linear combination of the basis functions is a CLF compatible with the demonstrator. 
\end{enumerate}
\end{theorem}
\begin{proof}
	1) Suppose that $\vc_j \in C_{j}$. Then, $\vc_j$ satisfies the following conditions (Eq.~\eqref{eq:C_j+1}):
	\[
	V_{\vc_j}(\vx_{j}) > 0\ \land\ \nabla V_{\vc_j}\cdot f(\vx_{j}, \vu_{j}) < 0 \,.
	\]
	However, the verifier guarantees that $\vc_j$ is a counterexample for Eq.~\eqref{eq:clf-def}. I.e.,
	\[
		V_{\vc_j}(\vx_{j}) \leq 0\ \lor\ \nabla V_{\vc_j}\cdot f(\vx_{j}, \vu_{j}) \geq 0 \,,
	\]
	which is a contradiction. Therefore, $\vc_j \not\in C_{j}$.
	
	2) The algorithm declares success if the verifier could not find a counterexample. In other words, $V_{\vc_j}$ satisfies conditions of Eq.~\eqref{eq:clf-def} and therefore a CLF.
	
	3) The algorithm declares failure if $C_j = \emptyset$. On the other hand, by definition, $C_j$ yields the
	set of all $\vc$ s.t. $V_\vc$ (which is linear combination 
	of basis functions) is compatible with the
	observations $O_j$. Therefore, $C_j = \emptyset$ implies that 
	that no linear combination of the basis functions is compatible
	with the $O_j$ and therefore compatible with the demonstrator.
\end{proof}

One possible choice of basis functions involves monomials
$g_j(\vx):\ \vx^{\alpha_j}$ wherein $|\alpha_j|_1 \leq D_V$ for some
degree bound $D_V$ for the learning concept (CLF).  Inverse results
suggest polynomial basis for Lyapunov functions are expressive enough
for verification of exponentially stable, smooth nonlinear systems
over a bounded region~\cite{peet2008polynomial}. This, justifies using
polynomial basis for CLF.

In the next two section we present implementations of each of the modules involved, namely the learner and the verifier.

\section{Learner}\label{sec:learner}
Recall that the learner needs to check if there exists a $\vc$
s.t. $V_\vc$ is compatible with the observation set $O$
(Definition~\ref{def:compatible-data}). In other words, we wish to
check
\[
(\exists \vc \in \C) \ V_\vc(\vzero) = 0 \wedge \bigwedge_{(\vx_i, \vu_i) \in O} 
\left(\begin{array}{c} V_\vc(\vx_i) > 0\ \land\ \\ \nabla V_\vc \cdot f(\vx_i, \vu_i) < 0 \end{array}\right)\,.
\]
Note that each function $V_{\vc}(\vx_i) : \vc^t \cdot \vg(\vx_i)$ in
our hypothesis space, is linear in $\vc$.  Also,
$\nabla V_\vc.f(\vx_i, \vu_i)$ is linear in $\vc$:
\[
\nabla V_\vc.f(\vx_i, \vu_i) = \sum_{k=1}^r c_k \nabla g_k(\vx_i).f(\vx_i, \vu_i) \,.
\]
The (initial) space of all candidates $C$ is assumed to be a
hyper-rectangular box, and therefore a polytope. Let
$\overline{C_j}$ represent the topological closure of the set $C_j$
obtained at the $j^{th}$ iteration (see Eq.~\eqref{eq:C_j_0}).
\begin{lemma}\label{lemma:cj-convex}
For each $j \geq 0$, $\overline{C_j}$ is a polytope.
\end{lemma}
\begin{proof}
	We prove by induction. Initially $C$ is an hyper-rectangular box. 
	Also, $C_0 : C \cap H_0$, where
	\[
	H_0 = \{\vc \ | \ V_\vc(\vzero) = \sum_{k=1}^r c_k g_k(\vzero) = 0 \} \,.
	\] 
	As $V_\vc$ is linear in $\vc$, $H_0 : \{\vc \ | \ \va_0^t . \vc = b_0\}$ is a hyper-plane, where $\va_0$ and $b_0$ depend on the values of, $g_k(\vzero)$ ($k = 1,\ldots,r$). And $C_0$ would be intersection of
	a polytope and a hyper-plane, which is a polytope.
	Now, assume $\overline{C_{j-1}}$ is a polytope. Recall that $C_{j}$ is defined as
	$C_{j} : C_{j-1} \cap H_{j}$ (Eq.~\eqref{eq:C_j+1}), where
	\[
	H_{j}:\  \left\{ \vc \ |\ \begin{array}{c} \sum_{k=1}^{r} (c_k \ g_k(\vx_{j})) > 0\ \land\ \\ \sum_{k=1}^{r} (c_k \ \nabla g_k(\vx_j) \cdot f(\vx_{j}, \vu_{j})) < 0 \end{array} \right\} \,.
	\]
	Notice that $f(\vx_{j}, \vu_{j})$, $g_k(\vx_i)$, and $\nabla g_k(\vx_i)$ are constants and
	\begin{align*}
	H_{j} : &H_{j1} \cap H_{j2} \\ 
	H_{j1} : &\{\vc \ | \ \va_{j1}^t . \vc > b_{j1}\} \\
	& = \{\vc | \sum_{k=1}^{r} (c_k \ g_k(\vx_{j})) > 0\} \\
	H_{j2} : &\{\vc \ | \ \va_{j2}^t . \vc > b_{j2}\}\\ & =
	\{\vc | 
	\sum_{k=1}^{r} (c_k \ \nabla g_k(\vx_j) \cdot f(\vx_{j}, \vu_{j})) < 0 \}\,.
	\end{align*}
	Therefore, $\overline{C_{j}}$ is intersection of a polytope ($\overline{C_{j-1}}$) and two half-spaces ($H_{j}$) which yields another polytope.
\end{proof}

The learner should sample a point $\vc_j \in C_{j-1}$ at $j^{th}$
iteration, which is equivalent to checking emptiness of a polytope
with some strict inequalities. This is solved using slight
modification of simplex method, using infinitesimals for strict
inequalities, or using interior point
methods~\cite{Vanderbei/2004/Linear}.  We will now demonstrate that by
choosing $\vc_j$ carefully, we can guarantee the polynomial time termination
of our learning framework.

\subsection{Termination}
Recall that in the framework, the learner provides a candidate and the
verifier refutes the candidate by a counterexample and a new
observation is generated by the demonstrator.  The following
lemma relates the sample $\vc_j \in C_{j-1}$ at the $j^{th}$ iteration
and the set $C_{j}$ in the subsequent iteration.

\begin{lemma}\label{lemma:cj-half-space}
  There exists a half-space  $H^*_{j}:\ \va^t \vc \geq b$ such that (a) $\vc_j$ lies on boundary of hyperplane $H^*_{j}$, and (b)
  $C_{j} \subseteq C_{j-1} \cap H^*_{j}$.
\end{lemma}
\begin{proof}
    Recall that we have $\vc_j \in C_{j-1}$ but $\vc_j \not\in C_{j}$ by Theorem~\ref{thm:formal-learning-thm}.
  Let $\hat{H}_j: \va^t \vc = \hat{b}$ be a separating hyperplane between the (convex) set $C_{j}$ and the point $\vc_j$, such that $C_{j} \subseteq \{ \vc\ |\ \va^t \vc \geq \hat{b}\}$. By setting the offset $b:\ \va^t \vc_j$,
  we note that $b \leq \hat{b}$. Therefore, by defining $H^*_{j}$ as $\va^t \vc \geq b$, we obtain
  the required half-space that satisfies conditions (a) and (b).
\end{proof}

While sampling a point from $C_{j-1}$ is solved by solving a linear programming problem, Lemma.~\ref{lemma:cj-half-space} suggests that the choice
of $\vc_j$ governs the convergence of the algorithm. Figure.~\ref{fig:learning-iteration}
demonstrates the importance of this choice by showing candidate $\vc_j$, hyperplanes $H_{j1}$
and $H_{j2}$ and $C_{j}$.

For a faster termination, we wish to remove a ``large portion'' of
$C_{j-1}$ to obtain a ``smaller'' $C_{j}$.  There are two important
factors which affect this: (i) counterexample $\vx_j$ selection and
(ii) candidate $\vc_j$ selection.  Counterexample $\vx_j$, would
affect $\vu_j: \D(\vx_j)$, $g(\vx_j)$, and $f(\vx_j, \vu_j)$ and
therefore defines the hyper-planes $H_{j1}$ and $H_{j2}$. On the other
hand, candidate  $\vc_j \not\in C_{j}$. We
postpone discussion on the counterexample selection to the next
section, and for the rest of this section we focus on different
techniques to generate a candidate $\vc_j \in C_{j-1}$. 

The goal is to find a $\vc_j$ s.t. 
\begin{equation} \label{eq:volume-reduction}
	\Vol(C_{j}) \leq \alpha \Vol(C_{j-1}) \,,
\end{equation}
for each iteration $j$ and a fixed constant $0 \leq \alpha < 1$, independent of the hyperplanes $H_{j1}$ and $H_{j2}$. Here $\Vol(C_j)$ represents the
  volume of the (closure) of the set $C_j$. Since closure of $C_j$ is
  contained in $C$ which happens to be compact, this volume will
  always be finite. Note that if we can guarantee Eq.~\eqref{eq:volume-reduction},
it immediately follows that $\Vol(C_j) \leq \alpha^j \Vol(C_0)$. This implies
that the volume of the remaining candidates ``vanishes'' rapidly.

\begin{remark} By referring to $\Vol(C_j)$, we are implicitly
    assuming that $C_j$ is not embedded inside a subspace of
    $\reals^r$, i.e., it is full-dimensional. However, this assumption
    is not strictly true. Specifically, $C_0 : C \cap H_0$, where
    $H_0$ is a hyper-plane. Thus, strictly speaking, the volume of
    $C_0$ in $\reals^r$ is $0$. This issue is easily addressed by
    first factoring out the linearity space of $C_0$, i.e., the affine
    hull of $C_0$. This is performed by using the equality constraints
    that describe the affine hull to eliminate variables from
    $C_0$. Subsequently, $C_0$ can be treated as a full dimensional
    polytope in $\reals^{r-d_j}$, wherein $d_j$ is the dimension of
    its linearity space.
  
Furthermore, since $C_{j} \subseteq C_0$, we can continue
    to express $C_{j} $ inside $\reals^{r-d_j}$ using the same basis
    vectors as $C_0$.  A further complication arises if $C_{j}$ is
    embedded inside a smaller subspace. We do not treat this case in
    our analysis. However, note that this can happen for at most $r$
    iterations and thus, does not pose a problem for the termination
    analysis.
\end{remark}

Intuitively, it is clear from Figure~\ref{fig:learning-iteration} that 
a candidate at the \emph{center} of $C_{j-1}$ would be a good
one. We now relate the choice of $\vc_j$ to an appropriate
  definition of center, so that Eq.~\eqref{eq:volume-reduction} is satisfied.

\begin{figure}[t]
\begin{center}
\begin{tikzpicture}[scale=0.8]
\draw[draw=black, line width = 1.5pt, fill=green!20](0,0) -- (2,0.5) -- (2,2.5) -- (0,4) -- (-1,1) -- cycle;
\draw[fill=red!20] (0.5,1.45) circle (0.1);
\draw[draw=blue!30, line width=1.5pt, pattern = north west lines, pattern color=black] (0,0) -- (2,0.5) -- (2,1.1) -- (1.2,1.9) -- ( -0.35,0.35) -- cycle;
\draw[draw=black, line width=1.5pt, dashed](-0.7,0.0) -- (2.3, 3.0);
\draw[draw=black, line width=1.5pt, dashed](2.3,0.8) -- (-1, 4.1);
\node at (0.22,1.45) {$\vc_j$};
\node at (0.3,2.3) {$C_{j-1}$};
\node at (0.8,0.6) {$C_{j}$};
\node at (-1.,3.5){$H_{j1}$};
\node at (2.8,2.8){$H_{j2}$};
\end{tikzpicture}
\end{center}
\caption{Search space: Original candidate region $C_j$ (green) at the start of the
  $j^{th}$ iteration, the candidate $\vc_j$, and the new region
  $C_{j+1}$ (hatched region with blue lines).}\label{fig:learning-iteration}
\end{figure}
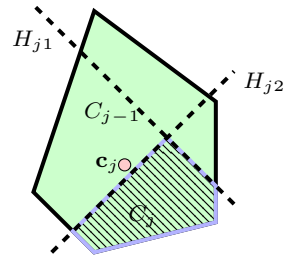

\paragraph{Center of Maximum Volume Ellipsoid}
Maximum volume ellipsoid (MVE) inscribed inside a polytope is unique with many useful
characteristics.
\begin{theorem}[Tarasov et al.\cite{tarasov1988method}] \label{thm:mve}
Let $\vc_j$ be chosen as the center of the MVE inscribed in $C_{j-1}$. Then,
\[ \Vol\left(C_{j}\right) \leq \left(1-\frac{1}{r}\right) \Vol\left(C_{j-1}\right) \,.\]
\end{theorem}

\begin{figure}[t]
\begin{center}
	\includegraphics[width=0.4\textwidth]{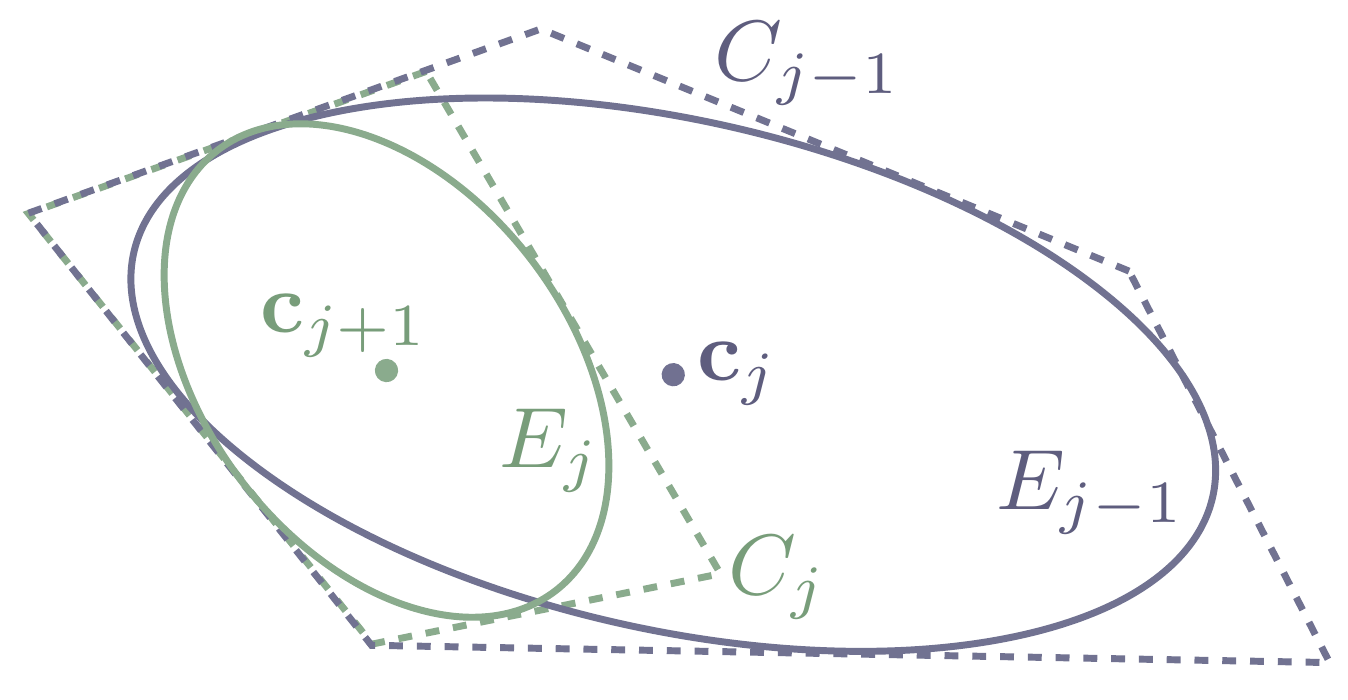}
\end{center}
\caption{Search Space: Original candidate region $C_{j-1}$ ($C_{j}$) is shown in blue (green) polygon.
The maximum volume ellipsoid $E_{j-1}$ ($E_{j}$) is inscribed in $C_{j-1}$ ($C_{j}$) and
its center is the candidate $\vc_j$ ($\vc_{j+1}$).}
\label{fig:ellipsoid}
\end{figure}

Recall, here that $r$ is the number of basis functions such that $C_{j-1} \subseteq \reals^r$. This leads us to a scheme that guarantees termination
of the overall procedure in finitely many steps under 
some assumptions. The idea is simple. Select the center of the MVE inscribed in $C_{j-1}$ at each iteration (Fig.~\ref{fig:ellipsoid}).

Let $C \subseteq (-\Delta, \Delta)^r$ for $\Delta >
0$. Furthermore, let us additionally terminate the procedure
  as having failed whenever the $\Vol(C_j) < (2\delta)^r$ for
  some arbitrarily small $\delta > 0$. This additional termination
  condition is easily justified when one considers the precision
  limits of floating point numbers and sets of small volumes. Clearly,
  as the volume of the sets $C_j$ decreases exponentially, each point
  inside the set will be quite close to one that is outside, requiring
  high precision arithmetic to represent and sample from the sets
  $C_j$. 

\begin{theorem} \label{thm:termination}
If at each step $\vc_j$ is chosen as the center of the MVE in $C_{j-1}$, the learning loop terminates in at most
\[ \frac{r (\log(\Delta) - \log(\delta))}{- \log\left(1 - \frac{1}{r}\right) } = O(r^2) \ \mbox{iterations}\,.\]
\end{theorem}

\begin{proof}
	Initially, $\Vol(C_0) <  (2\Delta)^r$. Then by Theorem~\ref{thm:mve}
	\begin{align*}
	&\Vol(C_j) \leq (1 - \frac{1}{r})^j \ \Vol(C_0) < (1 - \frac{1}{r})^j (2\Delta)^r \\
	\implies & \log\left(\frac{\Vol(C_j)}{(2\Delta)^r}\right) < j \  \log(1-\frac{1}{r})\,.
	\end{align*}
	After $k = \frac{r(\log(\Delta)-\log(\delta))}{-\log(1-\frac{1}{r})}$ iterations:
	\begin{equation*}
		\log\left(\frac{\Vol(C_j)}{(2\Delta)^r}\right) < \frac{r(\log(\Delta)-\log(\delta))}{-\log(1-\frac{1}{r})} \  \log(1-\frac{1}{r})\,,
	\end{equation*}
	and
	\begin{align*}
		\implies & \log\left(\frac{\Vol(C_j)}{(2\Delta)^r}\right) < r \log\left(\frac{\delta}{\Delta}\right) \\
		\implies & \log\left(\frac{\Vol(C_j)}{(2\Delta)^r}\right) < r \log\left(\frac{2\delta}{2\Delta}\right) \\
		\implies & \log(\Vol(C_k)) < \log((2\delta)^r)\,.
	\end{align*}	
	And it is concluded that $\Vol(C_k) < (2\delta)^r$, which is the termination condition. And asymptotically $-\log(1 - \frac{1}{r})$ is $\Omega(\frac{1}{r})$ (can be shown using Taylor expansion as $r \rightarrow \infty$) and therefore, the maximum number of iterations would be $O(r^2)$.
\end{proof}

However, checking the termination condition is computationally
expensive as calculating the volume of a polytope is
$\sharp P$ hard, i.e., as hard as counting the number of
  solutions to a SAT problem. One solution is to first calculate an
upper bound on the number of iterations using
Theorem~\ref{thm:termination}, and stop if the number of iterations
has exceeded the upper-bound.

A better approach is to consider some robustness for the candidate.
\begin{definition}[Robust Compatibility]
A candidate $\vc$ is $\delta$-robust for $\delta > 0$ w.r.t. 
observations (demonstrator), iff for each $\hat{\vc} \in \B_{\delta}(\vc)$, 
$V_{\hat{\vc}}:\hat{\vc}^t \cdot \vg(\vx)$ is
compatible with observations (demonstrator) as well. 
\end{definition}

Let $E_j$ be the MVE inscribed inside $C_j$ (Fig.~\ref{fig:ellipsoid}).
Following the robustness assumption, it is sufficient to terminate
the procedure whenever:
\begin{equation}\label{eq:termin-cond-2}
\Vol(E_j) < \gamma \delta^r\,,
\end{equation}
where $\gamma$ is the volume of $r$-ball with radius $1$.
\begin{theorem}[\cite{tarasov1988method,khachiyan1990inequality}] \label{thm:mve-2}
Let $\vc_j$ be chosen as the center of $E_{j-1}$. Then,
\[ \Vol(E_{j}) \leq \left(\frac{8}{9}\right) \Vol\left(E_{j-1}\right) \,.\] 
\end{theorem}

\begin{theorem} \label{thm:termination2}
If at each step $\vc_j$ is chosen as the center of $E_{j-1}$, the learning loop condition defined by Eq.~\eqref{eq:termin-cond-2} is violated in at most
\[ \frac{r (\log(\Delta) - \log(\delta))}{- \log\left(\frac{8}{9}\right) } = O(r) \ \mbox{iterations}\,.\]
\end{theorem}

\begin{proof}
	Initially, $\B_{\Delta}(\vzero)$ is the MVE inside box $[-\Delta, \Delta]^r$ and therefore, $\Vol(E_0) < \gamma \Delta^r$. 
	Then by Theorem~\ref{thm:mve}
	\begin{align*}
	&\Vol(E_j) \leq (\frac{8}{9})^j \ \Vol(E_0) < (\frac{8}{9})^j \gamma \Delta^r \\
	\implies & \log(\Vol(E_j)) - \log(\gamma \Delta^r) < j \  \log(\frac{8}{9}) \,.
	\end{align*}
	After $k = \frac{r(\log(\Delta)-\log(\delta))}{-\log(\frac{8}{9})}$ iterations:
	\begin{equation*}
		\log(\Vol(E_k)) - \log(\gamma \Delta^r) < \frac{r(\log(\Delta)-\log(\delta))}{-\log(\frac{8}{9})} \  \log(\frac{8}{9}) \,,
	\end{equation*}
	and
	\begin{align*}
		\implies & \log(\Vol(E_k)) - \log(\gamma \Delta^r) < r(\log(\delta) - \log(\Delta)) \\
		\implies & \log(\Vol(E_k)) - \log(\gamma \Delta^r) < \log(\gamma \delta^r) - \log(\gamma \Delta^r) \\
		\implies & \log(\Vol(E_k)) < \log(\gamma \delta^r)\,.
	\end{align*}	
	It is concluded that $\Vol(E_k) < \gamma \delta^r$, which is the termination condition. And asymptotically the maximum number of iterations would be $O(r)$.
\end{proof}

Volume of an ellipsoid is effectively computable and thus, such termination 
condition can be checked easily. Also, the convergence rate is linear in $r$
as opposed to $r^2$, when the robustness is not guaranteed.

\begin{theorem}\label{thm:clf-or-no-robust-solution}
	The learning framework either finds a control Lyapunov functions or proves that no linear combination of
	the basis function would yield a function with robust compatibility 
	with the demonstrator.
\end{theorem}
\begin{proof}
	By Theorem~\ref{thm:formal-learning-thm}, if verifier certifies correctness of
	a solution $V$, then $V$ is a CLF. Assume that the framework terminates after
	$k$ iterations and no solution is found. Then, by Theorem~\ref{thm:termination},
	$\Vol(E_k) < \gamma \delta^r$. This means that a ball with radius $\delta$ 
	would not fit in $C_k$ as $E_k$ is the MVE inscribed inside $C_k$.
	In other words
	\[
	(\forall \vc \in C_k) \ (\exists \hat{\vc} \in \B_{\delta}(\vc)) \ \hat{\vc} \not\in C_k\,.
	\]
	On the other hand, for all $\vc \not\in C_k$, $V_\vc$ is not compatible
	with the observations $O_j$. Therefore, even if there is a CLF $V_\vc$ 
	s.t. $\vc \in C_k$, the CLF is not robust in its compatibility with the
	demonstrator.
\end{proof}
The MVE itself can be computed by solving a convex optimization 
problem\cite{tarasov1988method,vandenberghe1998determinant}.

\paragraph{Other Definitions for Center of Polytope:}
Beside the center of MVE inscribed inside a polytope,
there are other notions for defining center of a polytope. These include
the center of gravity and 
Chebyshev center.
Center of gravity provides the following inequality~~\cite{bland1981ellipsoid}
\[ \Vol\left(C_{j}\right) \leq \left(1-\frac{1}{e}\right) \Vol\left(C_{j}\right)
< 0.64 \ \Vol(C_{j-1}) \,,\]
meaning that the volume of candidate set is reduced by at least 36\% at each iteration.
Unfortunately, calculating center of gravity is very expensive.
Chebyshev center~\cite{elzinga1975central} of a polytope is the center of the largest Euclidean ball 
that lies inside the polytope. Finding a Chebyshev center for a polytope
is equivalent to solving a linear program, and while it yields a good heuristic, it 
would not provide an inequality in the form of Eq.~\eqref{eq:volume-reduction}.

There are also notions for defining center for a set of constraints, 
including analytic center, and volumetric center.  
Assuming $C : \{ \vc \ | \bigwedge_i \va_i^t . \vc < b_i \}$, then
analytic center for $\bigwedge_i \va_i^t . \vc < b_i$ is defined as
\[
ac(\bigwedge_i \va_i^t . \vc < b_i) = \argmin_{\vc} - \sum_i \log (b_i - \va_i^t . \vc) \,.
\]
Notice that infinitely many inequalities can represent $C$ and any point inside $C$ can be
an analytic center depending on the inequalities. 
Atkinson et al.~\cite{atkinson1995cutting} and Vaidya~\cite{vaidya1996new} provide
candidate generation techniques, based on these centers 
, along with appropriate termination conditions and 
convergence analysis.

\section{Verifier}\label{sec:verifier}

The verifier checks the CLF conditions in Eq.~\eqref{eq:clf-param-def} for a 
candidate $V_{\vc_j}(\vx): \vc_j^t \cdot \vg(\vx)$. Since the CLF is generated by 
the learner, it is guaranteed that $V_{\vc_j}(\vzero) = 0$ (Eq.~\eqref{eq:C_j_0}). 
Accordingly, verification is split into two 
separate checks:

\noindent\textbf{(A)} Check if $V_{\vc_j}(\vx)$ is a positive polynomial for
  $\vx \neq \vzero$, or equivalently:
  \begin{equation}\label{eq:positivity-cond}
 (\exists\ \vx \neq \vzero)\ V_{\vc_j}(\vx) \leq 0 \,.
\end{equation}

\noindent\textbf{(B)} Check if the Lie derivative of $V_{\vc_j}$ can be made negative for
each $\vx \neq \vzero$ by a choice $\vu \in U$: 
\begin{equation}\label{eq:decrease-cond-init}
 (\exists \vx \neq \vzero ) \ (\forall \vu \in U)\ (\nabla V_{\vc_j}) \cdot f(\vx, \vu) \geq 0 \,.
\end{equation}
This problem \emph{seems} harder  due to the presence of a \emph{quantifier alternation}.
\begin{lemma} \label{lem:control-dual}
Eq.~\eqref{eq:decrease-cond-init} holds for some $\vx \neq \vzero$ iff
\begin{equation}\label{eq:decr-condition}
\begin{array}{ll}
(\exists\ \vx \neq \vzero,\vlam) & \vlam\geq \vzero, \vlam^t \vb \geq -\nabla V_{\vc_j}.f_0(\vx) \\	
& A_i^t \vlam=\nabla V_{\vc_j}.f_i(\vx) (i \in \{1 \ldots m\}).
\end{array}
\end{equation}
\end{lemma}

\begin{proof}
	Suppose Eq.~\eqref{eq:decrease-cond-init} holds. Then, for the given $V$, there
	exists a $\vx \neq \vzero$ s.t.
	\begin{equation} \label{eq:u-cond-before-farkas-lemma}
	(\forall \vu \in U) \ \nabla V \cdot f(\vx,\vu) = \left( \begin{array}{c} \nabla V \cdot f_0(\vx) + \\ \sum\limits_{i=1}^{m} \nabla V \cdot f_i(\vx) u_i \end{array} \right)\hspace{-0.1cm} \geq0,
	\end{equation}
	which is equivalent to:
	\[
		(\not\exists \vu) A \vu \geq \vb \land \nabla V \cdot f_0(\vx) + \sum_{i=1}^{m} \nabla V \cdot f_i(\vx) u_i < 0 \,.
	\]
	 This yields a set of linear inequalities (w.r.t. $\vu$). Using Farkas lemma, this
	 is equivalent to
	\begin{align*}
	(\exists \vlam \geq 0) & A_i^t \vlam = \nabla V \cdot f_i(\vx) (i \in \{1...m\})\\
	& \vlam^t \vb \geq -\nabla V \cdot f_0(\vx).
	\end{align*}
	Thus, for a given $V$, Eq.~\eqref{eq:decrease-cond-init} is equivalent
	to Eq.~\eqref{eq:decr-condition}.
\end{proof}

The verifier needs to check Eq.~\eqref{eq:positivity-cond} and
Eq.~\eqref{eq:decr-condition}. This problem is \emph{in general}
undecidable if the basis functions include trigonometric and
exponential functions. However, $\delta$-decision procedures can solve
these problems approximately~\cite{gao2013dreal}. In our
experience, $\delta$-decision procedures do not scale as
  verifiers for the range of benchmarks we wish to
  tackle. Nevertheless, these solvers allow us to conveniently
implement a verifier for small but hard problems involving rational
and trigonometric functions.

Assuming that the dynamics and chosen bases are polynomials in $\vx$,
the verification problem reduces to checking if a given semi-algebraic
set defined by polynomial inequalities is empty.  The verification
problem for polynomial dynamics and polynomial CLFs is decidable with 
a high complexity (NP hard)~\cite{Basu+Pollock+Roy/03/Algorithms}.
Exact
approaches using semi-algebraic
geometry~\cite{Brown:2007:CQE:1277548.1277557} or branch-and-bound
solvers (including the dReal approach cited above) can tackle this
problem precisely. However, for scalability, we consent to a
relaxation using SDP solvers. We now present a
relaxation using semidefinite programming (SDP) solvers.

\subsection{SDP Relaxation}

Let $\vw:\ [\vx, \vlam]$ collect the state variables $\vx$ and the
dual variables $\vlam$ involved in the conditions stated
in~\eqref{eq:decr-condition}. The core idea behind the SDP relaxation
is to consider a vector collecting all monomials of degree up to $D$:
\[ \vm:\ \left(\begin{array}{c}
               1 \\  w_1 \\ w_2\\ \ldots \\ \vw^{D}\\ \end{array}\right) \,,\]
wherein $D$ is chosen to be at least half 
of the maximum degree in $\vx$ among all monomials in $g_j(\vx)$ and
$\nabla g_j \cdot f_i(\vx)$:
\[
D \geq \frac{1}{2} \max\left( \bigcup_{j} \left( \{ \mbox{deg}(g_j) \} \cup \{ \bigcup_i\mbox{deg}(\nabla g_j \cdot f_i ) \} \right) \right).
\]
Let us define $Z(\vw): \vm \vm^t$, which is a symmetric matrix
  of monomial terms of degree at most $2D$. Each polynomial of degree
up to $2D$ may now be written as a trace inner product
\[p(\vx, \vlam):\ \tupleof{ P, Z(\vw)} = \mathsf{trace}( P Z(\vw) )\,,\]
wherein the matrix $P$ has real-valued entries that define the
  coefficients in $p$ corresponding to the various
  monomials. Although, $Z$ is a function of $\vx$ and $\vlam$, we will
  write $Z(\vx)$ as a function of just $\vx$ to denote the matrix
  $Z([\vx, \vzero])$ (i.e., set $\vlam = \vzero$).

Checking Eq.~\eqref{eq:positivity-cond}
is equivalent to solving the following optimization
problem over $\vx$
\begin{equation}\label{eq:positivity-cond-relax}
\begin{array}{ll}
 \mathsf{max}_{\vx} \tupleof{I,Z(\vx)} & \\
 \mathsf{ s.t. }  &\tupleof{\mathcal{V}_{\vc_j}, Z(\vx)} \leq 0\,, \\
\end{array}
\end{equation}
wherein $I$ is the identity matrix, and $V_{\vc_j}(\vx)$ is written in the
inner product form as $\tupleof{\mathcal{V}_{\vc_j}, Z(\vx)}$.  Let
$\tupleof{\Lambda_k, Z(\vw)}$ represent the variable $\lambda_k$. $\vlam$ is
represented as vector $\Lambda(Z(\vw))$, wherein the $k^{th}$ element is
$\tupleof{\Lambda_k, Z(\vw)}$.  Then, the conditions in
~\eqref{eq:decr-condition} are now written as
\begin{equation}\label{eq:decr-cond-relax}
 \begin{array}{ll}
 \mathsf{max}_{\vw} \tupleof{I,Z(\vw)} & \\
 \mathsf{s.t.} & \hspace{-1.4cm} \tupleof{F_{\vc_j,i}, Z(\vw)} = A_i^t 
 \Lambda(Z(\vw)),\ i \in \{1,\ldots, m\}  \\
& \hspace{-1.4cm} \tupleof{-F_{\vc_j,0}, Z(\vw)} \leq \vb^t \Lambda(Z(\vw)) \\
& \hspace{-1.4cm} \Lambda(Z(\vw)) \geq 0 \,,
\end{array}
\end{equation}
wherein the components $\nabla V_{\vc_j} \cdot f_i(\vx)$ 
defining the Lie derivatives of $V_{\vc_j}$ are now written
in terms of $Z(\vw)$ as $\tupleof{F_{\vc_j,i},Z(\vw)}$.
Notice that $Z(\vzero)$ is a square matrix where the first element ($Z(\vzero)_{1,1}$) is $1$ and the rest of the entries are zero. Let $Z_0 = Z(\vzero)$ . Then $\tupleof{I, Z_0} = 1$, and $(\forall \vw) \ Z(\vw) \succeq Z_0$.

The SDP relaxation is used to solve these problems and provide an upper
bound of the solution and $D$ defines the degree of 
relaxation~\cite{henrion2009gloptipoly}. The 
relaxation treats $Z(\vw)$ as a fresh matrix variable $Z$ that is no longer
a function of $\vw$. The constraint $Z \succeq Z_0$ is added.
$Z(\vw): \vm \vm^t$ is a rank one matrix and ideally, $Z$ should
be constrained to be rank one as well. However, such a constraint is
non-convex, and therefore, will be dropped from our relaxation.
Also, constraints involving $Z(\vw)$ in Eqs.~\eqref{eq:positivity-cond-relax} and~\eqref{eq:decr-cond-relax} are added
as support constraints (cf.~\cite{lasserre2001global,lasserre2009moments,henrion2009gloptipoly}).
Both optimization
problems (Eqs.\eqref{eq:positivity-cond-relax} and~\eqref{eq:decr-cond-relax}) are
feasible by setting $Z$ to be $Z_0$. Furthermore, if the optimal
solution for each problem is $1$ in the SDP relaxation, then we will
conclude that the given candidate is a CLF. Unfortunately, the
  converse is not necessarily true: the relaxation may fail to
  recognize that a given candidate is in fact a CLF.

\begin{lemma}\label{lem:non-zero-sol}
Whenever the relaxed optimization problems in Eqs.~\eqref{eq:positivity-cond-relax} and ~\eqref{eq:decr-cond-relax}
yield $1$ as a solution, then the given candidate $V_{\vc_j}(\vx)$ is in fact a CLF. 
\end{lemma}

\begin{proof}
	Suppose that $V_{\vc_j}$ is not a CLF but both optimization problems yield an optimal value of $1$. Then, one of Eq.~\eqref{eq:positivity-cond} or Eq.~\eqref{eq:decrease-cond-init} is satisfied. 
	I.e. $(\exists \vx^* \neq \vzero, \vlam^* \geq \vzero)$ s.t. $V_{\vc_j}(\vx^*) \leq 0$ or $A_i^t \vlam^*=\nabla V_{\vc_j}.f_i(\vx^*) (i \in \{1 \ldots m\}) ,\vlam^{*t} \vb \geq - \nabla V_{\vc_j}.f_0(\vx^*)$.
		Let $\vw^* = [\vx^*, \vlam^*]$ and therefore $Z(\vw^*) \succeq Z_0$ 
		is a solution for Eq.~\eqref{eq:positivity-cond-relax} or
        Eq.~\eqref{eq:decr-cond-relax}.  Let $Z' = Z(\vx^*) - Z_0$. As
        $\vw^* \neq \vzero$, $Z'$ has a non-zero diagonal element, and
        since $Z' \succeq 0$, we may also conclude that one of the
        eigenvalues of $Z'$ must be positive. Therefore,
        $\tupleof{I, Z'} > 0$ as the trace of $Z'$ is the sum of
        eigenvalues of $Z'$. Thus,
        $\tupleof{I, Z(\vw)} > \tupleof{I, Z_0} = 1$. Thus,
          the optimal solution of at least one of the two problems has
          to be greater than one. This contradicts our original
          assumption.
\end{proof}

However, the converse is not true. It is possible for $Z \succeq Z_0$
to be optimal for either relaxed condition, but 
  $Z \not= Z(\vw)$ for any $\vw$.
This happens because (as mentioned earlier) the relaxation drops two key
constraints to convexify the conditions: (1) $Z$ has to be a rank one
matrix written as $Z: \vm \vm^t$ and (2) there is a $\vw$ such that
$\vm$ is the vector of monomials corresponding to $\vw$. 

\begin{lemma}\label{lem:no-lam-relaxation}
	Suppose Eq.~\eqref{eq:decr-cond-relax} has a solution $Z \not= Z_0$, then
	\begin{align*}
		& (\forall \vu \in U) \ \tupleof{F_{\vc_j,0}, Z} + \sum_{i=1}^m \tupleof{F_{\vc_j,i}, Z} u_i \geq 0\,.
	\end{align*}
\end{lemma}
\begin{proof}
	While in the relaxed problem, the relation between monomials are lost, each inequality in
	Eq.~\eqref{eq:decr-cond-relax} holds. Let $\hat{\vlam} = \Lambda(Z)$. Then, we have:
	\begin{align*}
		\tupleof{F_{\vc_j,i}, Z} = A_i^t \hat{\vlam},\ i \in \{1,\ldots, m\}  \\
		\tupleof{-F_{\vc_j,0}, Z} \leq \vb^t \hat{\vlam} , \ \hat{\vlam} \geq 0\,.
	\end{align*}	
	Similar to Lemma.~\ref{lem:control-dual} (using Farkas Lemma) this is equivalent to
	\begin{equation*}
(\forall \vu \in U) \ \tupleof{F_{\vc_j,0}, Z} + \sum_{i=1}^m \tupleof{F_{\vc_j,i}, Z} u_i \geq 0 \,.
	\end{equation*}
\end{proof}

\subsection{Lifting the Counterexamples}

Thus far, we have observed that the relaxed optimization
  problems (Eqs.~\eqref{eq:positivity-cond-relax} 
  and~\eqref{eq:decr-cond-relax}) yield matrices $Z$ as
  counterexamples, rather than vectors $\vx$. Furthermore, given a
  solution $Z$, there is no way for us to extract a corresponding
  $\vx$ for reasons mentioned above. We solve this issue by ``lifting''
our entire learning loop to work with  observations of the form:
\[ O_j: \{ (Z_1, \vu_1),\ldots,(Z_{j}, \vu_{j})\} \,,\]
effectively replacing states $\vx_i$ by matrices $Z_i$.

Also, each basis function $g_k(\vx)$ in $\vg$ is now written instead as $\tupleof{G_k, Z}$.
The candidates are therefore, $\sum_{k=1}^r  c_k \tupleof{ G_k, Z}$.
 Likewise, we write the components of its Lie
 derivative $\nabla g_k \cdot f_i$ in terms of $Z$ ($\tupleof{G_{ki}, Z}$).
Therefore
\begin{align}\label{eq:relaxed-template}
	\mathcal{V}_\vc = \sum_{k=1}^r c_{k} G_k \ , \ F_{\vc,i} = \sum_{k=1}^r c_{k} G_{ki}\,.
\end{align}
 
\begin{definition}[Relaxed CLF]\label{def:relaxed-CLF}
	A polynomial function $V_\vc(\vx) = \sum_{k=1}^r c_k g_k(\vx)$, s.t. $\tupleof{\mathcal{V}_\vc, Z_0} = 0$ is a $D$-relaxed CLF iff for all
	$Z \not= Z_0$:
	\begin{equation}\label{eq:relaxed-clf}
	\begin{array}{l}  \tupleof{\mathcal{V}_\vc, Z} > 0 \ \land \\
		(\exists \vu \in U) \ \tupleof{F_{\vc,0}, Z} + \sum_{i=1}^m \tupleof{F_{\vc,i}, Z} < 0\,.
	\end{array}
		\end{equation}
\end{definition}

\begin{theorem}\label{thm:relaxed-CLF-vs-CLF}
	A relaxed CLF is a CLF.
\end{theorem}
\begin{proof}
	Suppose that $V_\vc$ is not a CLF. The proof is complete by showing that $V_\vc$
	is not a relaxed CLF. If $V_\vc(\vzero) \neq 0$, then $\tupleof{\mathcal{V}_\vc, Z_0} \neq 0$ and $V_\vc$ is not a relaxed CLF. Otherwise, according to Eq.~\eqref{eq:clf-def} 
	there exists a $\vx \neq \vzero$ s.t.
	\[
	V_\vc(\vx) \leq 0  \ \lor \ (\forall \vu \in U) \ \nabla V_\vc.f(\vx, \vu) \geq 0 \,.
	\]
	Therefore, there exists $\vx \neq \vzero$ s.t.
	\begin{align*}
	&\tupleof{\mathcal{V}_\vc, Z(\vx)} \leq 0 \ \lor \\
	&(\forall \vu \in U) \ \tupleof{F_{\vc,0}, Z(\vx)} + \sum_{i=1}^m \tupleof{F_{\vc,i}, Z(\vx)} u_i \geq 0	\,.
	\end{align*}
	Setting $Z:\ Z(\vx)$ shows that  
	$V_\vc$ is not a relaxed CLF, since the negation of Eq.~\eqref{eq:relaxed-clf} holds.
\end{proof}

We lift the overall formal learning framework to work with
  matrices $Z$ as counterexamples using the following modifications to
  various parts of the framework:
\begin{enumerate}
\item First, for each $(Z_j, \vu_j)$ in the observation set, $Z_j$ is the feasible solution
  returned by the SDP solver while solving Eqs.~\eqref{eq:decr-cond-relax} and ~\eqref{eq:positivity-cond-relax}.

\item However, the demonstrator $\D$ requires its input to be a state
  $\vx \in X$.  We define a projection operator $\pi: \zeta \mapsto X$
  mapping each $Z$ to a state $\vx: \pi(Z)$, such that the
  demonstrator operates over $\pi(Z_j)$ at each step. Note that the
  vector of monomials $\vm$ used to define $Z$ from $\vx$ includes the
  degree one terms $x_1, \ldots, x_n$. The projection operator
  simply selects the entries from $Z$ corresponding to these
  variables. Other more sophisticated projections are also possible,
  but not considered in this work.

\item The space of all candidates $C$ remains unaltered except
    that each basis polynomial is now interpreted as
    $g_j: \tupleof{G_j, Z}$ and similarly for the Lie derivative
    $(\nabla g_j)\cdot f(\vx, \vu)$. Thus, the learner is effectively
    unaltered.
\end{enumerate}

\begin{definition}[Relaxed Observation Compatibility] \label{def:compatible-data-relaxed}
	A polynomial function $V_\vc$ is said to be compatible with a set of 
	$D$-relaxed-observations $O$ iff $V_\vc$ respects the $D$-relaxed CLF conditions 
	(Eq.~\eqref{eq:clf-def}) for every point in $O$:
	\begin{align*}
	& \tupleof{\mathcal{V}_\vc, Z_0} = 0 \ \wedge \\	
	&\bigwedge\limits_{(Z_k, \vu_k) \in O_j}
\left(\begin{array}{c} \tupleof{\mathcal{V}_\vc, Z_k} > 0\ \land\ \\ \tupleof{F_{\vc,0}, Z_k} + \sum_{i=1}^m \tupleof{F_{\vc,i}, Z_k}u_{ki} < 0 \end{array}\right)\,.
	\end{align*}
\end{definition}

\begin{definition}[Relaxed Demonstrator Compatibility] \label{def:compatible-dem-relaxed}
	A polynomial function $V_\vc$ is said to be compatible with a relaxed-demonstrator
	$\D \circ \pi$ iff $V_\vc$ respects the $D$-relaxed CLF conditions 
	(Eq.~\eqref{eq:clf-def}) for every observation that can be generated by 
	the relaxed-demonstrator:
	\begin{align*}
	& \tupleof{\mathcal{V}_\vc, Z_0} = 0 \ \wedge \\ &(\forall Z \succeq Z_0, \ Z \neq Z_0)\\
	& \ \ \ \ \ \ \ \ 
	\left(\begin{array}{c} \tupleof{\mathcal{V}_\vc, Z} > 0\ \land\ \\ \tupleof{F_{\vc,0}, Z} + \sum_{i=1}^m \tupleof{F_{\vc,i}, Z} \D(\pi(Z))_i < 0 \end{array}\right)\,.	
	\end{align*}
	In other words, $V_\vc$ is a relaxed Lyapunov function for the closed loop system
	$\Psi(X, U, f, \D \circ \pi)$.
\end{definition}

\begin{theorem}
	The adapted formal learning framework terminates and either finds a CLF $V$, or proves that 
	no linear combination of basis functions would yield a
	CLF, with robust compatibility w.r.t. the (relaxed) demonstrator.
\end{theorem}
\begin{proof}
	$C_{j-1}$ represents all $\vc$ s.t. $V_\vc$
	is compatible with relaxed-observation $O_{j-1}$. Still $\mathcal{V}_\vc$ and $F_{\vc,i}$
	are linear in $\vc$ (Eq.~\eqref{eq:relaxed-template}), and therefore $C_{j-1}$ which is
	the set of all $\vc \in C$ s.t.
	\begin{equation*}
	\begin{array}{l}
		\tupleof{\mathcal{V}_\vc, Z_0} = 0 \ \wedge \\
		\bigwedge\limits_{(Z_k, \vu_k) \in O_{j-1}}
			\left( \begin{array}{c} \tupleof{\mathcal{V}_\vc, Z_k} > 0\ \land\ \\ 
				\sum_{i=1}^m \tupleof{F_{\vc,i}, Z_k}u_{ki}
				+ \tupleof{F_{\vc,0}, Z_k} < 0 
			\end{array} \right)
	\end{array} \,,
	\end{equation*}
 	is a polytope (similar to Lemma~\ref{lemma:cj-convex}).
	Suppose, at $j^{th}$ iteration, $V_{\vc_j} : \vc_j^t . \vg$ is generated by the learner. 
	The relaxed verifier solves Eqs.~\eqref{eq:positivity-cond-relax} 
  and~\eqref{eq:decr-cond-relax}. If the optimal solution for these problems are $1$,
  by Lemma~\ref{lem:non-zero-sol}, $V_{\vc_j}$ is a CLF. Otherwise, it returns a 
  counterexample $Z_j \succeq Z_0$ and $Z_j \neq Z_0$. More over, according to Eqs.~\eqref{eq:positivity-cond-relax} and~\eqref{eq:decr-cond-relax} and Lemma~\ref{lem:no-lam-relaxation}:
  \begin{align*}
  &\tupleof{\mathcal{V}_{\vc_j}, Z_j} \leq 0 \ \lor \\ &(\forall \vu \in U) \ \tupleof{F_{\vc_j,0}, Z_j} + \sum_{i=1}^m \tupleof{F_{\vc_j,i}, Z_j} u_i \geq 0\,.	
  \end{align*}
  	In other words, $V_{\vc_j}$ is not a $D$-relaxed CLF. Next, the demonstrator 
	generates a proper feedback for $\pi(Z_j)$ and observation
	$(Z_j, \D(\pi(Z_j)))$ is added to the set of observations.
	Notice that $V_{\vc_j}$ does
	not respect the $D$-relaxed CLF conditions for $(Z_j, \D(\pi(Z_j)))$. I.e.
	\begin{align*}
	&\tupleof{\mathcal{V}_{\vc_j}, Z_j} \leq 0 \ \lor \\ &\tupleof{F_{\vc_j,0}, Z_j} + \sum_{i=1}^m \tupleof{F_{\vc_j,i}, Z_j} \D(\pi(Z_j))_i \geq 0 \,.	
	\end{align*}
	Therefore, the new set $C_{j}$ does not contain $\vc_j$.
	Now, the learner uses the center of maximum volume ellipsoid, 
	to generate the next candidate. This process repeats and the learning
	procedure terminates in finite iterations. When the algorithm
	returns with no solution, it means that $\Vol(C_j)$ $\leq \gamma \delta^r$.
	Similar to Theorem~\ref{thm:clf-or-no-robust-solution}, this guarantees 
	that no ball of radius $\delta$ fits inside $C_j$, which represents the
	set of all linear combination of basis functions, compatible
	with the relaxed observations. Therefore, no linear combination of basis functions
	would yield a CLF with robust compatibility with the relaxed 
	observation and therefore with the relaxed-demonstrator.
\end{proof}

In the rest of this paper, we use CLF for discussions. Nevertheless, the same
results can be applied to relaxed CLF as well.

\subsection{Counterexamples Selection}\label{sec:counterexample-selection}

As discussed earlier, in Section~\ref{sec:learner}, there are
  two important factors that affect the overall convergence rate of
  the learning framework: (a) the choice of a candidate
  $\vc_j \in C_{j-1}$ and (b) the choice of a counterexample $\vx_j$ that
  shows that the current candidate $V_{\vc_j}$ is not a CLF. We will now
  discuss the choice of a ``good'' counterexample.

As mentioned, when there is a counterexample $\vx_j$ for $V_{\vc_j}$, 
there are two half spaces 
$H_{j1} : \{\vc \ | \ \va_{j1}^t . \vc > b_{j1}\}$, and
$H_{j2} : \{\vc \ | \ \va_{j2}^t . \vc > b_{j2}\}$ such that 
$C_{j} : C_{j-1} \cap H_{j1} \cap H_{j2}$. In particular,
$\vc_j \not\in C_{j}$, yields the following constraints over $\vc_j$:
\begin{equation}\label{eq:cj-property-counterexample}
\va_{j1}^t . \vc_j \leq b_{j1} \lor \va_{j2}^t . \vc_j \leq b_{j2} \,.
\end{equation}
In general, the counterexample affects the coefficients of the
  half-spaces $\va_{jl}, b_{jl}$ for $l \in \{1,2\}$.  To wit, the
counterexample $\vx_j$ defines values for $\vu_j : \D(\vx_j)$,
$g_i(\vx_j)$, $f_i(\vx_j, \vu_j)$, which in turn, define $H_{j1}$ and
$H_{j2}$.  Thus, a good counterexample should ``remove'' as large a
set as possible from $C_{j-1}$. Looking at
  Eq.~\eqref{eq:cj-property-counterexample}, it is clear that
  $\va_{jl}^t . \vc_j - b_{jl} $ would measure how ``far away'' the
  counterexample is from the boundary of the half-space $H_{jl}$,
  assuming that $||\va_{jl}||$ is kept constant.  As proposed in our
earlier work~\cite{Ravanbakhsh-Others/2015/Counter-LMI}, one could
find a counterexample that maximizes these quantities, so that a
  ``good'' counterexample can be selected.  For checking~\eqref{eq:positivity-cond}, the
verifier finds a counterexample $\vx$ that maximizes a slack variable $\gamma$ s.t.
\[
V_{\vc_j}(\vx) \leq -\gamma \,,
\]
and for the second check~\eqref{eq:decr-condition}, the slack variable
$\gamma$ is introduced and maximized as follows:
\begin{align*}
&\vlam \geq \gamma \ \land \ \bigwedge_{i=1}^m A_i^t \vlam = \nabla V_{\vc_j} \cdot f_i(\vx) \ \land \\
&\vlam^t . \vb \geq -\nabla V_{\vc_j} \cdot f_0(\vx) + \gamma \,.
\end{align*}

As such, we cannot prove improved bounds on the number of
  iterations to terminate using this approach. However, we do, in
  fact, see a significant decrease in the number of iterations 
  by adding an objective function to the selection of the counterexample.

\section{Specifications}\label{sec:spec}
In previous sections, the problem of finding a CLF was
discussed. However, the concept can be extend to other Lyapunov-like
arguments that are useful for specifications such as
  reach-while-stay, and safety. In this section,
some of these specifications are addressed.
\subsection{Local Lyapunov Function}
Many nonlinear systems are only locally stabilizable, especially
in presence of input saturation. Therefore, we wish to study 
stabilization inside a compact set $S$. Let $int(R)$ be the interior
of set $R$. We consider 
a compact and connected set $S \subset X$ where the origin 
$\vzero \in int(S)$ is the state we seek to stabilize to. Furthermore, we
restrict the set $S$ to be a basic semi-algebraic set defined by a
conjunction of polynomial inequalities:
\[ S: \{ \vx \in \reals^n\ |\ p_{S,1}(\vx) \leq 0, \ldots, p_{S,k}(\vx) \leq 0 \} \,.\]

The stabilization problem can be reduced to the problem of
finding a local CLF $V$ which respect the following constraints

\begin{equation}\label{eq:local-clf-def} 
\begin{array}{rl}
& V(\vzero) = 0 \\
  (\forall \vx \in S \setminus \{\vzero\}) \ & V(\vx) > 0 \\
	(\forall \vx \in S \setminus \{\vzero\}) \ (\exists \vu \in U)\ & \nabla V \cdot f(\vx, \vu) < 0 \,. \\ 
\end{array}
\end{equation}
Given a function $V$ and a comparison predicate $\Join \in \{ =, \leq, <, \geq, > \} $, we define $V^{\Join \beta}$ as the set:
\[ V^{\Join \beta} = \{\vx | V(\vx) \Join \beta \} \,. \]
Let $\beta^*$ be maximum $\beta$ s.t. $V^{\leq \beta} \subseteq S$.
Having a CLF $V$, it guarantees that there is a strategy to keep the state
inside $V^{< \beta}$, and stabilize to the origin (Fig.~\ref{fig:clf}). 

\begin{theorem}
	Given a control affine system $\Psi$, where $U : \reals^m$ 
	and a polynomial control Lyapunov function $V$ satisfying Eq.~\eqref{eq:local-clf-def}, there is a feedback function $\K$ for which if $\vx_0 \in V^{< \beta^*}$, then:
	\begin{enumerate}
		\item $(\forall t \geq 0) \ \vx(t) \in S$
		\item $(\forall \epsilon > 0) \ (\exists T \geq 0) \ \norm{\vx(T) - \vzero} < \epsilon$\,.
	\end{enumerate}
\end{theorem}
\begin{proof}
	First, using Sontag results, 
	there exists a feedback function $\K^*$ s.t. while $\vx \in S$, then 
	$\frac{dV}{dt} = \nabla V \cdot f(\vx, \vu) < 0$~\cite{sontag1989universal}. Assuming $\vx(0) = \vx_0 \in V^{<\beta^*} \subset S$, then initially $V(\vx(0)) < \beta^*$. Now, assume the state reaches $\partial S$ at time $t_2$. By continuity, there is a time $t_1 \leq t_2$
		s.t. $\vx(t_1) \in \partial (V^{<\beta^*})$ and $(\forall t \in [0, t_1]) \ \vx(t) \in S$. Thus, $V(\vx(t_1)) = \beta^*$ and
	\[
	V(\vx(t_1)) = \left(V(\vx(0)) + \int_{0}^{t_1} \frac{dV}{dt} dt\right) < V(\vx(0)) \,.
	\]
	This means $V(\vx(t_1)) < \beta^*$, which is a contradiction. Therefore, the state never reaches $\partial S$ and remains in $int(S)$ forever.
	
	$V$ would be a Lyapunov function for
	the closed loop system when the control unit is replaced with the feedback function $\K^*$ and using standard results in Lyapunov theory
	$(\forall \epsilon > 0) \ (\exists T \geq 0) \ ||\vx(T) - 0|| < \epsilon$.
\end{proof}
Finding a local CLF is similar to finding a global one. One only needs to
consider set $S$ in the formulation. The observation set would consists of
$(\vx_i, \vu_i)_{i=1}^j$ where $\vx_i$ is inside $S$ and the verifier would
check the following conditions:
\begin{align*}
	(\exists \vx \neq \vzero)& \bigwedge_{i=1}^k p_{S,i}(\vx) \leq 0 \land V(\vx) \geq 0 \\
	(\exists \vx \neq \vzero)& \bigwedge_{i=1}^k p_{S,i}(\vx) \leq 0 \land 
	(\forall \vu \in U) \ \nabla V \cdot f(\vx, \vu) \geq 0 \,,
\end{align*}
which is as hard as the one solved in Section.~\ref{sec:verifier}.

\begin{lemma} \label{lem:completeness}
	Assuming (i) the demonstrator function $\D$ is smooth, (ii) the closed loop system with feedback law $\D$ is exponentially stable over a bounded region $S$, then there exists a local polynomial CLF, compatible with $\D$.
\end{lemma}
\begin{proof}
	Under assumption (i) and (ii), one can show that a polynomial local Lyapunov function $V$ (not control Lyapunov function) exists for the closed loop system $\Psi(X, U, f , \D)$~\cite{peet2008polynomial}:
	\[
	V(\vzero)=0 \ \land
	(\forall \vx \in S \setminus \vzero) \left( \begin{array}{c}
		V(\vx) > 0 \\
		\nabla V \cdot f(\vx, \D(\vx)) < 0
	\end{array} \right) \,.
	\]
	 This means that $V$ is compatible with the demonstrator. $V$ is also a local CLF as it satisfies Eq.~\eqref{eq:local-clf-def}.
\end{proof}

As mentioned, the learning framework fails when the basis functions are not expressive to capture a CLF compatible with the demonstrator and one needs to update the demonstrator and/or the set of basis functions. However, if one believes that the demonstrator satisfies the conditions in Lemma~\ref{lem:completeness}, then, success of the learning procedure is guaranteed, provided the set of basis functions is rich enough. 

%
%

\subsection{Barrier Certificate}
Barrier certificates are used to guarantee safety properties for the
system.  More specifically, given compact and connected semi-algebraic
sets $S$ (safe) and $I$ (initial) s.t. $I \subset int(S)$, the
  overall goal is to ensure that whenever $\vx(0) \in I$, we have
  $\vx(t) \in S$ for all time $t \geq 0$. The sets $S,I$ are expressed
  as semi-algebraic sets of the following form:
\begin{align*}
	S: \{ \vx \in \reals^n\ |\ p_{S,1}(\vx) \leq 0, \ldots, p_{S,k}(\vx) \leq 0 \}\\
	I: \{ \vx \in \reals^n\ |\ p_{I,1}(\vx) \leq 0, \ldots, p_{I,l}(\vx) \leq 0 \}\,.
\end{align*} 

The safety problem can be reduced to the problem of
finding a (relaxed~\cite{prajna2004safety}) control barrier certificate $B$ which respect 
the following constraints~\cite{WIELAND2007462}:
\begin{equation}\label{eq:barrier-cert-def} 
\begin{array}{rl}
(\forall \vx \in I) \ & B(\vx) < 0 \\
(\forall \vx \not\in int(S)) \ & B(\vx) > 0 \\
(\forall \vx \in S \setminus int(I)) \ (\exists \vu \in U)\ & \nabla B \cdot f(\vx, \vu) < 0 \,. \\ 
\end{array}
\end{equation}
To find such a barrier certificate, one needs to define $B$ as a linear
combination of basis functions and use the framework to find a correct $B$.
The verifier would check the following conditions that negate each of the
conditions in Eq.~\eqref{eq:barrier-cert-def}.
First we check if there is a $\vx \in I$ such that $B(\vx) \geq 0$.
\[
 (\exists \vx)\ \ \bigwedge_{j=1}^l p_{I,j}(\vx) \leq 0 \ \land\ B(\vx) \geq 0\,.
\]

Next, we check if there exists a $\vx \not \in int(S)$ such that $B(\vx) \leq 0$. Clearly, if $\vx \not\in int(S)$, we have $p_{S,i}(\vx) \geq 0$ for at least one $i \in \{1,\ldots,k\}$. This yields $k$ conditions of the form:
\[
  (\exists \vx) \ p_{S,i}(\vx) \geq 0 \land B(\vx) \leq 0,\ i \in \{ 1, \ldots, k\}\,.
\]
Finally, we ask if $\exists \vx \in S \setminus int(I)$ that violates the decrease condition. Doing so, we obtain $l$ conditions. For each $i \in \{ 1, \ldots, l \}$, we solve
\begin{align*}
(\exists \vx) & \ \underset{\vx \not\in int(I)}{\underbrace{p_{I,i}(\vx) \geq 0}} \land\ \underset{\vx \in S}{\underbrace{\bigwedge_{j=1}^k p_{S,j}(\vx) \leq 0}} \\
	     &\ \land (\forall \vu \in U) \ \nabla B \cdot f(\vx, \vu) \geq 0 \,,
\end{align*}

Overall, we have $1 + k + l$ different checks. If any of these checks result
in $\vx$, it serves as a counterexample to the conditions for a barrier function
~\eqref{eq:barrier-cert-def}.

As before, we choose basis functions $g_1, \ldots, g_r$ for the barrier
set $B_\vc: \sum_{k=1}^r c_k g_k(\vx)$.
Given observations set $O_j: \{ (\vx_1, \vu_1), \ldots, (\vx_{j}, \vu_{j})\}$, the corresponding
 candidate set $C_j$ of observation compatible barrier functions
is defined as the following:
\[
C_j: \left\{ \vc |\hspace{-0.1cm}
\begin{array}{l}
\bigwedge\limits_{(\vx_i, \vu_i) \in O_j}
\left(\begin{array}{rl} 
\vx_i \in I \rightarrow & B_\vc(\vx_i) < 0 \ \land \\
\vx_i \not\in int(S) \rightarrow & B_\vc(\vx_i) > 0 \ \land \\
\vx_i \in S\setminus int(I)& \\ \rightarrow \nabla B_\vc & . f(\vx_i, \vu_i)<0 \end{array}\right)\end{array}\right\}.
\]
The LHS of the implication for each observation $(\vx_i, \vu_i)$ is evaluated
 and the RHS constraint is added only when the LHS holds. Nevertheless, $\overline{C_j}$ remains a polytope similar to Lemma.~\ref{lemma:cj-convex}.

\begin{remark}
  For the original control barrier certificates, it is sufficient to
  check whether $B$ can be decreased on the boundary ($B^{=0}$).  The
  relaxed version of control barrier certificates is introduced by
  Prajna et al.~\cite{prajna2004safety} using sum of squares (SOS)
  relaxation. Here we use this relaxation to simplify the candidate
  generation process. However, for the verification process this
  relaxation is not needed and without any complication, one could
  verify the original conditions as opposed to the relaxed ones. This
  trick will improve the precision of the method.
\end{remark}

\subsection{Reach-While-Stay}
In this problem, the goal is to reach a target set $T$ from an initial
set $I$, while staying in a safe set $S$, wherein
  $I \subseteq S$. The set $S$ is assumed to be compact.  By
combining the local Lyapunov function and a barrier certificate, one
can define a smooth, Lyapunov-like
function $V$, that satisfies the following conditions
(see~\cite{Ravanbakhsh-Others/2016/Robust}):

\begin{equation}\label{eq:lyapunov-like-def} 
  \begin{array}{lrl}
C1:&(\forall \vx \in I) & V(\vx) < 0 \\
C2:&(\forall \vx \not\in int(S)) & V(\vx) > 0 \\
C3:&(\forall \vx \in S \setminus int(T)) (\exists \vu \in U)& \nabla V \cdot f(\vx,\vu)\hspace{-0.05cm}<\hspace{-0.05cm}0. \\ 
\end{array}
\end{equation}

 We briefly sketch the argument as to why such a Lyapunov-like
  function satisfies the reach-while-stay, referring the
  reader to our earlier work on control certificates for a detailed
  proof~\cite{Ravanbakhsh-Others/2016/Robust}.  Suppose we have found
  a function $V$ satisfying~\eqref{eq:lyapunov-like-def}.  $V$ is
  strictly negative over the initial set $I$ and strictly positive
  outside the safe set $S$.  Furthermore, as long as the flow remains
  inside the set $S$ without reaching the interior of the target $T$,
  there exists a control input at each state to strictly decrease the
  value of $V$.  Combining these observations, we conclude either (a)
  the flow remains forever inside set $S \setminus int(T)$ or (b) must visit the
  interior of set $T$ (before possibly leaving $S$). However, option (a) is ruled
  out because $S \setminus int(T)$ is a compact set and $V$ is a continuous
  function. Therefore, if the flow were to remain within $S \setminus int(T)$ forever
  then $V(\vx(t)) \rightarrow -\infty$ as $t \rightarrow \infty$,
  which directly contradicts the fact that $V$ must be lower bounded
  on a compact set $S \setminus int(T)$. We therefore, conclude that the flow must stay
  inside $S$ and eventually visit the interior of the target $T$.  

The learning framework extends easily to search for a function $V$
that satisfies the constraints in Eq.~\eqref{eq:lyapunov-like-def}.

\subsection{Finite-time Reachability}
The idea of funnels has been developed to use the Lyapunov argument
for finite-time reachability~\cite{mason1985mechanics}.
Then, following Majumdar et al., a library of
control funnels can provide building blocks for motion
planning~\cite{majumdar2013robust}. Likewise, control funnels are used to
reduce reach-avoid problem to timed automata~\cite{bouyer2017timed}.

  In this section, we consider Lyapunov-like functions for establishing
  control funnels.  Let $I$ be a set of initial states for the plant ($\vx(0) \in I$),
  and $T$ be the target set that the system should reach at time $\T > 0$ ($\vx(\T) \in int(T)$).
  Let $S$ be the safe set, such that $I, T \subseteq S$
  and $\vx(t) \in S$ for time $t \in [0,\T]$. The goal is to find a controller
  that guarantees that whenever $\vx(0) \in I$, we have $\vx(t) \in S$ for
  all $t \in [0, \T]$ and $\vx(\T) \in int(T)$. To solve this, we search instead
  for a control Lyapunov-like function $V(\vx,t)$ that is a function of
  the state and time, with the following properties:

\begin{equation}\label{eq:c-funnel-def} 
\begin{array}{lrl}
C1:& (\forall \vx \in I) & V(\vx, 0) < 0 \\
C2:& (\forall \vx \not\in int(T)) & V(\vx, \T)\ > \ 0 \\
C3:& \left(\forall \begin{array}{l}t \in [0, \T] \\ 
		\vx \not\in int(S) \end{array}\right) & V(\vx, t) > 0 \\
C4:&\left(\begin{array}{l}\forall t \in [0, \T]\\ 
		\forall \vx \in S\end{array}\right) (\exists \vu \in U)& \dot{V}(t, \vx, \vu) < 0 \,,\\ 
\end{array}
\end{equation}
where $\dot{V}(t, \vx, \vu) = \frac{\partial V}{\partial t} + \nabla V \cdot f(\vx, \vu)$.
First of all, when initialized to $\vx(0) \in I$, we have $V(\vx,0) < 0$ by condition
  C1. Next, the controller's action through condition $C4$ guarantees that $\frac{dV}{dt} < 0$ over the
  trajectory for $t \in [0, \T]$, as long as $\vx \in S$. Through $C3$, we can guarantee that $\vx(t) \in S$ for
  $t \in [0,\T]$. Finally, it follows that $V(\vx(\T),\T) < 0$. Through $C2$, we conclude that
  $\vx \in int(T)$. 
As depicted in Fig.~\ref{fig:funnel}, the set $V^{=0}$ forms a barrier, and set $V^{<0}$ 
forms the required funnel, while $t \leq \T$.
\begin{figure}[t]
\begin{center}
	\includegraphics[width=0.48\textwidth]{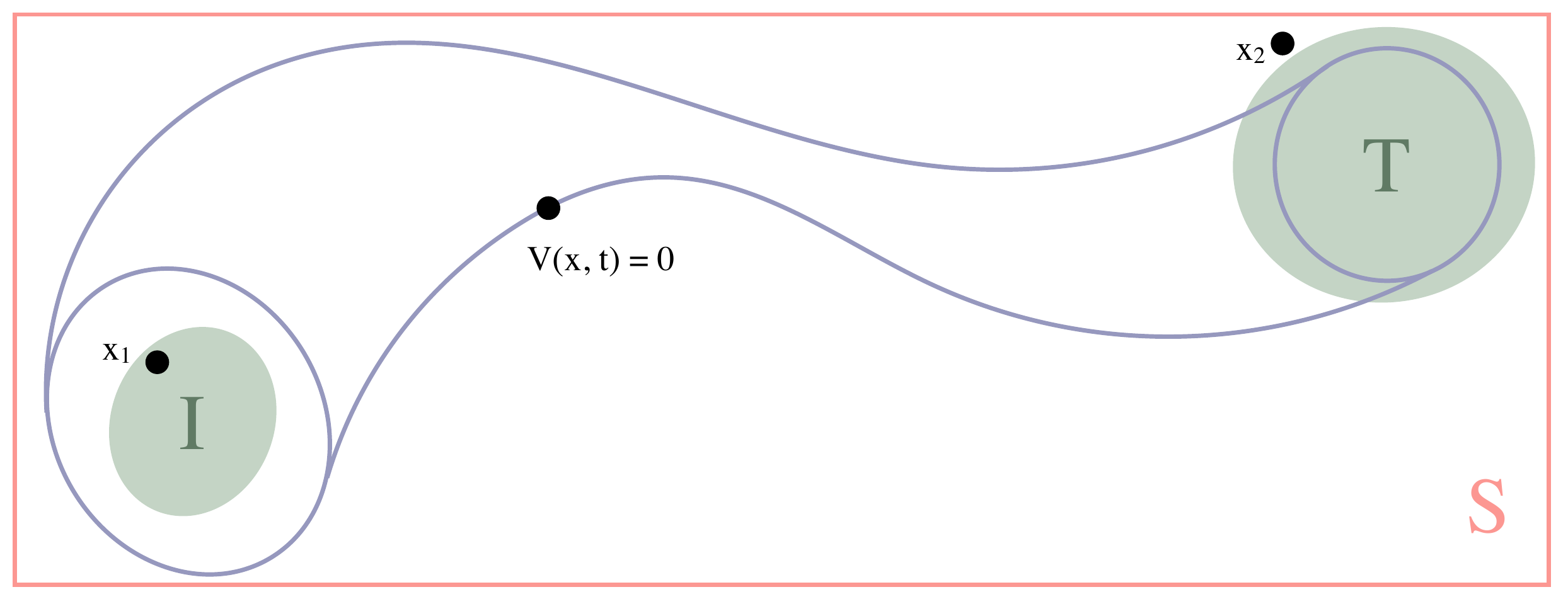}
\end{center}
\caption{A schematic view of a control funnel. Blue lines show the boundary of
the funnel $V(\vx, t) = 0$. Also, initially $V(\vx_1, 0) < 0$ and at the end of
horizon, $V(\vx_2, \T) > 0$.}\label{fig:funnel} 
\end{figure}

\begin{theorem}
	Given compact semi-algebraic sets $I$, $S$, $T$, a time horizon $\T$, and a smooth function $V$ satisfying Eq.~\eqref{eq:c-funnel-def}, there exists a control
	strategy s.t. for all traces of the closed loop system, if $\vx(0) \in I$, then
	\begin{enumerate}
		\item $(\forall t \in [0, \T]) \ \vx(t) \in S$
		\item $\vx(\T) \in int(T)$.
	\end{enumerate}
\end{theorem}
\begin{proof}
	Using Sontag result~\cite{sontag1989universal,WIELAND2007462}, there is a feedback $\K$ which decreases
	value of $V$ while $t \in [0, \T]$ and $\vx \in S$:
	\[
	(\forall t \in [0, \T], \vx \in S) \ \dot{V}(t, \vx, \K(\vx)) < 0\,.
	\]
	 Now, assume $\vx(0) \in I$.
	By the first condition of Eq.~\eqref{eq:c-funnel-def}, $V(\vx(0), 0) < 0$.
	Assume there is a time $t \in [0, \T]$ s.t. $\vx(t) \not\in S$. By compactness of $S$,
	and smooth dynamics, there is a time $t_2$ s.t. $V(\vx(t_2), t_2) \in \partial S$
	and for all $t < t_2$, $\vx(t) \in int(S)$. According to the third condition of 
	Eq.~\eqref{eq:c-funnel-def}, $V(\vx(t_2), t_2) > 0$. Since $V$ is a smooth function
	there is a time $t_1$ ($0 < t_1 < t_2$) s.t. $V(\vx(t_1), t_1) = 0$ and for all 
	$t < t_1$, $V(\vx(t), t) \in S$. By the fourth condition in Eq.~\eqref{eq:c-funnel-def}:
	\begin{align*}
	V(\vx(t_1), t_1) &= V(\vx(0), 0) + \int_0^{t_1} \dot{V}(t, \vx(t), \K(\vx(t))) \\	
					 &< V(\vx(0),0) < 0 \, .
	\end{align*}
	This is a contradiction and therefore, for all $t \in [0, \T]$, $\vx(t) \in S$.
	And similar to the argument above, it is guaranteed that for all $t \in [0, \T]$,
	$V(\vx(t),t) < 0$. By the second condition of Eq.~\eqref{eq:c-funnel-def}, it
	is guaranteed that if $\vx(\T) \not\in int(T)$, then $V(\vx(\T), \T) > 0$. 
	Therefore, $\vx(\T) \in int(T)$.
\end{proof}

Using the Lyapunov-like conditions~\eqref{eq:c-funnel-def}, the
problem of finding such control funnels (respecting
Eq.~\eqref{eq:c-funnel-def}) belongs to the class of problem which
could be solved with our method.

\section{Experiments}\label{sec:expr}
In this section, we describe numerical results on some case studies.
We first describe our implementation of the techniques
  described thus far.  The verifier component is implemented using
tool Gloptipoly~\cite{henrion2009gloptipoly}, which in turn uses
Mosek to solve SDP problems~\cite{mosek2010mosek}, and only needs a
degree of relaxation $D$ as its input. For the demonstrator, a
nonlinear MPC scheme is used, which is solved using a gradient descent
algorithm.  For each benchmark, the following parameters are tuned
to obtain the cost function:
\begin{enumerate}
	\item time step $\tau$
	\item number of horizon steps $N$
	\item $Q$, $R$, and $H$ for the cost function:
	\[
	\begin{array}{c}
	\left( \sum_{i = 1}^{N-1} \vx(i\tau)^t \ Q \ \vx(i\tau) + \vu(i\tau)^t \ R \ \vu(i\tau) \right) \\	+ \vx(N\tau)^t \ H \ \vx(N\tau)\,.
	\end{array}
	\]
      \end{enumerate}
As such, an MPC cost function is designed to enforce a specification such as
stability or reaching a target set. However, since the approach provides no
guarantees, we run  hundreds of simulations of the closed loop system starting 
from randomly selected initial states to check whether the specifications are
met. Failing this, the cost function is adjusted, repeating the testing process.
And finally, for the learner, quadratic polynomials
are used as candidates for the desired Lyapunov-like functions. Nevertheless,
more complicated polynomials are also supported by our implementation.
Beside these inputs, each control problem has a specification.
For example, for a \emph{reach-while-stay} problem,
the target set $T$, initial set $I$, and safe set $S$ are provided
as inputs.

All the computations
reported in this section were performed on a Mac Book Pro with 2.9 GHz
Intel Core i7 processor and 16GB of RAM. The reported CLFs are rounded
to two decimal points. The implementation is available upon request.

\subsection{Case Study I:} This system is two-wheeled mobile robot
modeled with five states $[x, y, v, \theta, \gamma]$ and two control
inputs~\cite{francis2016models}, where $x$ and $y$ define the position
of the robot, $v$ is its velocity, $\theta$ is the rotational
position and $\gamma$ is the angle between the front and rear
  axles.  The goal is to stabilize the robot to a target velocity
$v^*=5$, and $\theta^* = \gamma^* = y^* = 0$ as shown in Fig.~\ref{fig:bicycle}. The dynamics of the model is as
follows:
\begin{equation*}\label{ex:bicycle-dyn}
	\left[ \begin{array}{l}
		\dot{x} \\
		\dot{y} \\ \dot{v} \\ \dot{\theta} \\ \dot{\sigma}
	\end{array}\right] = 
	\left[ \begin{array}{l}
		v\cos(\theta) \\
		v\sin(\theta) \\ u_1 \\ v\sigma \\ u_2
	\end{array} \right] \,,
\end{equation*}
where $\sigma = tan(\gamma)$ (see Fig.~\ref{fig:bicycle}). 
Variable $x$ is immaterial in
the stabilization problem and is dropped to obtain a model with 
four state variables $[y, v, \theta, \sigma]$. Also, the sine function is approximated with a
polynomial of degree one. 
The inputs are saturated over the intervals $U: [-10, 10]\times[-10, 10]$, and
the specification is reach-while-stay, provided by the following sets
\[
\begin{array}{rl}
		S: &[-2, 2]\times[3, 7]\times[-1, 1]\times[-1, 1] \\
		I: &\B_{0.4}(\vzero) \\
		T: &\B_{0.1}(\vzero) \,.
\end{array}
\]
\begin{figure}[t]
\begin{center}
	\includegraphics[width=0.4\textwidth]{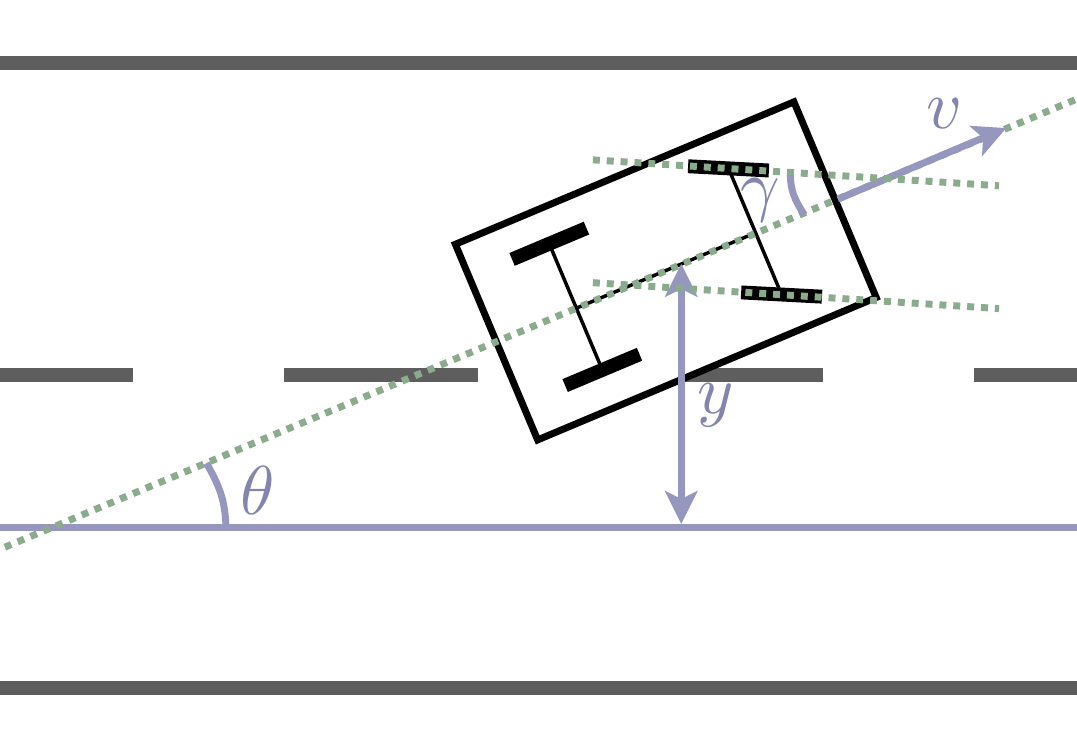}
      \end{center}
\caption{A schematic view of the bicycle model.}\label{fig:bicycle} 
\end{figure}
The method finds the following CLF:
\begin{align*}
V =
&0.37 y^2
+ 0.52 y\theta
+ 3.11 \theta^2
+ 0.98 y\sigma
+ 2.23 \sigma\theta +\\
&4.46 \sigma^2
- 0.36 vy
- 0.29 v\theta
+ 0.95 v\sigma
+ 3.86 v^2\,.
\end{align*}

This CLF is used to design a
  controller. Fig.~\ref{fig:bicycle-sim} shows the projection of
  trajectories on to $x$-$y$ plane for the synthesized controller in red.  The
  blue trajectories are generated using the MPC controller that served as the
  demonstrator. The
  behavior of the system for both controllers are similar but not
  identical. Notice that the initial state in
  Fig.~\ref{fig:bicycle-sim}(c) is not in the region of attraction
  (guaranteed region). Nevertheless, the CLF-based controller can
  still stabilize the system while keeping the system in the safe
  region. On the other hand, the MPC violates the safety constraints
  even when the safety constraints are imposed in the MPC
  scheme. The safety is violated because in the beginning $\theta$
  gets larger than $1$ and it gets close to $\pi/2$ (the robots moves
  almost vertically).

\begin{figure}[t]
\begin{center}
	\includegraphics[width=0.48\textwidth]{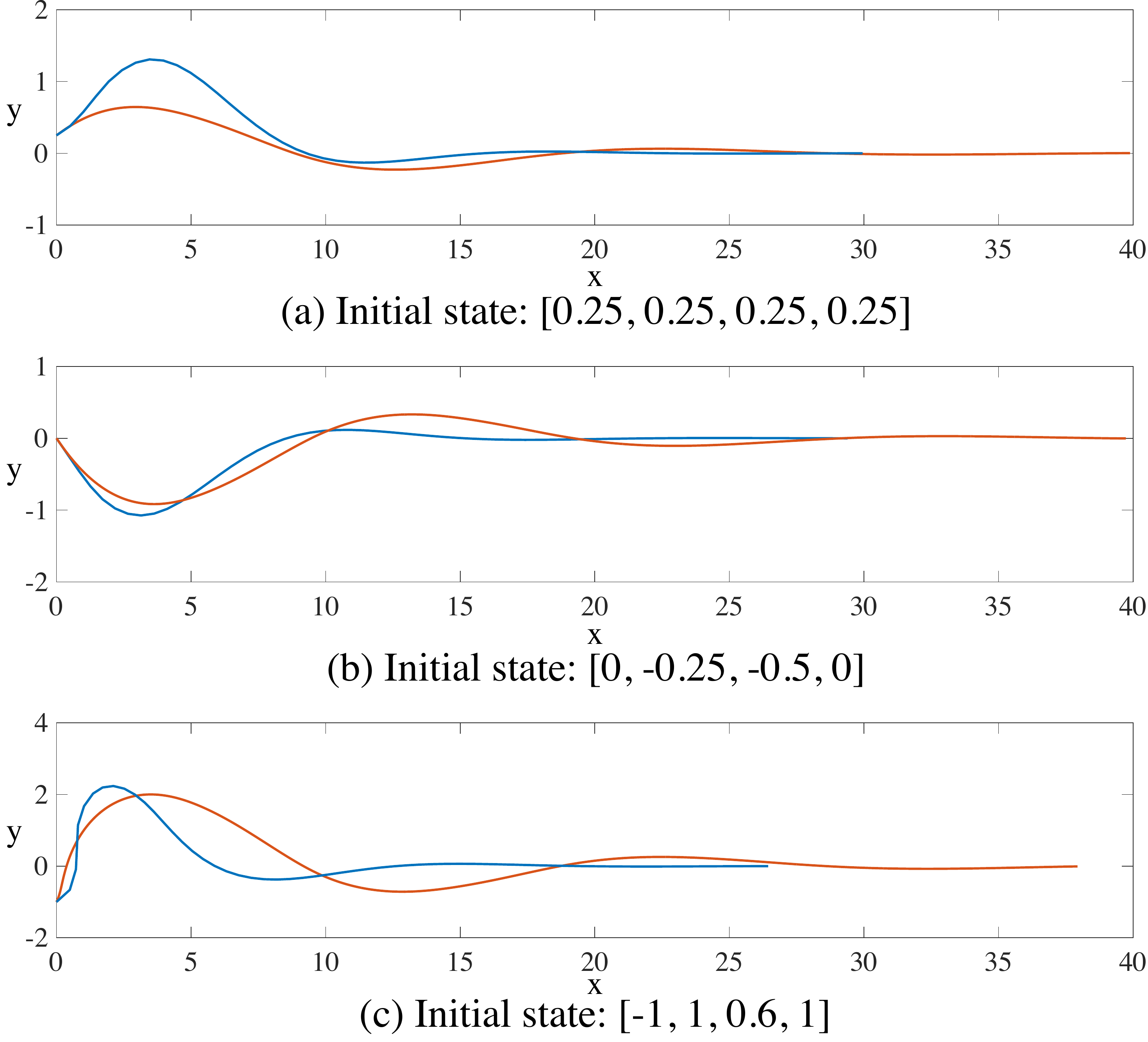}
      \end{center}
\caption{Simulation for the bicycle robot - Projected on x-y
 plane. Simulation traces are plotted for three different initial states. Blue (red) traces corresponds to trajectories of the system for MPC controller (CLF-based controller).}\label{fig:bicycle-sim} 
\end{figure}

\subsection{Case Study II:}
The problem of keeping the inverted pendulum in a vertical 
position is considered. This case study has applications in balancing 
two-wheeled robots~\cite{CHAN201389}. The system has two degrees
of freedom: the position of the cart $x$, and the degree
of the inverted pendulum $\theta$. The goal is to keep 
the pendulum in a vertical position by moving the cart with
input $u$ (Fig.~\ref{fig:inverted-pendulum}).

The system has four state variables 
$[x, \dot{x}, \theta, \dot{\theta}]$ with the following 
dynamics~\cite{landry2005dynamics}:
\begin{equation*}\label{eq:inverted-pendulum-dyn}
	\left[ \begin{array}{l}
		\ddot{\vx} \\ \ddot{\theta}
	\end{array}\right] = 
	\left[ \begin{array}{l}
		\frac{4u - 4\epsilon\dot{x} + 4ml\dot{\theta}^2 \sin(\theta) - 3mg\sin(\theta)\cos(\theta)}{4(M+m)-3m\cos^2(\theta)} \\ \frac{ (M+m)g\sin(\theta) - (u - \epsilon\dot{x}) \cos(\theta) - ml\dot{\theta}^2\sin(\theta)\cos(\theta)}{l(\frac{4}{3}(M+m)-m\cos(\theta)^2)}
	\end{array} \right] \,,
\end{equation*}
where $m = 0.21$ and $M=0.815$ are masses of the pendulum and the cart
respectively, $g=9.8$ is the gravitational acceleration, and $l=0.305$ 
is distance of center of mass of the pendulum from the cart. After partial
linearization, the dynamics have the following form:
\begin{equation*}\label{eq:inverted-pendulum-dyn}
	\left[ \begin{array}{l}
		\ddot{\vx} \\ \ddot{\theta}
	\end{array}\right] = 
	\left[ \begin{array}{l}
		4u + \frac{4(M+m)g \tan(\theta) - 3mg\sin(\theta)\cos(\theta)}{4(M+m)-3m\cos^2(\theta)} \\ \frac{- 3 u \cos(\theta)}{l}
	\end{array} \right] \,.
\end{equation*}

The trigonometric and rational functions are approximated with polynomials of degree three. The input is saturated $U:[-20, 20]$ and sets for
a safety specification are $S: [-1, 1]^4 , \ I:\B_{0.1}(\vzero)$.

Fig.~\ref{fig:inverted-sim} shows the some of the traces of the closed loop system 
for the CLF-based controller as well as the MPC controller. Notice that the trajectories of the CLF based
controller  are quite distinct from the MPC, especially in regions where the demonstration is not provided during the CLF synthesis process. For example, in Figure.~\ref{fig:inverted-sim}(b), the behaviors of these controllers are similar
outside the initial set $I$. However, inside $I$ (near the equilibrium) the 
behavior is different, since the demonstrations are only generated for states outside $I$.
The CLF-based controller is designed using the following CLF generated by the learning framework:

\begin{align*}
V =& 16.37 \dot{\theta}^2 + 50.37 \dot{\theta}\theta
+ 75.16 \theta^2 + 13.51 x\dot{\theta} 
+ 43.26 x\theta + \\
& 10.44 x^2 + 23.30 \dot{\theta}\dot{x} + 38.09 \dot{x}\theta
+ 11.13 \dot{x}x + 9.55 \dot{x}^2 \,.
\end{align*}

\begin{figure}[t]
\begin{center}
	\includegraphics[width=0.3\textwidth]{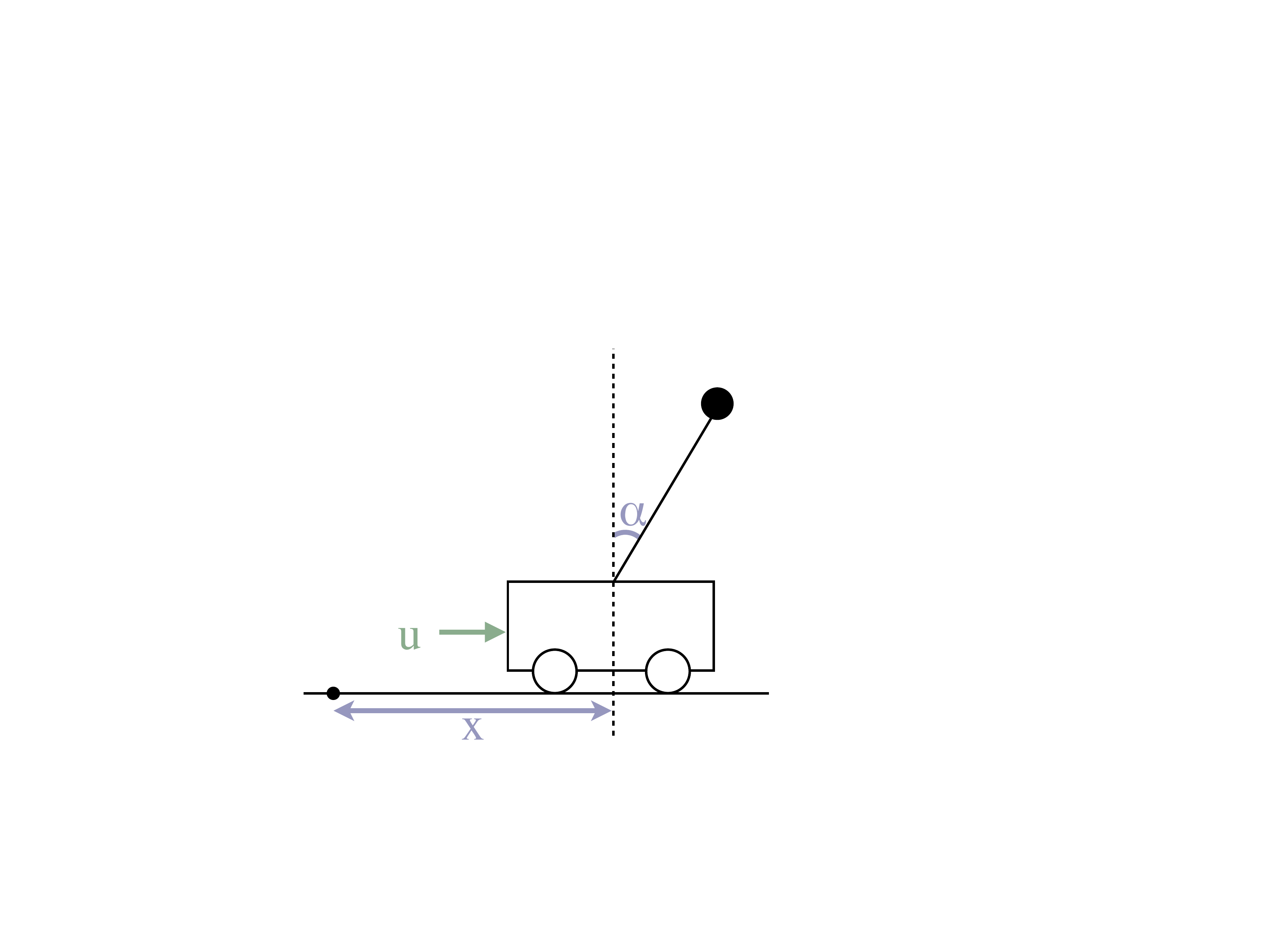}
\end{center}
\caption{A schematic view of the ``inverted pendulum on a cart".}\label{fig:inverted-pendulum} 
\end{figure}
\begin{figure*}[t]
\begin{center}
	\includegraphics[width=1\textwidth]{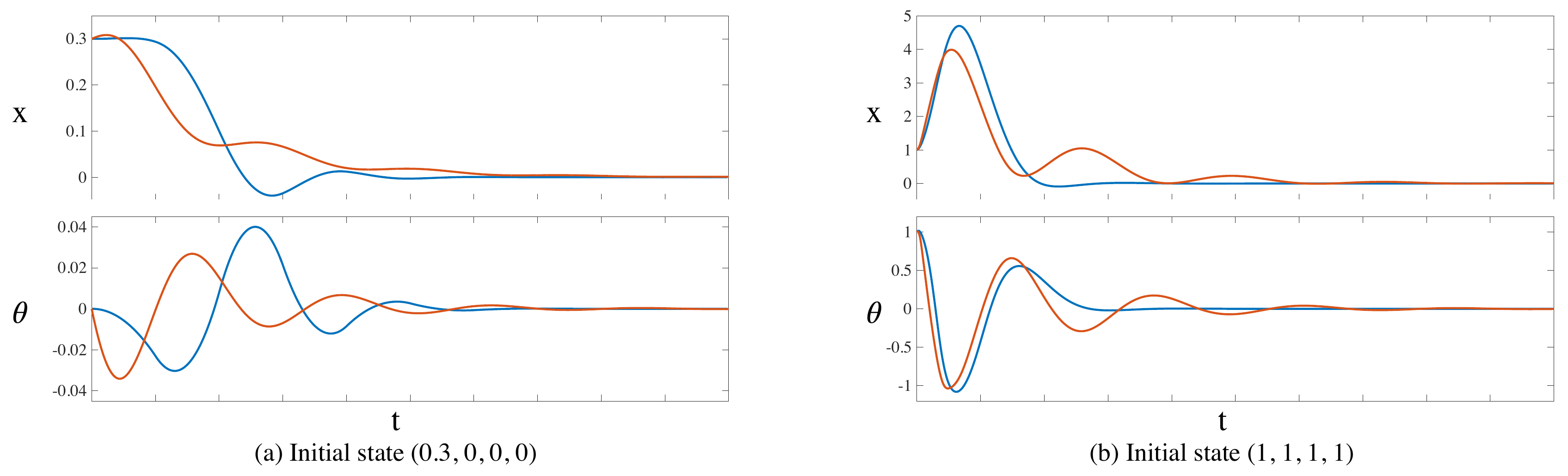}
\end{center}
\caption{Simulation for the inverted pendulum system. Simulation traces are plotted 
for two initial states. Red (blue) traces show the simulation for the CLF-based (MPC) controller.}\label{fig:inverted-sim} 
\end{figure*}

\subsection{Case Study III:}
Caltech ducted fan has been used to study the aerodynamics of a single 
wing of a thrust vectored, fixed wing aircraft~\cite{jadbabaie2002control}.  
In this case study, we wish to
design forward flight control in which the angle of attack needs to be
set for a stable forward flight. The model of the system is carefully 
calibrated through wind tunnel experiments. The system has four states:
$v$ is the velocity, $\gamma$ is the moving direction the ducted fan,
$\theta$ is the rotational position, and $q$ is the angular velocity.
The control inputs are the thrust $u$ and the angle at which the
thrust is applied $\delta_u$ (Fig.~\ref{fig:ducted-fan}). 
Also, the inputs are saturated: 
$U :[0, 13.5] \times [-0.45, 0.45]$.
The dynamics are:
\begin{equation*}\label{ex:ducted-fan-forward-dyn}
	\left[ \begin{array}{l}
		m \dot{v} \\ m v \dot{\gamma} \\ \dot{\theta} \\ J \dot{q}
	\end{array}\right] = 
	\left[ \begin{array}{l}
		-D(v, \alpha) - W \sin(\gamma) + u \cos(\alpha + \delta_u) \\
		L(v, \alpha) - W \cos(\gamma) + u \sin(\alpha + \delta_u) \\
		q \\ M(v, \alpha) - u l_T \sin(\delta_u)
	\end{array} \right] \,,
\end{equation*}
where the angle of attack $\alpha = \theta - \gamma$, and 
$D$, $L$, and $M$ are polynomials in $v$ and $\alpha$. 
For full list of parameters, see~\cite{jadbabaie2002control}. 
According to the dynamics, $\vx^*:\ [6, 0, 0.1771, 0]$ is a stable
equilibrium (for $\vu^*:\ [3.2, -0.138]$) where the ducted fan can move forward with velocity $6$.
Thus, the goal is to reach near $\vx^*$. 
The system is not affine in control. We replace $u$ and $\delta_u$ 
with $u_s = u \sin(\delta_u)$ and $u_c = u \cos(\delta_u)$:
\begin{equation*}\label{ex:ducted-fan-forward-dyn}
	\left[ \begin{array}{l}
		\dot{v} \\ \dot{\gamma} \\ \dot{\theta} \\ \dot{q}
	\end{array}\right] = 
	\left[ \begin{array}{l}
		\frac{-D(v, \alpha) - W \sin(\gamma) + u_c \cos(\alpha) -  u_s \sin(\alpha)}{m} \\
		\frac{L(v, \alpha) - W \cos(\gamma) + u_c \sin(\alpha) + u_s \cos(\alpha)}{mv} \\
		q \\ \frac{M(v, \alpha) - l_T u_s}{J}
	\end{array} \right] \,.
\end{equation*}
Projection of $U$ into the new coordinate will yield a sector of a circle.
Then, set $U$ is safely under-approximated by a polytope $\hat{U}$ as shown in Fig.~\ref{fig:uhat}.
\begin{figure}
\begin{center}
	\includegraphics[width=0.3\textwidth]{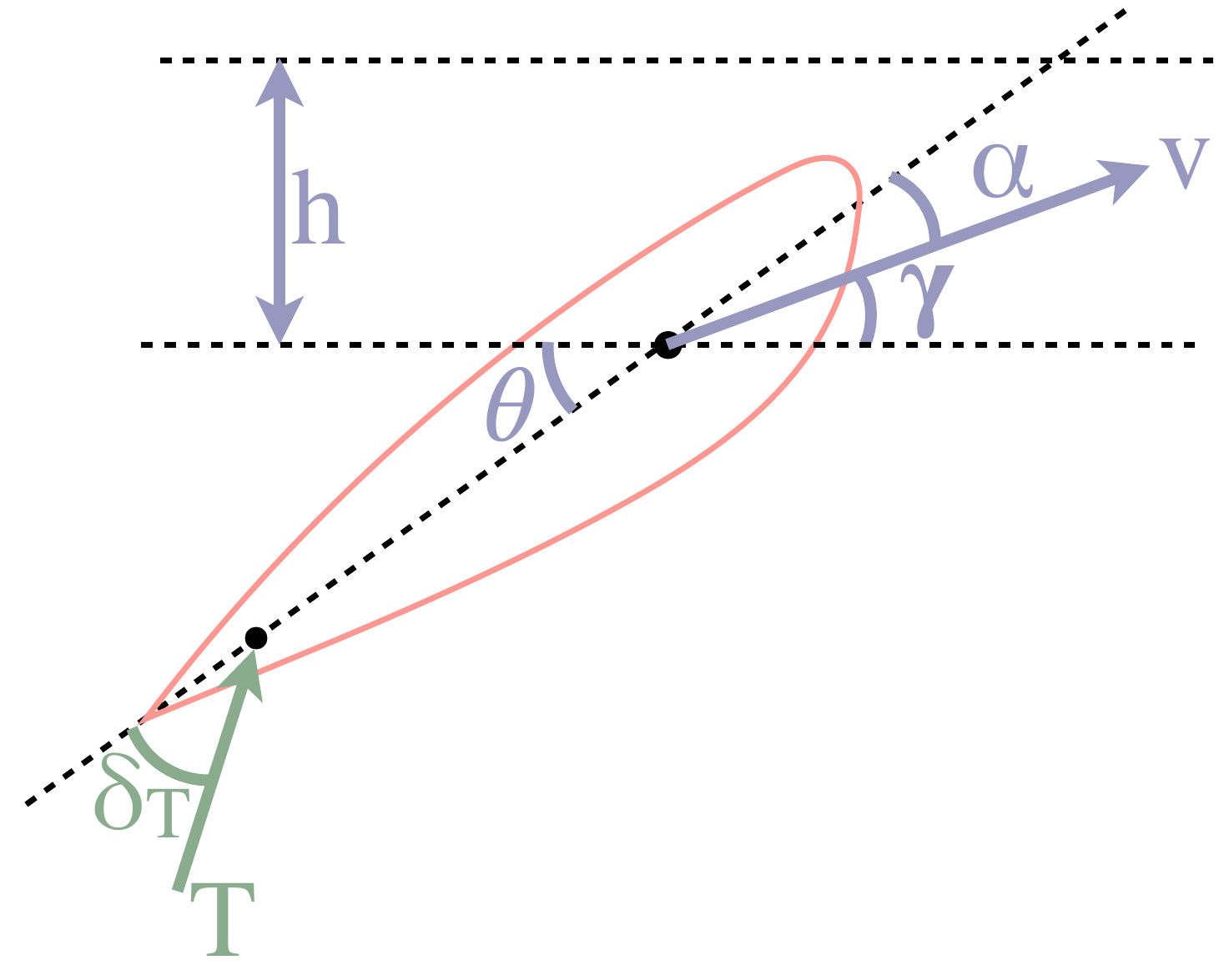}
\end{center}
\caption{A schematic view of the Caltech ducted fan.}\label{fig:ducted-fan} 
\end{figure}
\begin{figure}
\begin{center}
	\includegraphics[width=0.3\textwidth]{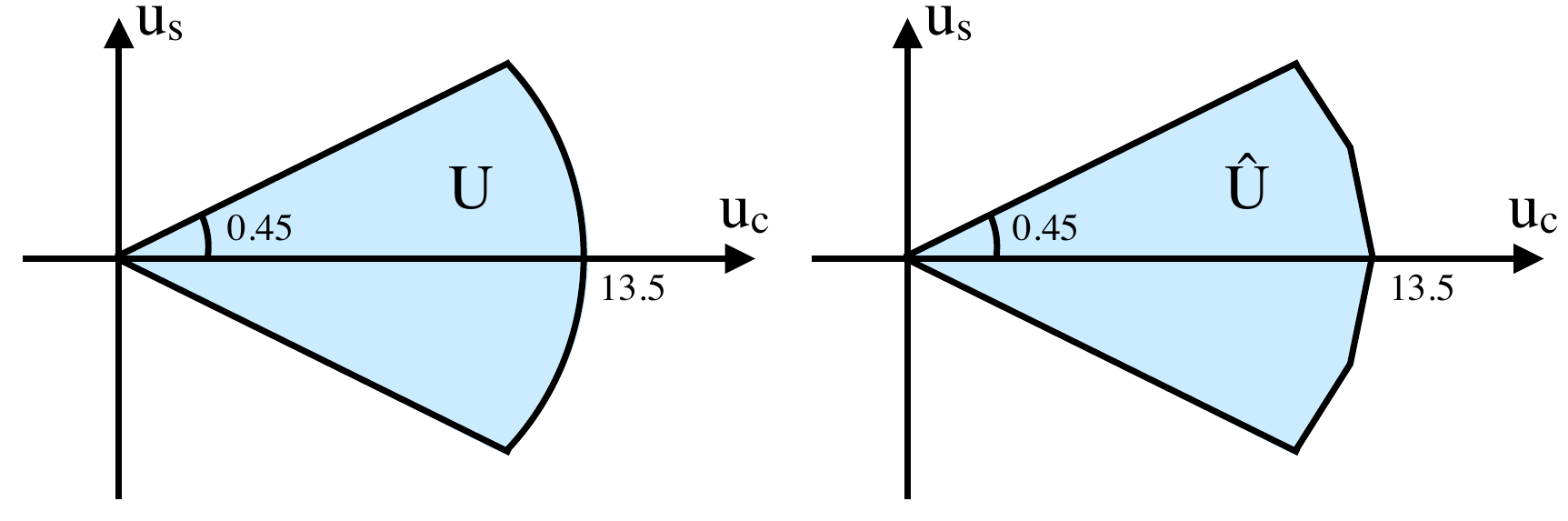}
\end{center}
\caption{Set of feasible inputs $U$ and its under approximation $\hat{U}$ in the new coordinate for case study III.}\label{fig:uhat} 
\end{figure}
Next, we perform a 
translation so that the $\vx^*$ ($\vu^*$) is the origin of the 
state (input) space in the new coordinate system.
In order to obtain a polynomial dynamics, we approximate 
$v^{-1}$, $\sin$ and $\cos$ with polynomials of degree one, three and three, 
respectively.
These changes yield a polynomial control affine dynamics, which fits the
description of our model. 
For the reach-while-stay specification, the sets are defined as the following:
\begin{align*}
S&:[3, 9]\times[-0.75, 0.75]\times[-0.75, 0.75]\times[-2, 2] \\
I&:\{[v, \gamma, \theta, q]^t |(0.4v)^2 + \gamma^2 + \theta^2 + q^2 < 0.4^2 \} \, \\
T&:\{[v, \gamma, \theta, q]^t |(0.4v)^2 + \gamma^2 + \theta^2 + q^2 < 0.05^2 \} \,.
\end{align*}
The projection of some of the traces of the system in $x$-$y$ plane is shown in Fig.~\ref{fig:forward-sim}. We set $x_0 = y_0 = 0$ and
\[
\dot{x} = v \cos(\gamma) , \ \dot{y} = v \sin(\gamma)\,.
\]
The CLF-based
controller is designed using the following generated CLF:
\begin{align*}
V =&+ 3.23 q^2 + 2.17 q\theta
+ 3.90 \theta^2 - 0.2 qv
- 0.45 v\theta \\
& + 0.53 v^2 + 1.66 q\gamma - 1.33 \gamma\theta
+ 0.48 v\gamma + 3.90 \gamma^2 \,.
\end{align*}
The traces show that the CLF-based controller stabilizes faster, however, the MPC
controller uses the aerodynamics to achieve the same goal with a better performance.
\begin{figure}
\begin{center}
	\includegraphics[width=0.48\textwidth]{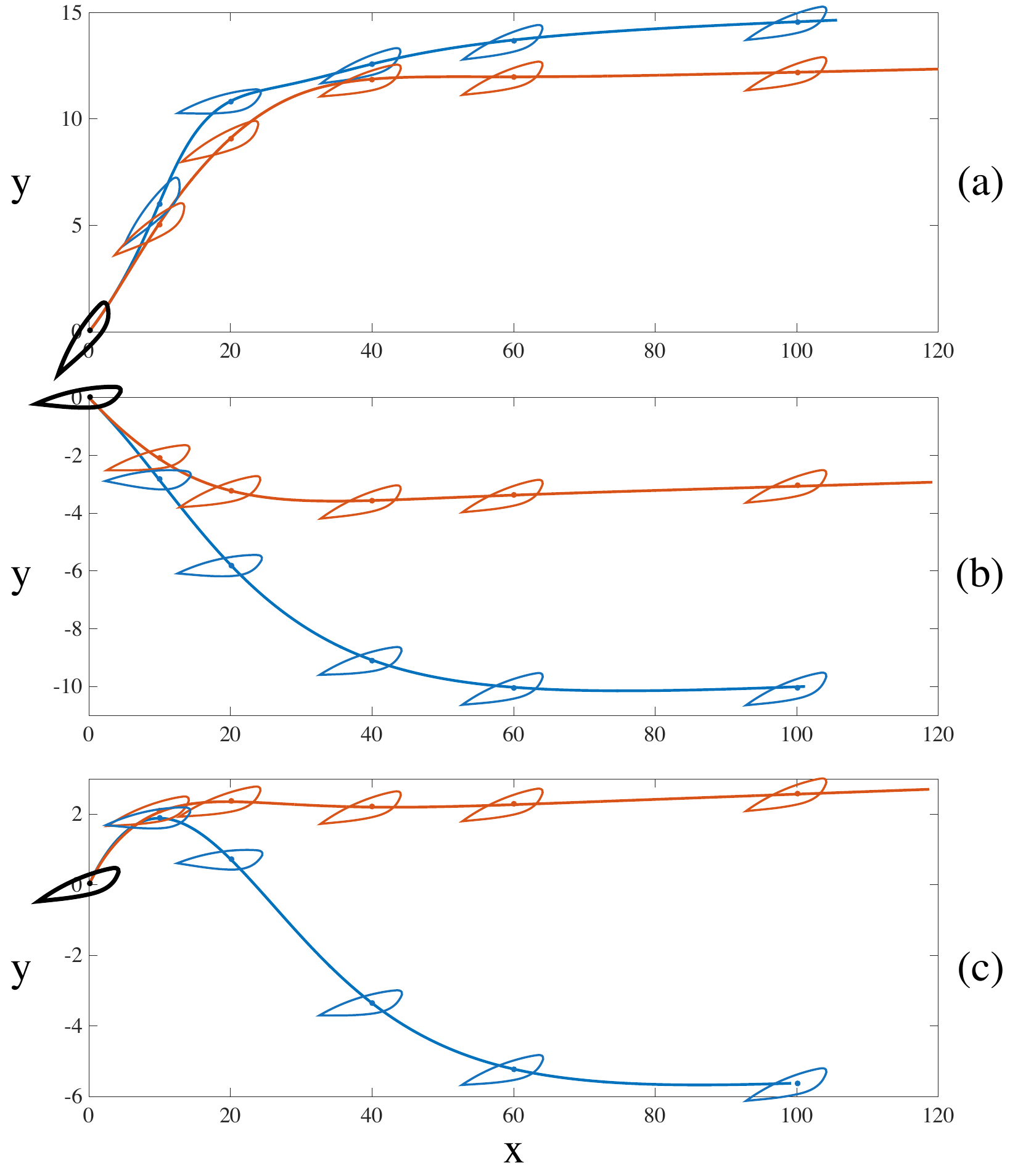}
\end{center}
\caption{Simulation for forward flight of 
Caltech ducted fan - Projected on x-y
 plane.
Blue (red) traces are trajectories of the closed loop system
 with the MPC (CLF-based) controller. The rotational position is shown for some of the states (in black for the initial state) for each trajectory. Initial states 
are $[2, 0.4, 0.717, 0]$, $[-1, -0.25, -0.133, 0]$, and $[-1, 0.4, 0.177, 0]$
for (a), (b), and (c), respectively.}\label{fig:forward-sim} 
\end{figure}

\subsection{Case Study IV:}\label{case:hover}
This case study addresses another problem for the planar Caltech ducted fan~\cite{jadbabaie2002control}.
The goal is to keep the planar ducted fan in a hover mode. The system has three degrees
of freedom, $x$, $y$, and $\theta$, which define the position and orientation of
the ducted fan. There are six state variables $x$, $y$, $\theta$, $\dot{x}$, $\dot{y}$, $\dot{\theta}$ and two control inputs $u_1$, $u_2$ ($U \in [-10, 10]\times[0, 10]$).
The dynamics are
\begin{equation*}\label{ex:ducted-fan-hover-dyn}
	\left[ \begin{array}{l}
		m \ddot{x} \\ m \ddot{y} \\ J \ddot{\theta}
	\end{array}\right] = 
	\left[ \begin{array}{l}
		-d_c\dot{x} + u_1 \cos(\theta) - u_2 \sin(\theta) \\
		-d_c\dot{y} + u_2 \cos(\theta) + u_1 \sin(\theta) - mg \\
		r u_1
	\end{array} \right] \,,
\end{equation*}
where $m = 11.2$, $g = 0.28$, $J = 0.0462$, $r = 0.156$ and $d_c = 0.1$. 
The system is stable at origin for $\vu^*: [0, mg]$. Therefore, we set
$\vu*$ as the origin for the input space.
The specification
is a reach-while-stay property with the following sets:
\begin{align*}
S &: [-1,1]\times[-1,1]\times[-0.7,0.7]\times[-1, 1]^3 \\
I &:\B_{0.25}(\vzero), T:\B_{0.1}(\vzero) \,.
\end{align*}
The trigonometric functions are approximated with degree two
polynomials and the procedure finds a quadratic CLF:
\begin{align*}
V =& 1.64 \dot{\theta}^2 - 0.56 \dot{\theta}\dot{y}
+ 13.53 \dot{y}^2 + 0.07 \dot{\theta}y + 1.15 y\dot{y} +\\
&1.16 y^2 + 1.74 \theta\dot{\theta} + 0.03 \dot{y}\theta - 0.77 y\theta + 4.80 \theta^2 -\\
&4.57 \dot{\theta}\dot{x} + 0.85 \dot{x}\dot{y} + 0.34 y\dot{x} - 8.59 \dot{x}\theta + 12.77 \dot{x}^2 -\\
&0.45 \dot{\theta}x + 0.06 \dot{y}x +  0.51 yx - 3.71 x\theta + 4.12 x\dot{x} + \\
&1.88 x^2 \,.
\end{align*}
Some of the traces are shown in Fig.~\ref{fig:hover-sim}. As the simulation suggest,
the MPC controller behaves very differently and the CLF-based controller
yield solutions with more oscillations. The CLF-based controller first
stabilizes $x$ and $\theta$ and then value of $y$ settles. Also, once 
the trace is inside the target region, the CLF-based controller does 
not guarantee decrease in $V$ as this fact is intuitively visible 
in Fig.~\ref{fig:hover-sim}(c).
\begin{figure}
\begin{center}
	\includegraphics[width=0.45\textwidth]{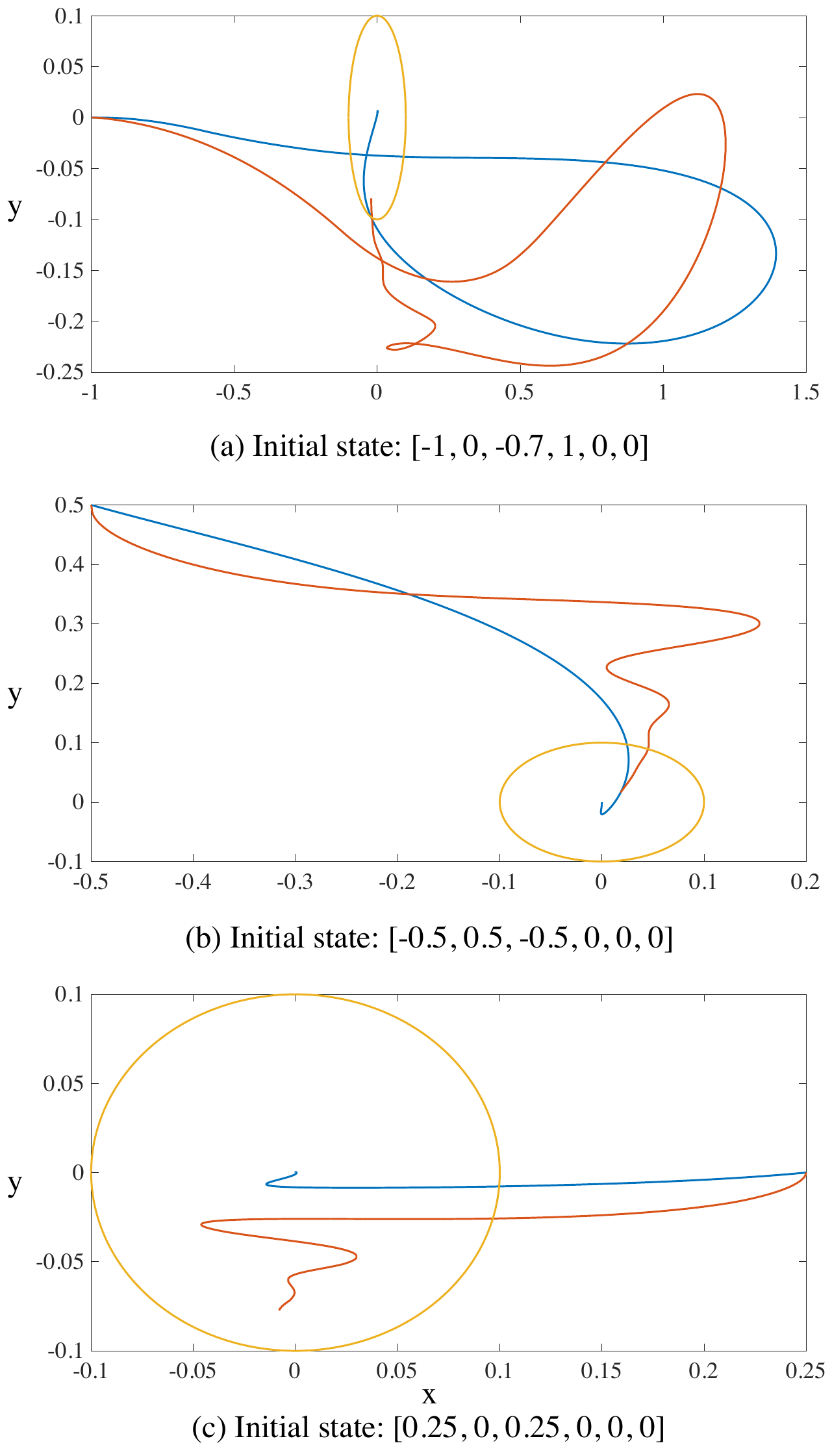}
\end{center}
\caption{Simulation for Case Study IV - Projected on x-y
 plane. The trajectories corresponding to the CLF-based (MPC) controller are shown in red 
(blue) lines. The boundary of the target set is 
shown in yellow.}\label{fig:hover-sim} 
\end{figure}
\subsection{Case Study V:}
In this case study, a unicycle model~\cite{liberzon2012switching} is considered.
It is known that no continuous feedback can stabilize the unicycle, and therefore no
continuous CLF exists. However, considering a reference trajectory for a moving unicycle,
one can keep the system near the reference trajectory, using control funnels.
The unicycle model has the dynamics: 
\[
\dot{x} = u_1 \cos(\theta) \ , \ \dot{y} = u_1 \sin(\theta) \ , \ \dot{\theta} = u_2\,.
\]
By a change of basis, a simpler dynamic model is used here (see.~\cite{liberzon2012switching}):
\[
\dot{x_1} = u_1, \dot{x_2} = u_2,  \dot{x_3} = x_1 u_2 - x_2 u_1 \,.
\]

\begin{figure*}[t]
\begin{center}
	\includegraphics[width=0.95\textwidth]{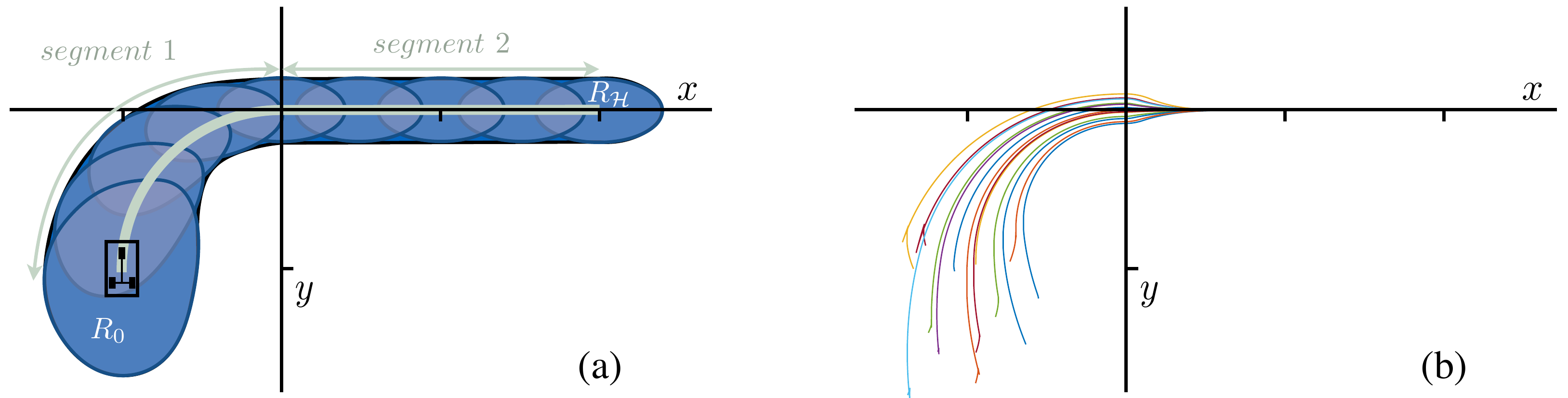}
\end{center}
\caption{(a) Trajectory tracking using control funnel - Projected on x-y
 plane. The reference trajectory is shown with the green line, consists of two segments.
 Starting from $R_0$, the state remains in the funnel (blue region) 
 until it reaches $R_\T$. Boundary of each smaller blue region shows
 the boundary of the funnel for a specific time. (b) Simulation traces for
 some random initial states.
 }\label{fig:unicycle} 
\end{figure*}

We consider a planning problem, in which starting
near $[\theta, x, y] = [\frac{\pi}{2}, -1, -1]$, the goal is to reach near $[\theta, x, y] = [0, 2, 0]$.
In the first step, a feasible trajectory $\vx^*(t)$ is generated as 
shown in Fig.~\ref{fig:unicycle}(a).
Then $\vx^*(t)$ is approximated with piecewise polynomials. 
More precisely, trajectory consists
 of two segments. The first segment brings the car to the
origin and the second segment moves the car to the
destination.
Each segment is approximated using
polynomials in $t$ with degree up to three:
\begin{align*}
&\mbox{seg. 2} : \begin{cases}
	\theta(t)^* = 0 \\
	x^*(t) = t \\
	y^*(t) = 0
\end{cases}  \\ 
&\mbox{seg. 1} : \begin{cases}
	\theta^*(t) = \pi - t \\
	x^*(t) = -(1-0.64t)(1+0.64t) \\
	y^*(t) = -(1-0.64t)(1-0.2t-0.25t^2) \,.
\end{cases}
\end{align*}
Let $Tr(\theta, x, y)$ represent the transformation of the state in terms
  of $(\theta, x, y)$ coordinate system to the $(x_1, x_2, x_3)$ coordinates.
Also, for two set $A$, and $B$, let $A \oplus B$ be the Minkowski sum of $A$ and $B$.
For example, we write $\{Tr(\theta, x, y)\} \oplus \B_{\delta}(\vzero)$ to denote a state
and a ball of radius $\delta$ around it. Moreover, let $S_1$ ($S_2$) be the minimal box
which contains the trajectory $\vx^*(\cdot)$ for the first (second) segment in the 
$(x_1, x_2, x_3)$ coordinates.
For the first segment, the goal is to reach from the initial set
$I:\{Tr(\pi/2, -1, -1)\} \oplus \B_{1}(\vzero)$ to the target set $T:\{Tr(0, 0, 0)\} \oplus \B_{1}(\vzero)$. 
Also, the safe set is defined as $S:S_1 \oplus [-1.5,1.5]^3$.
That is, an enlarged box around $S_1$.
And in the next segment, the goal is to reach from initial set $I:Tr(0, 0, 0) \oplus \B_{1}(\vzero)$
to $T:Tr(0, 2, 0) \oplus \B_{1}(\vzero)$ as the target, while staying in $S:S_2 \oplus [-2, 2]^3$.

For each segment, we search for a Lyapunov-like function $V$ as a time varying function, quadratic in the states. 
Our method is applied to this problem,
and we are able to find a strategy to implement the plan with guarantees.
The boundary of the funnels is shown in Fig.~\ref{fig:unicycle}(a). Also, some
simulation traces are shown in Fig.~\ref{fig:unicycle}(b), where the CLF controller
is implemented using the generated funnels. As simulations suggest, the funnels
can effectively stabilize the traces to the trajectory, when the unicycle
is moving forward.

\subsection{Performance}
As mentioned earlier, the inputs to the learning framework are the
plant, monomial basis functions, and the demonstrator. Also, the degree of 
relaxation $D$ is also considered as input.
At each iteration, first a MVE inscribed inside a polytope is calculated. This
task is performed quite efficiently. The MPC scheme used inside the demonstrator is
an input and we do not consider its performance
here. Nevertheless, MPC is known to be very efficient if it is carefully tuned.
We mention that the MPC parameters used here
are selected by a non-expert and usually the time step is very small and
the horizon is very long. 
Nevertheless, as the MPC is used offline, they are still suitable for
our framework. Also, costs matrices $Q$, $R$, and $H$ are diagonal:
\[
Q = diag(Q') \ , \ R = diag(R') \ , \ H = N diag(Q') \,,
\]
where $Q' \in \reals^n$ and $R' \in \reals^m$.  There are two other
important factors that determines the performance of the whole
learning framework: (i) the time taken by the verifier and
(ii) the number of iterations. Table.~\ref{tab:result} shows the
results of the learning framework for the set of case studies
  described thus far. For each problem instance, the parameters of
the MPC, as well as the degree of relaxation are provided.
Also, the performance of the learning framework is tabulated.
First, the procedure starts from $C : [-\Delta, \Delta]^r$ and terminates 
whenever $\Vol(E_j) < \gamma \delta^r$. We set $\Delta = 100$ and 
$\delta = 10^{-3}$. The results demonstrate that the method terminates in
few iterations, even for the cases where a compatible CLF does not
exists.

Notice that the number of demonstrations is different from the number
of iterations.  Recall that two separate problems are solved for the
verification.  One involves checking the positivity of $V$, and the
other involves checking whether $\nabla V$ can be
decreased. When a counterexample $\vx_j$ is found for the
  former problem, there is no need to check the latter
  condition. Furthermore, we do not require a demonstration for such
  a scenario. This optimization is added to speed up our overall
  procedure by avoiding expensive calls to the MPC.  To accommodate
this, our approach calculates $\hat{C}_{j+1}$ (instead of $C_{j+1}$) for such
counterexamples as:
\begin{equation}\label{eq:approx-C_j+1}
		\hat{C}_{j+1}:\ \hat{C}_j \cap \left\{ \vc \ |\ V_{\vc}(\vx_{j}) > 0 \right\} \,.
	\end{equation}
 Otherwise, if
the counterexample violates conditions on $\nabla V$, then
\begin{equation}\label{eq:approx-C_j+1-simple}
		\hat{C}_{j+1}:\ \hat{C}_j \cap \left\{ \vc \ |\ \begin{array}{c} V_{\vc}(\vx_{j}) > 0 \\
 			\nabla V_\vc.f(\vx_j, \vu_j) < 0 \end{array}
 \right\} \,.
	\end{equation}
However, $\vc_j \not\in \hat{C}_{j+1}$ for both cases and the 
convergence guarantees continue to hold.
As Table.~\ref{tab:result} shows, using this trick, the number 
of demonstrations can be much smaller than the total number of iterations.

At each iteration, several verification problems are solved which
involve solving large SDP problems. While the complexity of solving
SDP is polynomial in the number of variables, they are still hard to
solve. The verification problem is quite expensive when the number of
variables and degree of relaxation are large. Nevertheless, as the SDP
solvers mature further, we believe our method can solve larger problems, since
the verification procedure is currently the computational bottleneck
for the learning framework.  We note that, using larger degree of
relaxation does not necessarily lead to a longer learning process
(e.g. hover flight example).  For example, for the inverted pendulum
example, using degree of relaxation five the procedure finds a CLF
faster when compared to the case wherein the degree of relaxation is
set to four.

 \begin{table*}[t]
\caption{\small Results on the benchmark. $\tau$: MPC time step, $N$: number of horizon steps, $Q'$: defines MPC state cost, $R'$: defines MPC input cost, $D$: SDP relaxation degree bound, \#Dem : number of demonstrations, \#Itr: number of iterations, V. Time: total computation time for verification (minutes), Time: total computation time (minutes)}\label{tab:result} 
\begin{center}
\begin{tabular}{ ||l||c|c|c|c||c||c|c|c|c|c|| } 
 \hline
 \multicolumn{1}{||c||}{Problem} & \multicolumn{4}{c||}{Demonstrator} & Verifier & 
 \multicolumn{5}{c||}{Performance} \\
 \hline
 System Name & $\tau$ & $N$ & $Q'$ & $R'$ & $D$ & \#Dem & \# Itr & V. Time & Time & Status \\
 \hline
 \multirow{2}{*}{Unicycle-Segment 2} & \multirow{2}{*}{0.1} & \multirow{2}{*}{10} & \multirow{2}{*}{[1 1 1]} & \multirow{2}{*}{[1 1]} & 3 & 2 & 74 & 3 & 3 & Fail \\ 
  &  &  & &  & 4 & 2 & 57 & 4 & 4 & Succ\\ 
 \hline
 \multirow{2}{*}{Unicycle-Segment 1} & \multirow{2}{*}{0.1} & \multirow{2}{*}{20} & \multirow{2}{*}{[1 1 1]} & \multirow{2}{*}{[1 1]} & 3 & 27 & 86 & 9 & 10 & Fail \\ 
  &  &  &  & & 4 & 23 & 71 & 11 & 12 & Succ \\ 
 \hline
 \multirow{2}{*}{TORA} & \multirow{2}{*}{1} & \multirow{2}{*}{30} & \multirow{2}{*}{[1 1 1 1]} & \multirow{2}{*}{[1]} & 3 & 52 & 118 & 7 & 14 & Fail \\ 
  &  &  & & & 4 & 19 & 76 & 5 & 8 & Succ \\ 
  \hline
 \multirow{3}{*}{Inverted Pendulum} & \multirow{3}{*}{0.04} & \multirow{3}{*}{50} & \multirow{3}{*}{[10 1 1 1]} & \multirow{3}{*}{[10]} & 3 & 56 & 85 & 7 & 27 & Fail \\ 
  &  &  & & & 4 & 53 & 69 & 9 & 25 & Succ \\ 
  &  &  & & & 5 & 34 & 50 & 7 & 19 & Succ \\ 
  \hline
 \multirow{2}{*}{Bicycle} & \multirow{2}{*}{0.4} & \multirow{2}{*}{20} & \multirow{2}{*}{[1 1 1 1]} & \multirow{2}{*}{[1 1]} & 2 & 14 & 32 & 2 & 2 & Fail \\ 
  &  &  & & & 3 & 7 & 25 & 1 & 1 & Succ \\ 
  \hline
 \multirow{2}{*}{Bicycle $\times$ 2} & \multirow{2}{*}{0.4} & \multirow{2}{*}{20} & \multirow{2}{*}{[1 1 1 1 1 1 1 1]} & \multirow{2}{*}{[1 1 1 1]} & 2 & 119 & 225 & 77 & 90 & Fail \\ 
  &  &  & & & 3 & 30 & 81 & 43 & 46 & Succ\\ 
  \hline
 \multirow{2}{*}{Forward Flight} & \multirow{2}{*}{0.4} & \multirow{2}{*}{40} & \multirow{2}{*}{[1 1 1 1]} & \multirow{2}{*}{[1 1]} & 4 & 14 & 77 & 16 & 18 & Fail \\ 
  &  &  & & & 5 & 4 & 64 & 10 & 10 & Succ\\ 
  \hline
 \multirow{3}{*}{Hover Flight} & \multirow{3}{*}{0.4} & \multirow{3}{*}{40} & \multirow{3}{*}{[1 1 1 1 1 1]} & \multirow{3}{*}{[1 1]} & 2 & 57 & 147 & 12 & 40 & Fail \\ 
  &  &  & & & 3 & 57 & 124 & 21 & 47 & Succ \\ 
  &  &  & & & 4 & 51 & 116 & 30 & 54 & Succ \\ 
 \hline
\end{tabular}
\end{center}
\end{table*}

In previous sections, we discussed that two important factor governs the convergence of
the search process: (i) candidate selection, and (ii) counterexample selection. In order
to study the effect of these processes, we investigate different techniques to evaluate
their performances. For candidate selection, we consider three different methods.
In the first method, a Chebyshev center of $C_j$ is used as a candidate. In the second
method, the analytic center of constraints defining $C_j$ is the selected candidate and
redundant constraints are not dropped. And finally, in the last method, the center of
MVE inscribed in $C_j$ yields the candidate. Also, for each of these methods, we compare
the performance for two different cases: (i) a random counterexample is generated, 
(ii) the generated counterexample maximizes constraint violations (see Sec.~\ref{sec:counterexample-selection}). Table~\ref{tab:selection} shows the performance
for each of these cases, applied to the same set of problems. 
The results demonstrate that selecting good counterexamples would increase the convergence 
rate (fewer iterations). Nevertheless, the time it takes to generate these
counterexamples increases, and therefore, the overall performance degrades. In conclusion, while
generating good counterexamples provides better reduction in the space of candidates, it is computationally expensive, and  thus,
it seems to be beneficial to just rely on candidate selection for fast termination.
Table.~\ref{tab:selection} also suggests that Chebyshev center has the worst performance. 
Also, the MVE-based method performs better (fewer iterations) compared to the method 
which is  based on the analytic center.

\begin{table*}[t]
\caption{\small Results on different variations. I: number of iterations, VT: computation time for verification (minutes), T: total computation time (minutes), Simple CE: any counterexample, Max CE: counterexample with maximum violation}\label{tab:selection}
\begin{center}{\scriptsize
\begin{tabular}{ ||l||rrr|rrr||rrr|rrr||rrr|rrr||} 
 \hline
 \multirow{3}{*}{Problem} & \multicolumn{6}{c||}{Chebyshev Center} & \multicolumn{6}{c||}{Analytic Center} & \multicolumn{6}{c||}{MVE Center} \\
 \cline{2-19}
 & \multicolumn{3}{c|}{Simple CE} & \multicolumn{3}{c||}{Max CE}  & \multicolumn{3}{c|}{Simple CE} & \multicolumn{3}{c||}{Max CE}  & \multicolumn{3}{c|}{Simple CE} & \multicolumn{3}{c||}{Max CE} \\
 \cline{2-19}
  & I & VT & T & I & VT & T & I & VT & T & I & VT & T & I & VT & T & I & VT & T  \\
 \hline
 Unicycle - Seg. 2 & 83 & 4 & 4 & 22 & 9 & 9 & 76 & 5 & 6 & 23 & 9 & 10 & 57 & 4 & 4 & 15 & 6 & 6  \\
 Unicycle - Seg. 1 & 81 & 6 & 7 & 34 & 17 & 17 & 85 & 10 & 10 & 35 & 15 & 16 & 71 & 11 & 12 & 36 & 18 & 18  \\
 TORA & 185 & 7 & 10 & 52 & 12 & 15 & 95 & 5 & 9 & 36 & 9 & 11 & 76 & 5 & 8 & 36 & 12 & 14  \\
 Inverted Pend. & 163 & 10 & 23 & 85 & 22 & 30 & 57 & 8 & 20 & 51 & 22 & 32 & 50 & 7 & 19 & 35 & 18 & 25  \\
 Bicycle & 99 & 3 & 3 & 40 & 5 & 5 & 31 & 2 & 2 & 20 & 3 & 3 & 25 & 1 & 2 & 15 & 3 & 3  \\
 Bicycle $\times$ 2 & 759 & 121 & 127 & 438 & 244 & 246 & 96 & 47 & 50 & 77 & 141 & 143 & 81 & 43 & 46 & 66 & 132 & 133  \\
 Forward Flight & 676 & 20 & 21 & 34 & 30 & 31 & 113 & 15 & 16 & 21 & 18 & 19 & 64 & 10 & 10 & 16 & 16 & 16  \\
 Hover Flight & 499 & 65 & 90 & 196 & 113 & 127 & 146 & 36 & 67 & 90 & 92 & 109 & 116 & 30 & 54 & 75 & 69 & 82  \\
 \hline
 \end{tabular}
} \end{center}
 \end{table*}

\subsection{Comparison with Other Approaches}

We now compare our method against other techniques used to
automatically construct provably correct controllers.

\paragraph{Comparison with CEGIS:}
We have claimed that the use of demonstrator
helps our approach deal with a computationally
expensive quantifier alternation in the CLF condition.
To understand the impact of this aspect of our approach,
we first we compare
the proposed method with our previous work, namely counterexample
guided inductive synthesis
(CEGIS) that is designed to solve constraints with quantifier alternation,
and applied to the synthesis of CLFs~\cite{ravanbakhsh2015counterexample}. In this framework, the
learning process only relies on counterexamples provided by a verifier component, without involving
demonstrations. Despite a timeout that is set to two hours, our
CEGIS method timed out for \emph{all the problem instances} discussed in
this article, without discovering a CLF.  As a result, we exclude this approach from
further comparisons. These results suggest that demonstrations are essential
for fast convergence.

\paragraph{Learning CLFs from Data:} On the other hand, 
Khansari-Zadeh et al.~\cite{KHANSARIZADEH2014} learn likely CLFs from
demonstrations from sets of states that are sampled without (a) the use of a verifier to check,
and  (b) counterexamples as new samples, both of which are features of our approach.
Therefore, the correctness of the controller thus derived is not formally
guaranteed. To this end, we verify if the solution is in fact
a CLF.

The methodology of Khansari-Zadeh et al. is implemented using the following steps:
\begin{compactenum}
\item Choose a  parameterization of the desired CLF $V_{\vc}(\vx)$ (identical to our approach).
\item Generate samples in batches, wherein for each batch:
  \begin{compactenum}
  \item Sample $N_1=50$ states uniformly at random, and for each state $\vx_i$, add the constraint $V_{\vc}(\vx_i) \geq 0 $,
    for $i \in [1, N_1]$.
  \item Sample $N_2=5$ States at random, and for each state $\vx_j$ ($j \in [1, N_2]$),
    simulate the MPC demonstrator for $N_3 = 10$ time steps to obtain state control
    samples
    \[\{(\vx_{j,1}, \vu_{j,1}),\ldots,(\vx_{j,N_3}, \vu_{j,N_3})\} \,.\]
   
  \item Add the constraints $\grad V_{\vc} \cdot f |_{\vx = \vx_{j,k}, \vu=\vu_{j,k}}< 0$ for $j=1, \ldots, N_2$ and
    $k = 1, \ldots, N_3$ to enforce the negative definiteness of the CLF.
  \end{compactenum}
  \item At the end of batch $k$, solve the system of linear constraints thus far to check if there is a feasible solution.
  \item If there is no feasible solution, then \textbf{exit}, since no function in $V_{\vc}(\vx)$ is compatible with the data.
  \item If there is a feasible solution, check this solution using the \textsc{verifier}.
  \item If the verifier succeeds, then \textbf{exit} successfully with the CLF discovered.
  \item Otherwise, continue to generate another batch of samples.
\end{compactenum}
We enforce the constraint $V(\vx) > 0$ and $\grad V \cdot f < 0$  over different sets of samples,
since simulating the demonstrator is much more expensive for each point. 
The approach iterates between  generating successive batches of data until a preset
timeout of two hours as long as (a) there are CLFs remaining to consider and (b) no
CLF has been discovered thus far. The time taken to learn and verify the solution is not considered
against the total time limit, and also not added to the overall time reported.
Besides stability, the approach is also adapted for other properties, which are used in our benchmarks.

  \begin{table*}[t]
\caption{\small Results for ``demonstration-only" method. \#Sam.: number of samples, \#Dem: number of demonstrations, Case: best-case or worst-case, Time: total computation time (minutes), TO: time out ($>$ 2 hours).}\label{tab:no-CE}
\begin{center}
\begin{tabular}{ ||l||r|r||l|r|r|r|c||r|r|r|| } 
 \hline
 \multicolumn{1}{||c||}{Problem} &  \multicolumn{2}{c||}{Stats} &
 \multicolumn{5}{c||}{Performance} & \multicolumn{3}{c||}{Proposed Method} \\
 \hline
System Name & Succ. \% & TO \% & Case & \#Sam. & \#Dem. & Time & Status & \#Sam. & \#Dem. & Time \\
 \hline
 \multirow{2}{*}{Unicycle-Segment 2} & \multirow{2}{*}{60} & \multirow{2}{*}{0}
 & best & 400 & 40 & 1 & Succ & \multirow{2}{*}{65} & \multirow{2}{*}{2} & \multirow{2}{*}{4} \\ 
& & & worst & 600 & 72 & 1 & Fail & & &\\  
 \hline
\multirow{2}{*}{Unicycle-Segment 1} & \multirow{2}{*}{45} & \multirow{2}{*}{0}
 & best & 600 & 35 & 2 & Succ & \multirow{2}{*}{79} & \multirow{2}{*}{23} & \multirow{2}{*}{12} \\ 
& & & worst & 800 & 70 & 3 & Fail & & &\\ 
 \hline
\multirow{2}{*}{TORA} & \multirow{2}{*}{60} & \multirow{2}{*}{30}
 & best & 6300 & 535 & 43 & Succ & \multirow{2}{*}{84} & \multirow{2}{*}{19} & \multirow{2}{*}{8} \\ 
& & & worst & 17100 & 1580 & TO & Fail & & &\\ 
 \hline
\multirow{2}{*}{Inverted Pendulum} & \multirow{2}{*}{30} & \multirow{2}{*}{70}
 & best & 2250 & 137 & 84 & Succ & \multirow{2}{*}{58} & \multirow{2}{*}{34} & \multirow{2}{*}{19} \\ 
& & & worst & 15750 & 300 & TO & Fail & & &\\ 
 \hline
\multirow{2}{*}{Bicycle} & \multirow{2}{*}{100} & \multirow{2}{*}{0}
 & best & 2700 & 55 & 2 & Succ  & \multirow{2}{*}{33} & \multirow{2}{*}{7} & \multirow{2}{*}{1}\\ 
& & & worst & 54000 & 1883 & 48 & Succ & & &\\ 
 \hline
\multirow{2}{*}{Bicycle $\times$ 2} & \multirow{2}{*}{0} & \multirow{2}{*}{100}
 & best & 81600 & 1736 & TO & Fail & \multirow{2}{*}{89} & \multirow{2}{*}{30} & \multirow{2}{*}{46}\\ 
& & & worst & - & - & - & - & & &\\ 
 \hline
\multirow{2}{*}{Forward Flight} & \multirow{2}{*}{0} & \multirow{2}{*}{0}
 & best & 900 & 35 & 4 & Fail  & \multirow{2}{*}{72} & \multirow{2}{*}{4} & \multirow{2}{*}{10}\\ 
& & & worst & 2700 & 254 & 31 & Fail & & &\\ 
 \hline
\multirow{2}{*}{Hover Flight} & \multirow{2}{*}{0} & \multirow{2}{*}{100}
 & best & 7150 & 227 & TO & Fail  & \multirow{2}{*}{132} & \multirow{2}{*}{57} & \multirow{2}{*}{47}\\ 
& & & worst & - & - & - & - & & &\\ 
 \hline
 \end{tabular}
\end{center}
\end{table*}

 The results are reported in Table.~\ref{tab:no-CE}. Since the
 generation of random samples are involved, we run the procedure $10$
 times on each benchmark, and report the percentage of trials that
 succeeded in finding a CLF, the number of timeouts and the number of
 trials that ended in an infeasible set of constraints.  We note that
 the success rate is $100\%$ for just one problem instance. For four
 other problem instances, the method is successful for a fraction
 of the trials. The remaining benchmarks fail on all trials.
 Next, the minimum and maximum number of demonstrations needed in the
 trials to find a CLF is reported as the ``best-case'' and ``worst-case''
 respectively. We note that our approach requires much fewer demonstrations
 even when compared the best case scenario. Thus, we conclude from this
 data that the 
 time spent by our approach for  finding counterexamples is justified by the
 \emph{significant decrease} in the number of demonstrations, and thus, faster
 convergence. This is beneficial especially for cases where generating
 demonstrations is expensive.

 For one of the benchmarks (the forward flight problem of the Caltech ducted fan),
 the method stops for all cases because a function compatible with the data does not exist.
As such, this suggests that no CLF compatible with the demonstrator exists. 
On the other hand, our approach successfully finds a CLF while considering just
four demonstrations. 

Finally, for two of the larger  problem instances, we continue to obtain 
feasible solutions at the end of the time limit, although the verifier
cannot prove the learned function is a CLF. In other words, there are
values of $\vc$ left, that have not been considered by the verifier.
Our approach uses  counterexamples, along with  a judicious choice of candidate CLFs to eliminate
all but a bounded volume of candidates.

\paragraph{Comparison with Bilinear Solvers:} We now compare our method against approaches based on 
bilinear formulations found in related work~\cite{el1994synthesis,majumdar2013control,tan2004searching}.
We wish to find a Lyapunov function $V$ and a corresponding feedback law $K: X \mapsto U$, simultaneously.
Therefore, we assume $K$ is a linear combination of basis functions $K: \sum_{k=1}^{r'} \theta_{k} h_k(\vx)$.
Likewise, we parameterize $V$ as a linear combination of basis functions, as well: $V : \sum_{k=1}^r c_{k} g_k(\vx)$.
Then, we wish to find $\vc$ and $\vth$ that satisfy the constraints corresponding to the
property at hand. 
To synthesize a CLF, we wish to find $V_{\vc}, K_{\vth}$, so that $V_{\vc}(\vx)$ and its
Lie derivative under the feedback $u = K_{\vth}(\vx)$ is negative definite. This is relaxed
as an optimization problem:
\[ \begin{array}{rcl}
     \min\limits_{\vc,\vth,\gamma} \textcolor{red}{\gamma} &\  \\
        \mathsf{s.t.}
        & V_{\vc} \mbox{ is positive definite } \\
          & (\forall\ \vx \neq \vzero)\ \grad V_{\vc}(\vx) \cdot f(\vx, K_{\vth}(\vx)) \leq  {\color{red}\gamma} ||\vx||_2^2 \\
   \end{array}\]

The decision variables include $\vc, \vth$ that parameterize $V$ and $K$, respectively.
In fact, if a feasible solution  is obtained such that $\gamma < 0$ then we may stop the optimization and
declare that a CLF has been found. 
 To solve this bilinear problem, we use alternative minimization approach described below. 
First, $V$ is initialized to be a positive definite function (by initializing $\vc$ to some fixed value).
Then,  the approach repeatedly alternates between the following steps:
\begin{enumerate}
	\item $\vc$ is fixed, and we search for a $\vth$ that minimizes $\gamma$.
	\item $\vth$ is fixed, and we search for a $\vc$ that minimizes $\gamma$.
\end{enumerate}
Each of these problems can be relaxed using Sum of Squares (SOS)
programming~\cite{prajna2002introducing}.  The approach is iterated and
results in a sequence of values $\gamma_0 \geq \gamma_1 \geq \gamma_2 \geq \cdots \geq \gamma_i$,
wherein $\gamma_i$ is the value of the objective after $i$ optimization instances have
been solved. Since the solution of one optimization instance forms a feasible solution for the
subsequent instance, it follows that  $\gamma_i$ are monotonically nondecreasing. The iterations stop whenever
$\gamma$ does not decrease sufficiently between iterations.  After termination,
the approach succeeds in finding $V_{\vc}$, $K_{\vth}$ 
only if $\gamma < 0$. Otherwise the approach fails.

Finding a suitable initial value for $\vc$ is an important
factor for success. As proposed by Majumdar et al, we pose and solve a
linear feedback controller by applying the LQR method to the
linearization of the dynamics~\cite{majumdar2013control}. In this
case, we initialize $V$ using the optimal cost function provided by
the LQR.  We also note that the linearization for the dynamics is not
controllable for all cases and we can not always use this
initialization trick.

Additionally, Majumdar et al. (ibid) discuss solutions to handle input saturation, requiring
$K_{\vth}(\vx) \in U$ to avoid input saturation.
Here, we consider two different variations of this method:
(i) inputs are not saturated, (ii) inputs are saturated. We consider variation (ii) only if the method is successful without forcing the input saturation.
For the Lyapunov function $V$ we consider quadratic monomials as our basis functions, and for the feedback law
$K$, we consider both linear and quadratic basis functions as separate problem instances.
Similar to the SDP relaxation considered in this work,
the SOS programming approach uses a  degree limit $D$ for the multiplier polynomials used in the
positivstellensatz (cf.~\cite{lasserre2009moments}).  The limits used for the bilinear optimization
approach are identical to those used in our method for each benchmark. The bilinear method is adapted to other properties used in our
  benchmarks and the results are shown in Table~\ref{tab:bilinear}.

For the first two problem instances, the linearized dynamics are not
controllable, and thus, we can not use the LQR trick for
initialization. 
For the remaining instances, we were able to use the LQR trick successfully
to find an initial solution. Starting from this solution, the bilinear approach
is successful on four problem instances, but fails for the hover flight problem. This
suggests that even the LQR trick may not always provide a good
initialization.  For two of the larger problem instances,
the bilinear method fails because of numerical errors, when dealing
with large SDP problems. While the SOS programming has similar
complexity compared to our method, it encounters numerical problems
when solving large problems. We believe two factors are important
here. First, our method solves different smaller verification problems
and verifies each condition separately, while in a SOS formulation all
conditions on $V$ and $\nabla V$ are formulated into one big SDP
problem. Moreover, in our method when we encounter a numerical error,
we simply use the (potentially wrong) solution as a spurious
counterexample without losing the soundness. Then, using
demonstrations we continue the search. On the other hand, when the
bilinear optimization procedure encounters a numerical error, it
is unable to make further progress towards an optimal solution.

\begin{table}[t]
\caption{\small Results for ``bilinear formulation" method. $K$: basis functions used to parameterize $K$, L: basis functions are monomials with maximum degree $1$ (linear), Q: basis functions are monomials with maximum degree $2$ (quadratic), LQR: if LQR is used for initialization, ST.: saturation type, NP: numerical problem, St.: status.}\label{tab:bilinear}
\begin{center}\scriptsize
\begin{tabular}{ ||l||c|c||c|c|| } 
 \hline
 \multicolumn{1}{||c||}{Problem} &  \multicolumn{2}{|c||}{Param.} & 
 \multicolumn{2}{c||}{Status} \\
 \hline
 System Name & $K$ & LQR & ST.(i) & ST.(ii) \\ 
 \hline
\multirow{1}{*}{Unicycle-Seg. 2} & L & \crossMark & - & - \\
 \hline
\multirow{1}{*}{Unicycle-Seg. 1} & L & \crossMark & - & - \\
 \hline
TORA & L & \tick & \tick & \tick \\
 \hline
Inverted Pend. & L & \tick & \tick & \tick \\
 \hline
Bicycle & L & \tick & \tick & \tick \\
\hline
Bicycle $\times$ 2 & L & \tick & NP & - \\
 \hline
Forward Flight & L & \tick & NP & - \\
 \hline
\multirow{2}{*}{Hover Flight} & L & \tick & \crossMark & - \\
 & Q & \tick & \crossMark & - \\
 \hline
 \end{tabular}
\end{center}
\end{table}

In conclusion, our method has several benefits when compared to the
bilinear formulation. First, our method does not assume the linearized system is controllable to initialize a solution.
Second, our method uses demonstrations to generate a
candidate instead of a local search, and we provide an upper-bound on
the number of iterations. And finally, our method can sometimes recover from
numerically ill-posed SDPs, and thus scales better as demonstrated through
experiments. On the flip side, unlike the bilinear formulation, our method
relies on a demonstrator that may not be easy to implement.

\section{Related Work}\label{sec:related}
In this section, we review the related work from the robotics,
 control, and formal verification communities.

\paragraph{Synthesis of Lyapunov Functions \ from Data:}The problem of
synthesizing Lyapunov functions for a control system by observing the
states of the system in simulation has been investigated in the past
by Topcu et al. to learn Lyapunov functions along with the resulting
basin of attraction~\cite{topcu2007stability}. Whereas the original
problem is bilinear, the use of simulation data makes it easier to
postulate states that belong to the region of attraction, and
therefore find Lyapunov functions that belong to this region by
solving LMIs in each case. The application of this idea to larger
black-box systems is demonstrated by Kapinski et
al.~\cite{kapinski2014simulation}, where the counterexamples are used
to generate data iteratively.  Our approach focuses on controller
synthesis through learning a control Lyapunov function to replace an
existing controller. A key difference lies in the fact that \emph{we
do not attempt to prove that the original demonstrator is necessarily
correct}, but find a control Lyapunov function by assuming that the
demonstrator is able to stabilize the system for the specific states
that we query on. Another important contribution lies in our analysis
of the convergence of the learning with a bound on the maximum number
of queries needed. In fact, these results can also be applied to the
Lyapunov function synthesis approaches mentioned earlier.  Similar to
our work, Khansari-Zadeh et al.~\cite{KHANSARIZADEH2014} uses human
demonstrations to generate data and enforce CLF conditions for the
data points, to learn a CLF candidate. Their work does not include a
verifier and therefore, the CLF candidate may not, in fact, be a CLF.
However, the method can handle errors in the demonstrations by finding
a maximal set of observations for which a compatible CLF exists,
whereas our method does not address erroneous demonstrations.

\paragraph{Counter-Example Guided Inductive Synthesis:} Our
approach of alternating between a learning module that proposes a
candidate and a verification module that checks the proposed candidate
is identical to the counter-example guided inductive synthesis (CEGIS)
framework originally proposed in verification community by
Solar-Lezama et al.~\cite{solar2006combinatorial,solar2008program}.
As such, the CEGIS approach does not include a demonstrator that can
be queried. The extension of this approach Oracle-guided inductive
synthesis~\cite{jha2010oracle}, generalizes CEGIS using an
input/output \emph{oracle} that serves a similar role as a
demonstrator in this paper. However, the goal here is not to mimic the
demonstrator, but to satisfy the specifications.  Also, Jha et
al.~\cite{Jha2017} prove bounds on the number of queries for discrete
concept classes using results on exact concept learning in discrete
spaces~\cite{GOLDMAN199520}. In this article, we consider searching
over continuous concept class, and prove bounds on the number of
queries under a robustness assumption.

The CEGIS procedure has been used for the synthesis of CLFs recently
by
authors~\cite{Ravanbakhsh-Others/2015/Counter-LMI,Ravanbakhsh-Others/2016/Robust},
combining it with SDP solvers for verifying CLFs. The key difference
here lies in the use of the demonstrator module that simplifies the
learning module. In the absence of a demonstrator module, the problem
of finding a candidate reduces to solving linear constraints with
disjunctions, an NP-hard
problem~\cite{Ravanbakhsh-Others/2015/Counter-LMI}. Likewise, the
convergence results are quite
weak~\cite{ravanbakhsh2015counterexample}. In the setting of this
paper, however, the use of a MPC scheme as a demonstrator allows us to
use faster LP solvers and provide convergence guarantees.
Empirically, we are able to demonstrate the successful inference of CLFs
on systems with up to eight state variables, whereas previous work in
this space has been restricted to much smaller
problems~\cite{Ravanbakhsh-Others/2015/Counter-LMI}.

\paragraph{Learning from Demonstration:} The idea of learning from
demonstrations has a long history~\cite{ARGALL2009469}. The overall framework uses a
demonstrator that can, in fact, be a human
operator~\cite{KHANSARIZADEH2014,Khansari-Zadeh2017} or a complex
MPC-based control
law~\cite{stolle2006policies,atkeson2013trajectory,ross2011reduction,zhong2013value,Mordatch-RSS-14,zhang2016learning}.
The approaches differ on the nature of the interactions between the
learner and the demonstrator; as well as how the policy is inferred.
Our approach stands out in many ways: (a) We represent our policies by
CLFs which are polynomial. On one hand, these are much less powerful
than approaches that use neural networks~\cite{zhang2016learning}, for
instance. However, the advantage lies in our ability to solve
verification problems to ensure that the resulting policy learned
through the CLF is correct with respect to the underlying dynamical
model. (b) Our framework is \emph{adversarial}. The choice of the
counterexample to query the demonstrator comes from a failed attempt
to validate the current candidate. (c) Finally, we use simple yet
powerful ideas from convex optimization to place bounds on the number
of queries, paralleling some results on concept learning in discrete
spaces~\cite{GOLDMAN199520}.

\paragraph{Lyapunov Analysis for Controller Synthesis}
Sontag originally introduced Control Lyapunov functions
and provided a universal construction of a feedback law 
for a given CLF~\cite{sontag1983lyapunov,sontag1989universal}. As such,
the problem of learning CLFs is well known to be hard, involving
bilinear matrix inequalities (BMIs)~\cite{tan2004searching}.  An more
conservative (less precise) approach involves solving bilinear
problems simultaneously for a control law and a Lyapunov function
certifying it~\cite{el1994synthesis,majumdar2013control}. This also
leads to bilinear formulation. Prieur et al.~\cite{prieur1999uniting}
shows that the set of feasible solutions to such problem may not only
be non-convex, but also disconnected.  Nevertheless, there are some
attempts to solve these BMIs which are well known to be
NP-hard~\cite{henrion2005solving}. A common approach to solve these
BMIs is to perform an alternating minimization by fixing one set of
bilinear variables while minimizing over the other. Such an approach has
poor guarantees in practice, often ``getting stuck'' on a saddle point
that does not allow the technique to make progress in finding a
feasible solution~\cite{Helton+Merino/1997/Coordinate}. To combat
this, Majumdar et al. (ibid) use LQR controllers and their associated
Lyapunov functions for the linearization of the dynamics as good
initial seed solutions~\cite{majumdar2013control}. In contrast, our
approach simply assumes a demonstrator in the form of a MPC controller
that can be used to resolve the bilinearity. Furthermore, our approach
does not encounter the local saddle point problem. And finally, when
the inputs are saturated, the complexity of such a method is exponential
in the number of control inputs, while the complexity of our method
remains polynomial.

\paragraph{Formal Controller Synthesis}
The use of the learning framework with a demonstrator distinguishes
the approach in this paper from recently developed ideas based on
formal synthesis.  Majority of these techniques focus on a given
dynamical system and a specification of the correctness in temporal
logic to solve the problem of controller design to ensure that the
resulting trajectories of the closed loop satisfy the temporal
specifications.  Most of these approaches are based on
discretization of the state-space into cells to compute a discrete
abstraction of the overall
system~\cite{wongpiromsarn2011tulip,liu2013synthesis,rungger2016scots,mouelhi2013cosyma,kloetzer2008fully}.
Another set of solutions are based on formal parameter synthesis that
search for unknown parameters so that the specifications are
met~\cite{yordanov2008parameter,donze2009parameter}.  These methods
include  synthesize certificates (Lyapunov-like
functions) by solving nonlinear constraints either through
branch-and-bound
techniques~\cite{huang2015controller,ravanbakhsh2015counterexample},
or  through a combination of simulations and quantifier
elimination~\cite{taly2011synthesizing,taly2010switching}. Our method
is potentially more scalable, since the use of a demonstrator allows us to
solve convex constraints instead.  Raman et al. design a
model-predictive control (MPC) from temporal logic
properties~\cite{raman2015reactive}.  More specifically, MILP solvers
are used inside the MPC, which can be quite expensive for real-time
control applications. We instead learn a CLF from the MPC and the CLF
yields an easily computable feedback law (using Sontag's formula).

\paragraph{Occupation Measures}
In this paper, we use the Lyapunov function approach to synthesizing
controllers. An alternative is to use occupation
measures~\cite{rantzer2001dual,prajna2004nonlinear,lasserre2008nonlinear,majumdar2014convex}.
These methods formulate an infinite dimensional problem to maximize
the region of attraction and obtain a corresponding control law. This
is relaxed to a sequence of finite dimensional
SDPs~\cite{lasserre2001global}. Note however that the approach
computes an over approximation of the finite time backward reachable
set from the target and a corresponding control. Our framework here
instead seeks an under-approximation that yields a guaranteed
controller.

\paragraph{Modeling Inaccuracies and Safe Iterative Learning.} A key drawback of our approach is
its dependence on a mathematical model of the system for learning CLFs.
Although this model is by no means identical to the real system,
it is hoped that the CLF and the control law remain valid despite the
unmodeled dynamics. Our recent
work has successfully investigated physical experiments that use control
Lyapunov-like functions learned from mathematical models for path following problems on a
$\frac{1}{8}$-scale model vehicle using accurate indoor localization
to obtain full state information in real-time~\cite{Ravanbakhsh+Others/2018/Path}. The broader area of iterative learning controls considers the process
of learning how to control a given plant at the same time as
inferring a more refined model of the plant through exploration~\cite{French+Roghers/2000/Nonlinear}.
However, in order to avoid damaging the system, it is necessary to maintain
the system state in a safe set while learning the system dynamics. Recent work by Wang et al. consider a combination
of barrier certificates for maintaining safety while learning Gaussian
process models of the vehicle dynamics~\cite{Wang+Others/2017/Safe}. 
Another approach considers safe reinforcement learning that incrementally
refines a Gaussian process approximation of the unmodeled system dynamics, starting from a known initial model~\cite{Berkenkamp+Others/2017/Safe}. This approach uses a Lyapunov function and performs explorations at so-called ``safe points'' from which safety can be guaranteed during the exploration process. In doing so, the model of the system is updated along with an estimate of the safe set obtained as a region of attraction of the Lyapunov function.

\section{Discussion and Future Work}\label{sec:disscussion}
In this section, we discuss some current limitations of our approach
as well as possible extensions of our approach that can provide
avenues for future research.

\paragraph{Extension to Switched Systems:}
Thus far, our focus has been on control affine systems. We note that a
variation of our framework is applicable to switched systems.
Specifically, one can transform a plant wherein the control is
performed through switching between different modes into a problem
over control affine systems. Let $Q$ be a finite set of modes, such that
the  dynamics vary according the mode $q \in Q$
($\dot{\vx} = f_q(\vx)$). The controller is assumed to operate by selecting
the current mode $q$ of the plant. Then the condition on $\nabla V$ for
stabilizing switched systems:
\[
(\forall \vx \neq \vzero) \ (\exists q \in Q) \ \nabla V \cdot f_q(\vx) < 0 \,,
\]
is replaced with
\begin{align*}
(\forall \vx \neq \vzero) \ (\exists \vlam \geq \vzero, \sum_q \vlam_q = 1) \ \sum_{q} \vlam_q \left( \nabla V \cdot f_q(\vx) \right) < 0.
\end{align*}
This is identical to the conditions obtained for a control affine system, and
thus, our framework can readily extend to such systems.
Moreover, using the
original formulation, checking conditions on $\nabla V$ is even simpler (compared to Eq.~\eqref{eq:decrease-cond-init}):
\[
(\exists \vx \neq \vzero) \ \bigwedge_q \nabla V \cdot f_q(\vx) \geq 0 \,.
\]

\paragraph{Extensions to Discrete-Time Systems:} Control
  problems on discrete-time systems have been widely studied. MPC
  schemes are naturally implemented over such systems, and
  furthermore, Lyapunov-like conditions extend quite naturally. As
  such, our approach can be extended to discrete-time nonlinear
  systems defined by maps as opposed to ODEs. However, polynomial
  discrete systems are known to pose computational challenges: when
  the Lie derivative is replaced by a difference operator, the degree
  of the resulting polynomial can be larger.

\paragraph{Optimizing Performance Criteria:} Our approach stops as soon as
one CLF is discovered. However, no claims are made as to the optimality
of the CLF. The experimental results suggest that the controllers found
by the CLFs are quite different from the original demonstrator in terms
of their performance. An important extension to our work lies in
finding CLFs so that the resulting controllers optimize some performance
metric. One challenge lies in specifying these performance metrics as
functions of the coefficients of the CLF. A simple approach may consist
of using a black-box performance evaluation function over the CLF
discovered by our approach. Once a CLF is found, we may continue our
search but now target CLFs whose performance are strictly better than
the ones discovered thus far.


\paragraph{Other Verifiers:} The verifier is the main bottleneck in
our learning framework. While in theory, the SDP relaxation addresses
verification problems for polynomial system, the scalability for
systems of high dimensions is still an issue. There are alternative
solutions to the SDP relaxation, which promise better scalability. In
particular linear relaxations are more attractive for this
framework~\cite{ahmadi2014dsos,bensassi2015linear}. Using linear
relaxations, one could restrict the candidate space to positive
definite polynomials up front, and consider only the conditions over
$\nabla V$ during the verification process. Therefore, using linear
relaxations, not only the verification problem scales better, the
number of such verifications to be solved can be decreased.

For a highly nonlinear system, the degree of polynomials for the
dynamics as well as basis functions get larger. For these systems, the
scalability is even more challenging. In future we wish to explore the
the use of falsifiers (instead of verifiers) and move towards more
scalable
solutions~\cite{Abbas+Others/2013/Probabilistic,AnnapureddyLFS11tacas,Donze+Maler/2010/Robust}. While
falsifiers would not guarantee correctness, they can be used to find
concrete counterexamples.  And by dropping formal correctness, a
falsifier can replace the verifier in the learning framework.

\paragraph{Beyond Polynomial CLFs:} In this paper, we assumed that the CLF
candidate $V$ is a linear combination of some given basis
functions. While we showed that this model is precise enough to
address exponential stability over compact sets, there are systems for
which a smooth $V$ does not exist. Nevertheless, our framework can
also handle nonlinear models such as Gaussian mixture or feed forward
neural network models, especially if the verifier is replaced
  by a falsifier that can be implemented through simulations.
However, there are some serious drawbacks, including more expensive
candidate generation, and weaker convergence
guarantees. In future work we wish to investigate these models.

\paragraph{Beyond MPC-based Demonstrations:}
As mentioned earlier, the demonstrator is treated as a black-box. We
have investigated to use MPC as they are easy to design, and can
provide smooth feedbacks which in our experiments is the key to find a
smooth CLF. However, nonlinear MPC schemes using numerical
  optimization can guarantee convergence only to local minima, but
  this does not translate as such into guarantees of stability or that
  the original specifications are met.  However, if we employed human
demonstrators (for example, an expert who operates the system), the
demonstrator may include errors, and we may need to consider
approaches that can reject a subset of the given
demonstrations~\cite{KHANSARIZADEH2014}. In addition, the
demonstrations can lead to inconsistent data, wherein nearby
queries are handled using different strategies by the demonstrator,
leading to no single CLF that is compatible with the given demonstrations~\cite{chernova2008learning,BREAZEAL2006385}.
These problems are left for future work.

\section{Conclusion}
We have thus proposed an algorithmic learning framework for
synthesizing control Lyapunov-like functions for a variety of
properties including stability, reach-while-stay. 
The framework provides theoretical guarantees of soundness,
i.e., the synthesized controller is guaranteed to be correct by
construction against the given plant model.  Furthermore, our approach
uses ideas from convex analysis to provide termination guarantees and
bounds on the number of iterations.

\begin{acknowledgements}
  We are grateful to Mr. Sina Aghli, Mr. Souradeep Dutta,
  Prof. Christoffer Heckman and Prof. Eduardo Sontag for helpful
  discussions.  This work was funded in part by NSF under award
  numbers SHF 1527075 and CPS 1646556. All opinions expressed are
  those of the authors and not necessarily of the NSF.
\end{acknowledgements}

\bibliographystyle{spmpsci}      
\bibliography{ref}   


\end{document}